\input ./style/arxiv-general.cfg
\documentclass[aap,MSNbibl,dvips]{arximspdf}
\makeatletter
   \@ifpackageloaded{graphicx}{}{\usepackage{graphicx}}
\makeatother

\doi{10.1214/14-AAP1060} 
\volume{25}
\issue{5}
\pubyear{2015}
\firstpage{2743}
\lastpage{2808}
\docsubty{FLA}

\makeatletter

\newcommand{\rrvert}{\vert}

\newcommand{\llvert}{\vert}
\newcommand{\eqref}[1]{(\ref{#1})}
\def\mod{\mathrm{mod}}
\newcommand{\bigtimes}{\mathop{\mbox{\fontsize{17}{17}\selectfont{$\!\times$}}}}

\def\integers{{\mathbb Z}}

\def\cal{\mathcal}

\newtheorem{claim}{Claim}[section]
\newtheorem{lemma}[claim]{Lemma}
\newproclaim{fact}[claim]{Fact}
\newproclaim{assumption}[claim]{Assumption}
\newtheorem{theorem}{Theorem}

\newproclaim{definition}[claim]{Definition}
\newproclaim{remark}[claim]{Remark}

\def
\prob{{\mathbb P}}
\def\E{{\mathbb E}} 
\def\const{C}
\def\Po{\operatorname{Poisson}}
\def\Binom{\operatorname{Binom}}
\def\nbr{\partial}
\def\dG{d_G}
\def\Ev{\mathsf{E}}
\def\BPzero{{$\mathrm{BP}_0$}}
\def\BPstar{{$\mathrm{BP}_*$}}
\def\F{\mathsf{F}}

\def\naturals{\mathbb{N}}
\def\reals{\mathbb{R}}

\def\T{{\cal T}}
\def\tT{\widetilde{\cal T}}
\def\U{\mathsf{U}}
\def\root{\varnothing}
\def\O{{\cal O}}

\newcommand{\bH}{\mathbb{H}}
\newcommand{\bQ}{{\mathbb{Q}}}

\newcommand{\ad}{\alpha_{\mathrm{ d}}}
\newcommand{\as}{\alpha_{\mathrm{ s}}}

\newcommand{\peel}{{\mathsf{J}}}
\newcommand{\colp}{{\mathsf{T}}}

\newcommand{\GC}{G_{\mathtt{C}}}
\newcommand{\FC}{F_{\mathtt{C}}}
\newcommand{\VC}{V_{\mathtt{C}}}
\newcommand{\EC}{E_{\mathtt{C}}}
\newcommand{\TC}{T_{\mathtt{C}}}
\newcommand{\NC}{N_{\mathtt{C}}}
\newcommand{\uxC}{{\underline{x}_{\mathtt{C}}}}

\newcommand{\core}{\mathtt{C}}

\newcommand{\hVC}{\widehat{V}_{\mathtt{C}}}

\newcommand{\GP}{G_{\mathtt{P}}}
\newcommand{\FP}{F_{\mathtt{P}}}
\newcommand{\VP}{V_{\mathtt{P}}}
\newcommand{\VPstar}{V_{\mathtt{P,*}}}
\newcommand{\EP}{E_{\mathtt{P}}}
\newcommand{\GB}{G_{\mathtt{B}}}
\newcommand{\FB}{F_{\mathtt{B}}}
\newcommand{\VB}{V_{\mathtt{B}}}
\newcommand{\EB}{E_{\mathtt{B}}}
\newcommand{\GNC}{G_{\mathtt{NC}}}
\newcommand{\FNC}{F_{\mathtt{NC}}}
\newcommand{\VNC}{V_{\mathtt{NC}}}
\newcommand{\ENC}{E_{\mathtt{NC}}}
\newcommand{\Nl}{N_{\mathtt{ l}}}

\newcommand{\ECNC}{E_{\mathtt{C,NC}}}

\def\HP{{\mathbb H}_{\mathtt{P}}}
\def\mP{m_{\mathtt{P}}}
\def\nP{n_{\mathtt{P}}}
\def\RP{R^{\mathtt{P}}}
\def\aP{\alpha_{\mathtt{P}}}

\def\Bu{B_{\mathrm{ u}}}
\def\Bl{B_{\mathrm{ l}}}

\def\mC{m_{\mathtt{C}}}
\def\aC{\alpha_{\mathtt{C}}}
\def\etaC{\eta_{\mathtt{C}}}

\def\mNC{m_{\mathtt{NC}}}
\def\nNC{n_{\mathtt{NC}}}
\def\RNC{R^{\mathtt{NC}}}

\def\mb{m_{\mathrm{b}}}
\def\nb{n_{\mathrm{b}}}

\newcommand{\D}{\mathbb{D}}

\newcommand{\polylog}{\operatorname{polylog}}
\newcommand{\coeff}{\operatorname{coeff}}
\newcommand{\spn}{\operatorname{span}}

\def\cS{{\cal S}}
\def\cN{{\cal N}}

\newcommand\Sci[1]{{\cal S}_{\mathtt{C}, #1}}
\def\cC{{\cal C}}
\def\Inst{{\cal I}}
\def\G{{\mathbb G}}
\def\C{{\mathbb C}}
\def\ind{{\mathbb I}}
\def\ux{\underline{x}}
\def\GF{\mathsf{ G}{\mathbb F}}
\def\cG{{\cal G}}
\def\cV{{\cal V}}
\def\cE{{\cal E}}
\def\cH{{\cal H}}
\def\ve{\varepsilon}
\def\cut{\mathrm{cut}}
\def\deg{\mathrm{deg}}
\def\uh{\underline{h}}
\def\ub{\underline{b}}
\def\uc{\underline{c}}
\def\tmu{\tilde{\mu}}
\def\bm{\underline{m}}
\def\di{{\partial i}}
\def\droot{{\partial\o}}
\def\da{{\partial a}}
\def\rank{\operatorname{rank}}
\def\TV{\mathrm{TV}}

\def\cR{\mathcal{R}}
\def\cP{\mathcal{P}}
\def\bI{\mathbb{I}}
\def\bD{{\mathbb D}}
\newcommand\nuia[2]{\nu_{#1 \rightarrow #2}}
\newcommand\nuai[2]{\hat{\nu}_{#1 \rightarrow #2}}

\def\bK{{\mathbb K}}
\def\Ball{\mathsf{ B}}
\def\Hcore{{\mathbb H}_{\mathtt{C}}}
\def\Score{{\cal S}_{\mathtt{C}}}
\def\Lcore{{\cal L}_{\mathtt{C}}}
\def\ba{\bar{\alpha}}
\def\bR{\bar{R}}

\newcommand{\ul}[1]{\underline{#1}}
\def\hat{\widehat}
\def\eps{{\varepsilon}}
\def\wh{\widehat}
\def\hz{\hat{z}}
\def\hQ{\hat{Q}}
\def\N{{\cal{N}}}
\newcommand\term{{\cal{T}}}

\newcommand{\vtoc}[2]{\nu_{{#1}\rightarrow{#2}}}
\newcommand{\ctov}[2]{\wh{\nu}_{{#1}\rightarrow{#2}}}
\def\vtcv{\underline{\nu}}
\def\ctvv{\underline{\widehat{\nu}}}

\def\ed{\stackrel{{\mathrm{ d}}}{=}}
\def\K{{\cal K}}
\def\ve{\varepsilon}

\def\de{\mathrm{ d}}
\def\tR{\widetilde{R}}
\def\tm{\tilde{m}}
\def\tn{\tilde{n}}

\makeatother

\begin{document}
\begin{frontmatter}

\title{The set of solutions of random XORSAT formulae\thanksref{T1}}
\runtitle{The set of solutions of random XORSAT formulae}
\thankstext{T1}{Supported by NSF
Grants CCF-0743978, CCF-0915145, DMS-08-06211 and a Terman fellowship.}

\begin{aug}
\author[A]{\fnms{Morteza}~\snm{Ibrahimi}\ead[label=e1]{ibrahimi@stanford.edu}},
\author[B]{\fnms{Yash}~\snm{Kanoria}\corref{}\thanksref{T2}\ead[label=e2]{ykanoria@columbia.edu}},
\author[C]{\fnms{Matt}~\snm{Kraning}\ead[label=e3]{matt@qadium.com}}
\and
\author[D]{\fnms{Andrea}~\snm{Montanari}\ead[label=e4]{montanari@stanford.edu}}\vspace*{-3pt}
\thankstext{T2}{Supported by a 3Com Corporation Stanford Graduate
Fellowship.}
\runauthor{Ibrahimi, Kanoria, Kraning and Montanari}
\affiliation{Urban Engines, Columbia Business School,
Qadium, Inc. and Stanford University}
\address[A]{M. Ibrahimi\\
Urban Engines\\
Los Altos, California 94022\\
USA\\
\printead{e1}} 
\address[B]{Y. Kanoria\\
Decision, Risk and Operations Division\\
Graduate School of Business\\
Columbia University\\
New York, New York 10027\\
USA\\
\printead{e2}}
\address[C]{M. Kraning\\
Qadium, Inc.\\
San Francisco, California 94107\\
USA\\
\printead{e3}}
\address[D]{A. Montanari\\
Department of Electrical Engineering\\
\quad and Department of Statistics\\
Stanford University\\
Stanford, California 94305\\
USA\\
\printead{e4}}
\end{aug}

\received{\smonth{2} \syear{2012}}
\revised{\smonth{8} \syear{2014}}

%
\begin{abstract}
The XOR-satisfiability (XORSAT) problem requires finding an assignment
of $n$
Boolean variables that satisfy $m$ exclusive OR (XOR) clauses, whereby each
clause constrains a subset of the variables. We consider random XORSAT
instances, drawn uniformly at random from the ensemble of
formulae containing $n$ variables and $m$ clauses of size $k$. This model
presents several structural similarities to other ensembles of constraint
satisfaction problems, such as $k$-satisfiability ($k$-SAT),
hypergraph bicoloring and graph coloring. For many of these
ensembles, as the number of constraints per variable grows, the set of
solutions shatters into an exponential number of well-separated components.
This phenomenon appears to be related to the difficulty of solving random
instances of such problems.

We prove a complete characterization of this clustering phase
transition for
random $k$-XORSAT. In particular, we prove that the clustering
threshold is
sharp and determine its exact location. We prove that the set of
solutions has
large conductance below this threshold and that each of the clusters
has large
conductance above the same threshold.

Our proof constructs a very sparse basis for the set of solutions (or the
subset within a cluster). This construction is intimately tied to the
construction of specific subgraphs of the hypergraph associated with an
instance of
$k$-XORSAT. In order to study such subgraphs, we establish
novel local weak convergence results for them.\vspace*{-3pt}
\end{abstract}

%
\end{frontmatter}

\section{Introduction}\label{sec1}

An instance of XOR-satisfiability (XORSAT) is specified by an integer
$n$ (the
number of variables) and by a set of $m$ clauses of the form
$x_{i_a(1)}\oplus\cdots\oplus x_{i_a(k)}=b_a$ for $a\in[m]\equiv\{
1,\ldots,m\}$.
Here, $\oplus$ denotes modulo-$2$ sum, $\ub=(b_1,\ldots,b_m)$ is a
Boolean vector, $b_a\in\{0,1\}$, specified by the problem instances,
and $\ux=(x_1,\ldots,x_n)$ is a vector of Boolean variables
$x_i\in\{0,1\}$ that must be chosen to satisfy the clauses.

Standard linear algebra methods allow us to determine whether a given XORSAT
instance admits a solution, to find a solution, and even to count the
number of
solutions, all in polynomial time. In this paper,\vadjust{\goodbreak} we shall be
interested in the
structural properties of the set of solutions $\cS\subseteq\{0,1\}^n$
of a
random $k$-XORSAT formula. More explicitly, we consider a random XORSAT
instance $\Inst$ that is drawn uniformly at random within the set $\G(n,k,m)$
of instances with $m$ clauses over $n$ variables, whereby each clause involves
exactly $k$ variables. The set of solutions $\cS=\cS(\Inst)$ is then defined
as the set of binary vectors $\ux$ that satisfy all $m$ clauses.

Since $\Inst$ is a random formula, $\cS$ is a random subset of the Hamming
hypercube. The structural properties of $\cS$ are of interest for several
reasons. First of all, linear systems over finite fields are combinatorial
objects that emerge naturally in a number of fields. Dietzfelbinger and
collaborators \cite{Cuckoo} use a mapping between XORSAT and the matching
problem to establish tight thresholds for the performances of Cuckoo Hashing,
an archetypal load balancing scheme. Such thresholds are computed by
determining thresholds above which the set of solutions $\cS$ of a random
XORSAT formula becomes empty. The existence of solutions is in turn
related to
the existence of an even-degree subgraph in a random hypergraph. Random sparse
linear systems over finite fields are used to construct capacity achieving
error correcting codes \cite{Luby98,Luby01,RiUBOOK}. The decodability
of such
codes is related to the emergence of a nontrivial $2$-core in the same random
hypergraph---a phenomenon that will play a crucial role in the following.
Finally, structured linear systems over finite fields are generated by popular
factoring algorithms \cite{Factorization}.

In the present paper, we are also motivated by the close analogy between
random $k$-XORSAT and other random ensembles of constraint satisfaction
problems (CSPs). The prototypical example of this family is random
$k$-satisfiability ($k$-SAT). The random $k$-SAT ensemble can be
described in
complete analogy to random $k$-XORSAT with the modification of replacing
exclusive OR clauses by OR clauses among variables or their negations. Namely,
in $k$-SAT each clause takes the form $(x_{i_a(1)}'\vee\cdots\vee
x_{i_a(k)}')$,\vspace*{1pt} whereby $x_{i_a(\ell)}'=x_{i_a(\ell)}$ or
$x_{i_a(\ell)}'=\overline{x}_{i_a(\ell)}$. An extensive literature
\cite
{Monasson99Nature,MezParZec_science,AchlioetalNature,KMRSZ07,MM09,AchlioCoja}
provides strong support for the existence of two sharp thresholds in random
$k$-SAT, as the number of clauses per variable $\alpha= m/n$ grows.
First, as
$\alpha$ crosses a ``satisfiability threshold'' $\as(k)$, random
$k$-SAT formulae
pass from being with high probability (w.h.p., i.e., with probability
converging to $1$ as $n\to\infty$) satisfiable [for $\alpha<\as
(k)$] to
being w.h.p.
unsatisfiable [for $\alpha>\as(k)$]. For any $\alpha<\as(k)$ the
set of
solutions is therefore nonempty. However, it undergoes a dramatic
structural change as $\alpha$ crosses a second threshold $\ad(k)<\as
(k)$. While
for $\alpha<\ad(k)$, $\cS$ is w.h.p. ``well connected'' (more precise
definitions
will be given below), for $\alpha\in[\ad(k),\as(k)]$ it shatters
into an
exponential number of clusters. It has been argued that such a ``clustered''
structure of the space of solutions can have an intimate relation with the
failure of standard polynomial time algorithms when applied to random formulae
in this regime. The same scenario is thought to hold for a number of random
constraint satisfaction problems (including, e.g., proper
coloring of
random graphs, bicoloring random hypergraphs, Not All Equal-SAT, etc.).

Unfortunately, this fascinating picture is so far only conjectural.
Even the
best understood element, namely the existence of a satisfiability threshold
$\as(k)$ has not been established (with the exception of the special case
$k=2$). In an early breakthrough, Friedgut \cite{Friedgut} used
Fourier-analytic methods to prove the existence of a---possibly
$n$-dependent--- sequence of thresholds $\as(k;n)$. Proving that in
fact this
sequence can be taken to be $n$-independent is one of the most challenging
open problems in probabilistic combinatorics and random graph theory.
Understanding the precise connection between clustering of the space of
solutions and computational complexity is an even more daunting task.

Given such outstanding challenges, a fruitful line of research has
pursued the
analysis of somewhat simpler models. A very interesting possibility
is to study $k$-SAT formulae for large but still bounded values of $k$.
As explained in \cite{AchlioCoja}, each SAT clause eliminates only one binary
assignment of its $k$ variables, out of $2^k$ possible assignments of
the same
variables. Hence, for $k$ large, a single clause has a small effect on
the set
of solutions, and most binary vectors are satisfying unless the formula
includes about $2^k$ clauses per variable. This results in an ``averaging''
effect and suitable moment methods provide asymptotically sharp results for
large $k$. In particular, Achlioptas and Peres \cite{AchlioPeres}
proved upper
and lower bounds on $\as(k)$ that become asymptotically equivalent
(i.e., whose ratio converges to $1$) as $k$ gets
large. Achlioptas and Coja-Oghlan \cite{AchlioCoja,AchlioCojaRicci}, proved
that clustering indeed takes place in an interval of values of $\alpha
$ below
the satisfiability threshold and obtained upper and lower bounds on the
corresponding threshold $\ad(k)$ that are asymptotically equivalent
for large $k$. Finally, Coja-Oghlan \cite{CojaBetter} proved that
solutions can be found
w.h.p. in polynomial time for any $\alpha<\alpha_{\mathrm{ d},
\mathrm
{alg}}(k)$, whereby
$\alpha_{\mathrm{ d}, \mathrm{alg}}(k)$ is asymptotically equivalent
to $\ad(k)$ for large
$k$. Intriguingly, no algorithm is known that can provably find
solutions in
polynomial time for $\alpha\in((1+\delta)\ad(k),\as(k))$, for any
$\delta>0$,
and all $k\ge3$.

XORSAT is a very different example on which rigorous mathematical analysis
proved possible, thus providing precious complementary insights. The key
simplification is that the set of solutions $\cS$ is, in this case, an affine
subspace of the Hamming hypercube $\{0,1\}^n$ (viewed as a vector space over
$\GF[2]$). This implies a high degree of symmetry that can be
exploited to
obtain very sharp characterizations for large $n$, and any $k$ (we assume
throughout that $k\ge3$, since $2$-XORSAT is significantly simpler).

It was proved in \cite{DuboisFOCS} that, for $k= 3$, there exists an
$n$-independent threshold $\as(k)$ such that a random $k$-XORSAT
instance is
w.h.p. satisfiable if $\alpha<\as(k)$ and unsatisfiable if $\alpha
>\as
(k)$. The
proof constructs a subformula, by considering the $2$-core of the hypergraph
associated with the XORSAT instance. One can then prove that the original
formula is satisfiable if and only if the $2$-core subformula is. The
threshold for the latter can be determined exactly using the second moment
method. The proof was extended to all $k\ge4$ in \cite{Cuckoo}.

The existence of a $2$-core in a random XORSAT formula has a sharp
threshold when the number of clauses per variable $\alpha$ crosses a value
$\alpha_{\mathrm{core}}(k)$. This was argued to be intimately related
to the
appearance of clusters. In particular, \mbox{\cite{MezRicZec_XOR,CoccoXOR}}
give an
argument\setcounter{footnote}{2}\footnote{The argument of \cite{MezRicZec_XOR,CoccoXOR} is essentially
rigorous, but does not deal with several technical steps.} showing
that, above
$\alpha_{\mathrm{core}}(k)$, the space of solutions shatters
into exponentially many clusters. In other words, $\alpha_{\mathrm
{core}}(k)$ is an upper bound on the clustering threshold.
\cite{MezRicZec_XOR} further shows that, for $\alpha<\alpha_{\mathrm
{core}}(k)$, a particular coordinate of a solution can be changed by
changing $O(1)$ other variables on
average, without leaving the space of solutions. If this argument is
pushed a step further, one can show that, w.h.p., any coordinate can be changed
by flipping at most $O(\log n)$ other coordinates. This suggests that
it \emph{may} be possible to concatenate a sequence of such flips to
connect any two solutions via a path through the solution space, with
$O(\log n)$ steps. However, the analysis \cite{MezRicZec_XOR}
does not imply that this is the case, as it does not address the main
challenge, namely
to construct a path from any solution to any other solution.
In this work, we solve this problem
and provide the first proof of a lower bound of $\alpha_{\mathrm{core}}(k)$
on the clustering threshold $\ad$, thus establishing that
indeed $\ad(k) =\alpha_{\mathrm{core}}(k)$. For $\alpha>\ad(k)$ we prove
a sharp
characterization of the decomposition into clusters.

As mentioned above, random $k$-XORSAT formulae can be solved in
polynomial time
using linear algebra methods, and this appears to be insensitive to the
clustering threshold. Nevertheless, an intriguing algorithmic phase transition
might take place \emph{exactly} at the clustering threshold $\ad(k)$.
For any
$\alpha<\ad(k)$, solutions can be found in time linear in the number of
variables (the algorithm is in fact an important component of our
proof). On
the other hand, no algorithm is known that finds a solution in linear
time for
$\alpha\in(\ad(k),\as(k))$. We think that our proof sheds some light
on this
phenomenon.

\subsection{Main result}

In this paper, we obtain two sharp results characterizing the clustering
phase transition for random $k$-XORSAT:
\begin{longlist}[(ii)]
\item[(i)] We exactly determine the clustering threshold $\ad(k)$, proving
that the space of solutions is w.h.p. well connected for $\alpha<\ad
(k)$, and
instead shatters into exponentially many clusters for
$\alpha\in(\ad(k),\as(k))$.

\item[(ii)] We determine the exponential growth rate of the number of
clusters, that is, we show that this is w.h.p. $\exp\{n\Sigma(\alpha
;k)+o(n)\}$ where
$\Sigma(\alpha;k)$ is a nonrandom function which is explicitly given.
We prove
that each of the clusters is itself ``well connected.''
\end{longlist}

This is therefore the first random CSP ensemble for which a sharp threshold
for clustering is proved.

Earlier literature fell short of establishing (i) since it did not provide
any argument for connectedness below $\ad(k)$. Also, informal
calculations only suggested a lower bound on the number of clusters,
but did
not establish (ii) since they did not prove connectedness of each
cluster by
itself. The situation is akin to the analysis of Markov Chain Monte Carlo
methods: It is often significantly more challenging to prove rapid mixing
(connectedness of the space of configurations) than the opposite
(i.e., to find bottlenecks).

One important novelty is that the notion of connectedness used here is very
strong and goes beyond path connectivity, which was used earlier for $k$-SAT
\cite{AchlioCoja,AchlioCojaRicci}. We use a properly defined notion of
\emph{conductance} which we think can be applied to a broader set of
CSPs, and
has the advantage of being closely related to important algorithmic notions
(fast mixing for MCMC and expansion). Given a subset of the hypercube
$\cS\subseteq\{0,1\}^n$, and a positive integer $\ell$, we define the
conductance of $\cS$ as follows.
Construct the graph $\cG(\cS,\ell)$ with
vertex set $\cS$ and an edge connecting $\ux,\ux'\in\cS$ if and
only if
$d(\ux,\ux')\le\ell$ [here and below, $d( \cdot, \cdot)$
denotes the
Hamming distance, i.e., $d(\ux,\ux') = |\{i\dvtx1\leq i \leq n, x_i
\neq
x'_i \}|$,
where $\ux= (x_1, x_2, \ldots, x_n)$ and similarly for $\ux'$, and
$|B|$ denotes the cardinality of the set $B$].
Then we define the \emph{$\ell$th conductance of $\cS$} as the
graph conductance of $\cG(\cS,\ell)$, namely
%
%
\begin{equation}
\Phi(\cS;\ell) \equiv\min_{A\subseteq\cS} \frac{\cut_{\cG(\cS,\ell)}(A,\cS\setminus A)}{\min(|A|,|\cS
\setminus
A|)}
,\label{eq:ConductanceDef} %
\end{equation}
where, for a graph $\cG= (\cV, \cE)$, and any $B \subseteq\cV$, we define
\[
\cut_{\cG}(B,\cV\setminus B) \equiv\bigl|\{e \in\cE\dvtx\mbox {Exactly one
of the two endpoints of $e$ is in $B$} \}\bigr|.
\]
Notice that we measure
the volume of a set by the number of its vertices instead of the sum of its
degrees.\footnote{This difference is irrelevant for $\alpha<\alpha
_{\mathrm{d}}(k)$
since in this case $\cS$ will be taken to be an affine subspace of $\{
0,1\}^n$, and hence
$\cG(\cS,\ell)$ will be a regular graph. For $\alpha\in(\alpha
_{\mathrm{d}}(k),\alpha_{\mathrm s}(k))$, $\cS$ will be constructed
as the union of
a ``small'' number of affine spaces, and hence $\cG(\cS,\ell)$
should be
approximately regular. We keep the definition
(\ref{eq:ConductanceDef}) since it simplifies our statements.}

We define the distance between two subsets of the hypercube $\cS_1,
\cS_2
\subseteq\{0,1\}^n$ as
\[
d(\cS_1, \cS_2) \equiv\min_{\ux\in\cS_1, \ux' \in\cS_2} d
\bigl(\ux , \ux '\bigr).
\]
For our statements, $k\ge3$ is always fixed, together with a sequence
$m(n)=\alpha n$.

We say that a sequence of events $(\Ev_n)_{n>0}$ occurs \emph{with
high probability} (w.h.p.) if $\lim_{n\rightarrow\infty}
\prob(\Ev_n)=1$.
(We refer to Section~\ref{sec:Definitions} for a formal definition of
the underlying probability space.)

%
%
\begin{theorem}\label{thm:Main}
Let $\cS$ be the set of solutions of a random $k$-\emph{XORSAT}
formula with
$n$ variables and $m = n\alpha$ clauses. For any $k\ge3$, let $\ad(k)$
be defined as
%
%
\begin{equation}
\label{eq:AlphaDdef} \ad(k) \equiv\sup \bigl\{\alpha\in[0,1] \dvtx
z>1-e^{-k\alpha
z^{k-1}},
\forall z\in(0,1) \bigr\}.
\end{equation}
\begin{enumerate}
\item If $\alpha<\ad(k)$, there exists $\const=\const(\alpha
,k)<\infty$
such that, w.h.p., $\Phi(\cS;\break (\log n)^{\const}) \ge1/2$.
\item If $\alpha\in(\ad(k),\as(k))$, then there
exists $\ve=\ve(k;\alpha)>0$ such that, w.h.p., $\Phi(\cS;n\ve) =0$.
\item If $\alpha\in(\ad(k),\as(k))$, and
$\delta>0$ is arbitrary, then there exist constants $\const=\const
(\alpha,k)<\infty$,
$\ve=\ve(\alpha,k)>0$, $\Sigma=\Sigma(\alpha,k)>0$,
and a partition of the set of solutions
$\cS=\cS_1\cup\cdots\cup\cS_N$, such that, w.h.p., the following
properties hold:
\begin{enumerate}[(a)]
\item[(a)] For each $a\in[N]$, we have $\Phi(\cS_a;(\log n)^{\const})
\ge1/2$.
\item[(b)] For each $a\neq b\in[N]$, we have $d(\cS_a,\cS_b)\ge n\ve$.
\item[(c)]$\exp\{n(\Sigma-\delta)\}\le N\le\exp\{n(\Sigma+\delta)\}$.
Further, letting $Q$ be the largest positive solution of
$Q=1-\exp\{-k\alpha Q^{k-1}\}$ and $\hQ\equiv Q^{k-1}$, we have
$\Sigma(\alpha,k) = Q-k\alpha\hQ+(k-1)\alpha Q\hQ$.
\end{enumerate}
\end{enumerate}
\end{theorem}
%
%

\subsection{Conductance and sparse basis}

We will prove Theorem~\ref{thm:Main} by obtaining a fairly complete description
of the set $\cS$ both above and below $\ad(k)$. In a nutshell, for
$\alpha<\ad(k)$, $\cS$ admits a sparse basis, while for $\alpha>\ad(k)$
each of
the clusters $\cS_1, \ldots, \cS_N$ admits a sparse basis but their
union does not.
This is particularly suggestive of the connection between the
clustering phase
transitions and algorithm performance. Below $\ad(k)$ the space of solutions
admits a succinct explicit representation [in $O(n(\log n)^\const)$
bits]. Above $\ad(k)$, we
can produce a representation that is succinct but implicit (as
solutions of a
given formula), or explicit but prolix [no basis is known that can be encoded
in $o(n^2)$ bits].

Given a linear subspace $\cS\subseteq\{0,1\}^n$, we say that it
admits an
$s$-sparse basis if there exist vectors
$\ux^{(l)}\in\cS$ for $l \in\{1, \ldots, D\}$ such that $d(\ux
^{(l)},\underline{0})\le s$ and
$(\ux^{(l)})_{l=0}^D$ form a basis for $\cS$. The latter means that
the vectors are linearly independent and $\cS=  \{\sum_{l=1}^D
a_l\ux^{(l)}\dvtx(a_l)_{l=0}^D\in\{0,1\}^D  \}$.

We say that an affine space $\cS\subseteq\{0,1\}^n$ admits an
$s$-sparse basis if,
for $\ux^{(0)}\in\cS$, the linear subspace $\cS-\ux^{(0)}$
admits an $s$-sparse basis.
The property of having a sparse basis indeed
implies large conductance. The proof is immediate.

%
%
\begin{lemma} \label{lemma:Sparse}
If the affine subspace $\cS\subseteq\{0,1\}^n$ admits a $s$-sparse
basis, then
$\Phi(\cS;s)\ge1/2$.

Vice versa, assume that $\Phi(\cS;s)=0$. Then $\cS$ does not admit a
$s$-sparse basis.
\end{lemma}
%
%
\begin{pf}
We can assume, without loss of generality, that $\cS$ is a linear
space.
Let $d$ be its dimension. Further, given a graph $\cG$,
let, with a slight abuse of notation
%
%
\begin{equation}
\Phi(\cG) \equiv\min_{A\subseteq
\cS}\frac{\cut_{\cG}(A,\cS\setminus A)}{\min(|A|,|\cS\setminus A|)} ,
\end{equation}
so that $\Phi(\cS;\ell) = \Phi(\cG(\cS;\ell))$.

Assume that $\cS$ admits a $s$-sparse basis.
This immediately implies the graph $\cG(\cS,s)$ contains
a spanning subgraph that is isomorphic to the $d$-dimensional
hypercube $\cH_d$. Further, $\cG\mapsto\Phi(\cG)$ is monotone
increasing in the edge set of $\cG$. Therefore, $\Phi(\cS;s)\ge
\Phi(\cH_d)\ge1/2$ where the last inequality follows from
the standard isoperimetric inequality on the hypercube \cite{Hoory}.
\end{pf}
%

The characterization of the solution space in terms of sparsity of its
basis is given below.

%
\begin{theorem}\label{thm:MainSparse}
Let $\cS$ be the set of solutions of a random $k$-\emph{XORSAT}
formula with
$n$ variables and $m = n\alpha$ clauses. For any $k\ge3$, let $\ad(k)$
be defined as per equation~(\ref{eq:AlphaDdef}). Then the following hold:
\begin{enumerate}
\item If $\alpha<\ad(k)$, there exists $\const=\const(\alpha
,k)<\infty$
such that, w.h.p., $\cS$ admits a $(\log n)^\const$-sparse basis.
\item If $\alpha\in(\ad(k),\as(k))$, and
$\delta>0$ is arbitrary, then there exist constants $\const=\const
(\alpha,k)<\infty$,
$\ve=\ve(\alpha,k)>0$, $\Sigma=\Sigma(\alpha,k)>0$,
and a partition of the set of solutions
$\cS=\cS_1\cup\cdots\cup\cS_N$, such that, w.h.p., the following
properties hold:
\begin{enumerate}[(a)]
\item[(a)] For each $a\in[N]$, $\cS_a$ admits a $(\log n)^\const$-sparse basis.
\item[(b)] For each $a\neq b\in[N]$ we have $d(\cS_a,\cS_b)\ge n\ve$.
\item[(c)]$\exp\{n(\Sigma-\delta)\}\le N\le
\exp\{n(\Sigma+\delta)\}$. Further, $\Sigma$ is given by the same
expression given in Theorem~\ref{thm:Main}.
\end{enumerate}
\end{enumerate}
\end{theorem}

Clearly, this theorem immediately implies Theorem~\ref{thm:Main} by applying
Lem\-ma~\ref{lemma:Sparse}. The rest of this paper is devoted to
the proof of Theorem~\ref{thm:MainSparse}.

\subsection{Further technical contributions}

To a given a XORSAT instance $\Inst$, we can associate a bipartite graph
(``factor graph'') with vertex sets $F$ (\emph{factor} or \emph
{check} nodes)
corresponding to equations, and $V$ (\emph{variable} nodes) variables.
The edge set $E$ includes those pairs $(a,i)\in F\times V$ such that
variable $x_i$ participates in the $a$th equation.
The construction of the sparse basis in Theorem~\ref{thm:MainSparse}
relies heavily on a characterization of the random factor graph
associated to a
random XORSAT instance. This could be gleaned from the proof of \cite
{DuboisFOCS,Cuckoo}
that construct the $2$-core of $G$.
In order to prove Theorem~\ref{thm:MainSparse}, we characterize a
larger subgraph that we refer to as the \emph{backbone} of $G$. This
subgraph has the following interpretation: if two solutions $\ux$ and
$\ux'$ coincide on the core, then they coincide on every vertex of the
backbone.

The $2$-core of the random graph $G$ was studied in a number of papers
\cite{PittelSpencerWormald,Luby98,Molloy,DemboFSS}. The key tool in
these works is the analysis of an iterative procedure that constructs
the $2$-core in $\Theta(n)$ iterations. This procedure has an
important property:
At each step, the resulting graph remains uniformly random,
given a small number of parameters (essentially, its degree
distribution). Thanks to this property, the analysis of
\cite{PittelSpencerWormald,Luby98,Molloy,DemboFSS}
is reduced to the study of a Markov chain in $\integers^2$.
This is done by showing that sample paths of this chain are shown to
concentrate around
solutions of a certain ordinary differential equation.

Our analysis of the backbone has a similar starting point, namely the
study of an iterative procedure that constructs the backbone (indeed
we define formally the backbone as the fixed point of this
procedure). Unfortunately, the graphs generated by this procedure are
not uniformly random, conditional on a small number of
parameters. Hence, the techniques
\cite{PittelSpencerWormald,Luby98,Molloy,DemboFSS} do not apply.
We overcome this difficulty by characterizing the large-$n$ limit
of its fixed point using the theory of local weak convergence. This
is in turn challenging because the fixed point is not, \emph{a
priori}, a local function of $G$.

We consider this characterization of the backbone, and its proof,
to be a contribution of independent interest.

For describing the iterative procedure, we use the language of message
passing algorithms, and will refer to it as to ``belief propagation'' (BP),
as the same algorithm is also of interest in iterative coding; see
\cite{RiUBOOK,MM09}.
Given a factor graph $G=(F,V,E)$, the algorithm updates $2|E|$ \emph{messages}
indexed by directed edges in~$G$. In other
words, for each $(a,v)\in E$, $a\in F$ and $v\in V$, and any iteration
number $t\in\naturals$, we have two messages $\vtoc{v}{a}^t$,
and $\ctov{a}{v}^t$, taking values in $\{0,*\}$.
For $t\geq1$, messages are computed following the update rules:
%
%
\begin{equation}\label{eq:BP1}
\vtoc{v} {a}^{t} = \cases{
*, &\quad
$\mbox{if $\ctov{b} {v}^{t-1} = *$ for all $b \in\partial v \setminus a$,}$
\vspace*{2pt}\cr
0, &\quad $\mbox{otherwise,}$}
\end{equation}
and
%
%
\begin{equation}\label{eq:BP2}
\ctov{a} {v}^t = \cases{ %
0, &\quad $\mbox{if
$\vtoc{u} {a}^t = 0$ for all $u \in\partial a \setminus v$,}$
\vspace*{2pt}\cr
*, &\quad $\mbox{otherwise.}$}
\end{equation}
We call this algorithm \BPzero\ when all
messages are initialized to $0$: $\vtoc{v}{a}^0=\ctov{a}{v}^0=0$ for
all $(a,v)\in E$. It is not hard to see that \BPzero\ is
monotone,\footnote{This can be established by induction: Since we start
with all $0$s, clearly messages can only change from $0$ to $*$ in the
first iteration. Thereafter, this holds inductively for each subsequent
iteration since each of the update rules is monotone in the sense that
if the incoming messages only change from $0$ to $*$, then the same
holds for the outgoing messages.} in the sense that messages only
change from $0$ to $*$, and hence converges to a fixed point $\nuia{v}{a}^*$.

It is easy to check (see Lemma~\ref{lemma:PeelingBPFP} below) that the
core of $G$ coincides
with the subgraph induced by the factor nodes that receive no $*$
message at the fixed point of \BPzero. The backbone is instead the
subgraph induced by factor nodes that receive at most one $*$ message
at the fixed point.

Denote by $\tmu_n^*$ the probability distribution on rooted factor graphs
with marks on the edges constructed as follows.
Draw a graph uniformly at random from $\G(n,k,m)$. Choose a uniformly
random variable node $i \in V$ as the root. Mark the edges (in each
direction) with the messages corresponding to the \BPzero\ fixed point
$\nuia{v}{a}^*$.

We next construct a random tree $\tT_*(\alpha,k)$ with marks on the
directed edges as follows.
Marks take values in $\{0,*\}$ and to each undirected edge we
associate a
mark for each of the two directions. We will refer to the direction
toward the root as to the ``upward'' direction, and to the opposite
one as to the ``downward'' direction.
The marks correspond to fixed-point BP messages,
and we will call them messages as well in what follows.
First, consider only edges directed upward.
This is a multitype Galton--Watson (GW) tree. At the root generate
$\operatorname{Poisson}(k\alpha)$ offsprings, and mark each of the edges to $0$
independently with probability $\hQ$, and to $*$ otherwise.
At a nonroot variable node, if the parent edge is marked $0$, generate
$\Po(k\alpha(1-\hQ))$ descendant edges marked $*$ and $\Po_{\ge
1}(k\alpha\hQ)$ descendant edges
marked $0$ [here $\Po_{\Ev}(\lambda)$ denotes a Poisson random variable
with parameter $\lambda$ conditional on $\Ev$]. If the parent edge is
marked $*$, generate
$\Po(k\alpha(1-\hQ))$ descendant edges marked $*$ and
no descendant edges marked~$0$.
At a factor node, if the parent edge is marked $0$, generate $k-1$
descendant edges marked $0$. If the parent node is marked $*$, generate
$M\sim\Binom_{\le k-2}(k-1,Q)$ descendants marked $0$, and $k-1-M$
descendants marked $*$.

For edges directed downward,
marks are generated recursively following the usual BP
rules, cf. equations~(\ref{eq:BP1}), (\ref{eq:BP2}), starting from
the top
to the bottom.
It is easy to check that with this construction,
the marks in $\tT_*(\alpha,k)$ correspond to a
BP fixed point. Given a factor graph $G=(F,V,E)$, we use $\Ball_{G}(v,t)$
to denote the ball of radius $t$ centered at node $v \in V$. This ball
is defined inductively as follows: The $\Ball_{G}(v,0)$ consists of
node $v$
alone and no edges. For $t>0$, the $\Ball_{G}(v,t)$ includes $\Ball
_{G}(v,t-1)$.
In addition, it includes all factor nodes connected to variable nodes
in $\Ball_{G}(v,t-1)$ and associated edges, and all
variable nodes connected to those factor nodes and associated edges. [Thus,
$\Ball_{G}(v,t)$ includes nodes and edges up to a distance $t$ from $v$,
where variable nodes are said to
be separated by distance $1$ if they are connected to the same factor node.]
%

\begin{definition}\label{def:WC}
Let $\{G_n\}$, $G_n = (F_n,V_n,E_n)$ be a sequence of (random) factor graphs.
Let $\mu_n^{(t)}$ denote the empirical probability distribution of
$\Ball_{G_n}(v,t)$ when $v\in V_n$ is uniformly random.
Explicitly, for any locally finite rooted graph $\T_0$ of depth at
most $t$,
%
%
\begin{equation}
\mu_n^{(t)}\equiv\frac{1}{n} \sum
_{v \in V_n} \ind \bigl(\Ball_{G_n}(v,t)\simeq
\T_0 \bigr),
\end{equation}
(with $\simeq$ denoting equality up to graph vertex relabeling.)
We say that $\{G_n\}$ \emph{converges locally almost surely} to the
measure $\mu$ on
rooted graphs if, for any finite $t$, and any locally finite rooted
graph $\T_0$ of depth at most $t$, we have
%
%
\begin{equation}
\lim_{n\to\infty} \mu_n^{(t)}(\T_0)
= \mu^{(t)}(\T_0) \label
{eq:LocalConvergence}
\end{equation}
holds almost surely with respect to the graph law. Here, $\mu^{(t)}$
denotes the marginal of $\mu$ with respect to a ball of radius $t$
around the root.
\end{definition}

The same notion of local graph convergence was used earlier in the
literature, for instance, in
\cite{dembo2010ising,DemboMontanariBrazil,dembo2013factor}. Given a
random graph
distribution, we first draw a sequence of $\{G_n\}_{n\ge1}$, and then
check that $\mu_n^{(t)}(\T_0) \to\mu^{(t)}(\T_0)$ with probability
one. It is worth emphasizing the difference from a weaker notion (that
we never use below),
whereby we only check $\E_{G_n}\mu_n^{(t)}(\T_0) \to\mu^{(t)}(\T_0)$,
with $\E_{G_n}$ denoting expectation with respect to the graph
distribution.
In particular, establishing almost sure local graph convergence is
more
challenging that proving convergence of the expectation
$\E_{G_n}\mu_n^{(t)}(\T_0)$ since it requires to control the
deviations of the subgraph counts $\mu_n^{(t)}(\T_0)$.
With this clarification, we shall occasionally drop the ``almost
surely'' in ``converges locally almost surely.''

As part of our proof of Theorem~\ref{thm:MainSparse}, we obtain the
following result, which may be of independent interest.
(We refer to the next section for a complete definition of the
underlying probability space.)
%

\begin{theorem}
\label{thm:BP_FP_localweaklimit}
The sequence $\{\tmu_n^*\}_{n \geq0}$ converges locally almost surely
to the
probability distribution of $\tT_*$.
\end{theorem}

Theorem~\ref{thm:BP_FP_localweaklimit} is proved in Section~\ref{sec:BP_DE}.

Besides this, our proof uses several other ideas:
\begin{itemize}
\item We show that Theorem~\ref{thm:BP_FP_localweaklimit} can be used
to extend the low weight core solutions to low weight solutions of
the whole XORSAT instance (see Section~\ref{app:core_basis}).
\item We show that the periphery (the complement of the core in $G$) is
uniformly random with a given
degree sequence, conditioned on being
``peelable.'' We estimate precisely this degree distribution, and
show that the periphery is indeed peelable with positive probability
for that degree sequence (see Section~\ref{secperrr}).
\item In addition to the fixed point characterization, we obtain a
precise characterization of the
convergence rate of \BPzero\ (see Section~\ref{sec:BP_DE}), which allows
us to bound the sparsity of the basis constructed.
\item For $\alpha> \alpha_{\mathrm{d}}$, convergence to the BP fixed point
is geometric rather than quadratic. In this regime, we show that in
fact there are ``strings'' of degree 2 variable nodes that slow down
convergence but do not prevent the construction of a sparse basis. We
bound the sparsity by defining a certain ``collapse'' operator on such
strings (see Section~\ref{sec:collapse_peeling_fast}).
\end{itemize}

\subsection{Outline of the paper}
In Section~\ref{sec:Definitions}, we define some basic concepts and notation.
Section~\ref{sec:MainProof} describes the construction of clusters and sparse
bases, and uses this construction to prove Theorem~\ref{thm:MainSparse}.
Several basic lemmas necessary for the proof are stated in this section.

Section~\ref{sec:BP_DE} introduces a certain \emph{belief
propagation} (BP)
algorithm and a technical tool called \emph{density evolution}, that
play a key
role in our analysis: The BP algorithm naturally decomposes
the linear system into a ``backbone'' (consisting roughly of the 2-core
and the variables implied by it) and a ``periphery.''
Density evolution allows us to track the progress of BP, eventually
facilitating a tight characterization of basic parameters (like number
of nodes) of the backbone and periphery.

Section~\ref{sec:collapse_peeling_fast} bounds the number
of iterations of a ``peeling'' algorithm (related to BP) that plays a
key role
in our construction of a sparse basis. Section~\ref{secperrr}
proves a sharp characterization of the periphery.
Together, this yields the first (large) set of basis vectors.

Section~\ref{app:proof_of_core_separation_lemma} shows the 2-core has
very few sparse solutions, leading to well separated, small,
``core-clusters.'' Section~\ref{app:core_basis} shows how to produce a
sparse solution of the linear system corresponding to each sparse
solution of the 2-core subsystem. This yields the second (small) set of
basis vectors in our construction.

Several
technical lemmas are deferred to the \hyperref[app]{Appendices}.

A short version of this paper was presented at the ACM-SIAM Symposium
on Discrete Algorithms SODA 2012.
%
\section{Random $k$-XORSAT: Definitions and notation}
\label{sec:Definitions}

As described in the \hyperref[sec1]{Introduction}, each $k$-XORSAT
clause is actually a linear
equation over $\GF[2]$: $x_{i_a(1)}\oplus\cdots\oplus x_{i_a(k)}
=b_a$, for
$a\in[m]\equiv\{1,\ldots,m\}$. Introducing a vector $\uh_a\in\{
0,1\}
^n$, with
nonzero entries only at positions $i_1(a),\ldots,i_k(a)$, this can be
written as
$\uh_a^T\ux= b_a$. Hence, an instance is completely specified by the pair
$(\bH,\ub)$ where $\bH\in\{0,1\}^{m\times n}$ is a matrix with rows
$\uh_1^T,\ldots,\uh_m^T$ and $\ub=(b_1,\ldots,b_m)^T\in\{0,1\}^m$.
The space of solutions is therefore $\cS\equiv\{\ux\in\{0,1\}^n
\dvtx\bH
\ux=
\ub\ \mod\ 2\}$. If $\cS$ has at least one element $\ux^{(0)}$, then
$\cS\oplus\ux^{(0)}$ is just the set of solutions of the homogeneous linear
system corresponding to $\ub=\underline{0}$ (the kernel of $\bH$).
In the
following we
shall always assume $\alpha<\as(k)$, so that $\cS$ is nonempty w.h.p.
\cite{Cuckoo}. Note that, if $\cS$ is nonempty, then $\cS= \cS_0
\oplus\ux_0$ where $\ux_0\in\cS$ is any solution of the original
system and $\cS_0$ is the set of solutions of the homogeneous linear
system $\bH\ux= \underline{0}$. Since we are only interested in geometric
properties of the set of
solutions that are invariant under translation, we will assume
hereafter that
$\ub=\underline{0}$, and hence $\cS=\cS_0$.

An XORSAT instance is therefore completely specified by a binary matrix
$\bH$,
or equivalently by the corresponding factor graph $G=(F,V,E)$. This is a
bipartite graph with two sets of nodes: $F$ (\emph{factor} or \emph
{check} nodes)
corresponding to rows of $\bH$, and $V$ (\emph{variable} nodes)
corresponding to
columns of $\bH$. The edge set $E$ includes those pairs $(a,i)$, $a\in F$,
$i\in V$ such that $\bH_{ai}=1$.
We denote by $\G(n,k,m)$ the set of all factor graphs with $n$ labeled variable
nodes and $m$ labeled check nodes, each having degree exactly $k$ (with no
double edges). Note that $|\G(n,k,m)|= {n\choose k}^m$. With a slight
abuse of notation, we will denote by $\G(n,k,m)$ also the uniform
distribution over this set, and write $G\sim\G(n,k,m)$ for a
uniformly random such graph.

For $v\in V$ or $v\in F$, we denote by $\deg_G(v)$, the degree of node
$v$ in graph $G$ (omitting the subscript when clear from the context)
and we let $\partial v$ denote the set of neighbors of $v$.
We define the distance with respect to $G$ between two variable nodes
$i$, $j\in V$, denoted by $d_G(i,j)$ as the length of the shortest path from
$i$ to $j$ in $G$, whereby the length of a path is the number of check nodes
encountered along the path. Given a vector $\ux$, we denote by $\ux_A=
(x_i)_{i\in A}$ its restriction to $A$. The cardinality of set $A$ is
denoted by $|A|$.

We only consider the ``interesting'' case $k \geq3$, and the asymptotics
$m,n\to\infty$ with $m/n\to\alpha$ and $\alpha\in[0,\as(k))$, where
$\as(k)$ is the satisfiability threshold. Hence, $\bH$ has w.h.p.
maximum rank,
that is, $\rank(\bH) = m$ \cite{MM09}.
%

\begin{definition}
Let $F_0\subseteq F$. The \emph{subgraph induced} by $F_0$ is defined as
$(F_0,V_0, E_0)$ where $V_0 \equiv\{i \in V\dvtx\di\cap F_0\neq
\varnothing
\}$
and $E_0\equiv\{(a,i)\in E\dvtx a \in F_0, i \in V_0\}$. A \emph
{check-induced}
subgraph is the subgraph $(F_0,V_0,E_0)$ induced by some $F_0\subseteq F$.
Similarly, we can define the subgraph induced by $V_0 \subseteq V$, and
\emph{variable-induced} subgraphs.

Let $F_0\subseteq F$, $V_0\subseteq V$. The \emph{subgraph induced} by
$(F_0, V_0)$ is defined as $(F_0,V_0, E_0)$ where $E_0\equiv\{(a,i)\in
E\dvtx a
\in F_0, i \in V_0\}$.
\end{definition}

%
\begin{definition}
\label{def:stopping_set}
A \emph{stopping set} is a check-induced subgraph with the property
that every
variable node has degree larger than one with respect to the subgraph.
The \emph{$2$-core} of $G$ is its maximal stopping set.
\end{definition}
Notice that the \emph{maximal} stopping set of $G$ is uniquely defined
because the union of
two stopping sets is a stopping set.

All of our statements are with respect to the following probability
space, for a fixed $k\ge3$, and an integer sequence $\{m(n)\}_{n\in
\naturals}$.
For each $n$, we let $m=m(n)$ and consider the finite set $\Omega_n =
\G
(n,k,m)$ of $k$-XORSAT
instances with $n$ variables and $m$ clauses. Formally, each element
of $\G(n,k,m)$ is given by a pair $(\bH,\ub)$ where $\bH\in\{0,1\}
^{m\times
n}$ is a matrix with $k$ nonzero elements per row and
$\ub\in\{0,1\}^m$. [In the proofs, we shall occasionally replace
$\G(n,k,m)$ by slightly different sets---defined therein---for
technical convenience. The connection will be made clear.]

Since $\Omega_m$ is finite, it is straightforward to endow it with
the uniform probability measure $\prob_n$ over the complete
$\sigma$-algebra $2^{\Omega_n}$. The probability space underlying all
of our statements is the product space $\Omega=
\bigtimes_{n\in\naturals}\Omega_n$, with product probability measure
$\prob= \bigtimes_{n\in\naturals}\prob_n$. An event $\Ev\subseteq
\Omega$ is a an
element of the product $\sigma$-algebra. As a special example, let
$f_n\dvtx\Omega_n\to\reals$ be a sequence of functions, and $\omega=
(\omega_i)_{i\in\naturals}\in\Omega$.
Then existence of the limit $\lim_{n\to\infty}f_n(\omega_n)$ is a well
defined event in $\Omega$.

With a slight abuse of
language, we will identify any set $\Ev_n\subseteq\Omega_n$ with an
event, namely with the cylindrical set $C(\Ev_n) \equiv\{\omega=
(\omega_i)_{i\in\naturals}\in\bigtimes_{i\in\naturals}\Omega
_i\dvtx\omega_n\in\Ev_n\}$. We will typically write $\Ev_n$ for
$C(\Ev_n)$ and
$\prob(\Ev_n) = \prob(\{\omega_n\in\Ev_n\})$ for the
probability of such an event. We say that $\Ev_n$ occurs \emph{with high
probability} (\textit{w.h.p.}) if $\lim_{n\to\infty} \prob(\Ev_n) = 1$. We
say that
a sequence of events
$(\Ev_n)_{n>0}$ occurs \emph{eventually almost surely} if
$\lim_{n_0\to\infty}\prob(\bigcap_{n\ge n_0}C(\Ev_n)) = 1$.

Note that, with this probability space, the notion of local almost sure
convergence in Definition~\ref{def:WC} is well defined. Note that our
main results (Theorems \ref{thm:Main} and~\ref{thm:MainSparse})
are ``with high probability results,'' and hence do not require the
definition of a common probability space for different graph
sizes. This is indeed mainly a matter of technical convenience (and is
of course needed for Theorem~\ref{thm:BP_FP_localweaklimit}).

A key fact to be used in the following is that a giant $2$-core appears
abruptly at $\ad(k)$. Forms of the following statement appear in
\cite{Luby98,Molloy,DemboFSS}.
%

\begin{theorem}[(\cite{Luby98,Molloy,DemboFSS})]\label{lemma:BasicCore}
Assume $\alpha<\ad(k)$. Then w.h.p., a graph $G\sim\G(n,k,m)$ does not
contain any stopping set.

Vice versa, assume $\alpha>\ad(k)$. Then there exists
$\const(k)>0$ such that, w.h.p., a~graph $G$ drawn uniformly at random from
$\G(n,k,m)$ contains a
$2$-core of size larger than $\const(k)n$.
\end{theorem}

We will often refer to the depth-$t$ neighborhood of a node $v$ in $G$.
%

\begin{definition}
\label{def:ball}
Given a node $v \in V$ and an integer $t$, let $V'=\{u \dvtx u\in V,\dG
(u,v)\leq
t\}$. Then the \emph{ball of radius $t$ around node $v$} is defined as the
(variable-induced) subgraph $\Ball_G({v},{t})$ induced by $V'$. With
an abuse
of notation, we will use the same notation for the set of variable
nodes in
$\Ball_G({v},{t})$. Lastly, we define $|\Ball_G({v},{t})|$ to be the
number of
variable nodes in the subgraph $\Ball_G({v},{t})$.
\end{definition}

We will occasionally work with certain random infinite rooted factor
graphs, with marks on the edges or vertices. (Note that a factor graph
can be regarded as an ordinary graph, with additional marks on the
vertices to distinguish ``variable nodes'' from ``factor nodes.'') A
useful concept in this context is the one of ``unimodular'' random
rooted graphs, that we briefly recall next. For a more complete
presentation, we refer to the overview paper by Aldous and Lyons \cite
{AldousLyonsUnimodular}.

Informally, a random rooted (marked) graph is unimodular if
it looks the same (in distribution), when the root is moved to any
other vertex.
In order to formalize this notion, we denote by $\cG_*$ the space of
locally finite rooted graphs, with marks on the vertices or edges (we
assume marks to belong to some fixed finite set for simplicity).
We view two graphs that differ by an isomorphism as identical.
This space can be endowed by a metric that metrizes local
convergence, and hence a Borel $\sigma$-algera.

Analogously, we denote by
$\cG_{**}$ the space of doubly rooted graphs [a \emph{doubly rooted
graph} is a graph with two distinguished vertices, i.e.,
a triple $(G,u,v)$ where $G=(V,E)$ is a graph, and $u,v\in V$].
As for the simply rooted case, $\cG_{**}$ can be made into a complete metric
space; we regard it as a measurable space endowed
with the Borel $\sigma$-algebra.
%

\begin{definition}
Let $(G,\root)$ be a random rooted graph with root $\root$.
We say that $(G,\root)$ is unimodular if, for any measurable
function $f\dvtx\cG_{**}\to\reals_{\ge0}$, $(G,u,v)\mapsto
f(G,u,v)$, we have
%
%
\begin{equation}
\E \biggl[\sum_{v\in V(G)} f(G,\root,v)
\biggr] = \E \biggl[\sum_{v\in
V(G)} f(G,v,\root) \biggr].
\end{equation}
\end{definition}

Consequences, and equivalent versions of unimodularity can be found
in~\mbox{\cite{AldousLyonsUnimodular,MontanariStFlour}}.
%
%
\section{Proof of Theorem \texorpdfstring{\protect\ref{thm:MainSparse}}{2}}
\label{sec:MainProof}

In this section, we describe the construction of clusters and sparse
bases within the clusters [or for the whole space of solutions for
$\alpha\in[0,\ad(k))$]. 
The analysis of
this construction is given in Section~\ref{sec:AnalysisConstruction}
in terms
of a few technical lemmas. Finally, the formal proof of Theorem~\ref
{thm:MainSparse} is given in Section~\ref{sec:Together}.

%
%
\begin{table}
\tablewidth=200pt
\caption{Synchronous peeling algorithm} \label{table:sync_peeling}
\begin{tabular*}{200pt}{@{\extracolsep{\fill}}l@{}}
\hline
\multicolumn{1}{c}{\textsc{\textbf{Synchronous Peeling}}
$\bolds{(}\mathbf{Graph}\ \bolds{G= (F, V, E))}$}\\
\hline
$F' \leftarrow F$\\
$V' \leftarrow V$\\
$E' \leftarrow E$\\
$J_0 \leftarrow(F,V,E)$, $t=0$\\
While $J_t$ has a variable node of degree $\leq1$ do\\
\quad $t \leftarrow t+1$\\
\quad $V_t \leftarrow\{v \in V'\dvtx\deg_{G_{t-1}}(v) \leq1\}$\\
\quad $F_t \leftarrow\{a \in F'\dvtx(v,a) \in
\quad E'\ \mathrm{for\ some}\ v \in V_t\}$\\
\quad $E_t \leftarrow\{(v,a) \in E'\dvtx a \in F_t, v \in V'\}$\\
\quad $F' \leftarrow F'\setminus F_t$\\
\quad $V' \leftarrow V'\setminus V_t$\\
\quad $E' \leftarrow E'\setminus E_t$\\
\quad $J_t \leftarrow(F',V',E')$\\
End While\\
$\TC\leftarrow t$\\
$\GC\leftarrow G'$\\
Return $ ( \GC, \TC, (F_t)_{t=1}^{\TC},
(V_t)_{t=1}^{\TC}, (J_t)_{t=1}^{\TC}  )$\\
\hline
\end{tabular*}
\end{table}
%

\subsection{Construction of the sparse basis}
\label{subsec:sparse_basis_const_periphery}
The construction of a sparse basis, which is at the heart of
Theorem~\ref{thm:MainSparse}, is based on the following
algorithm,
formally stated in Table~\ref{table:sync_peeling}. The algorithm
constructs a sequence
of residual factor graphs $(J_t)_{t \geq0}$, starting with the
instance under
consideration $J_0=G$. At each step, the new graph is constructed by removing
all variable nodes of degree one or zero, their adjacent factor nodes,
and all
the edges adjacent to these factor nodes. We refer to the
algorithm as \emph{synchronous peeling} or simply \emph{peeling}.

We denote the sets of nodes and edges removed at step (or round) $t\ge
1$
by $(F_t,V_t,E_t)$, so that $J_{t-1} = (F_t,V_t,E_t)\cup J_t$.
Notice that, at each step, the residual graph $J_t$ is check-induced.
The algorithm
halts when the residual graph does not contain any variable node of degree
smaller than two. We let the total number of iterations be $\TC(G)$,
where we
will drop the explicit dependence on $G$ when it is clear from context. The
final residual graph is then $J_{\TC} \equiv\GC$. The following elementary
fact is used in several papers on this topic \cite{Luby98,Molloy,DemboFSS}.
%

\begin{remark}
\label{rem:peeling_leaves_2core}
The residual graph $\GC$ resulting at the end of
synchronous peeling is the $2$-core of $G$.
\end{remark}

It is convenient to reorder the factors (from $1$ to $m$) and variables
(from $1$ to $n$) as follows. We index the factors in increasing order
according to $F_1, F_2, \ldots, F_{\TC}$, choosing an arbitrary order within
each $F_t$ for $1 \leq t \leq\TC$.

For the variable nodes, we first index nodes in $V_1$, then nodes in
$V_2$ and
so on. Within each set $V_t$, the ordering is chosen in such a way that nodes
that have degree $0$ in $J_{t-1}$ have lower index than those with
degree $1$
[notice that, by definition, for any $v\in V_t$, $\deg_{J_{t-1}}(v)\le1$].
Finally, for variable nodes in $V_t$ that have degree $1$ in $J_{t-1}$,
we use
the following ordering. Each such node $v \in V_t$ is connected to a unique
factor node in $F_t$. Call this the \emph{associated factor}, and
denote it by
$f_v$. We order the nodes $\deg_{J_{t-1}}(v)=1$ according to the order
of their
associated factor, choosing an arbitrary internal order for variable nodes
with the same associated factor.

For $A\subseteq F$, $B\subseteq V$, we denote by $\bH_{A,B}$ the
submatrix of
$\bH$ consisting of rows with index $a\in A$ and columns $i\in B$. The
following structural lemma is immediate, and we omit its proof.

%
\begin{lemma}\label{lemma:structural}
Let $G$ be any factor graph [not necessarily in
$\G(n,k,m)$] with no $2$-core. With the order of factors
and variable nodes defined through synchronous peeling, the matrix $\bH
$ is
partitioned in $\TC\times\TC$ blocks\break $\{\bH_{F_s,V_t}\}_{1\le s\le
\TC, 1\le t\le\TC}$
with the following structure:
\begin{longlist}[1.]
\item[1.] For any $s> t$, $\bH_{F_s,V_t} =0$.
\item[2.] The diagonal blocks $\bH_{F_s,V_s}$, have a \emph{staircase}
structure, namely for each such block the columns can be
partitioned into consecutive groups $(\cC_l)_{l=0}^\ell$, for $\ell= |F_s|$,
such that columns in $\cC_0$ are equal to $0$, columns in $\cC_1$ have
only the first entry equal to $1$, columns in $\cC_2$ have only the
second entry equal to $1$, etc. See below for an example.
\end{longlist}
\end{lemma}

An example of a staircase matrix
%
%
\begin{equation}
\left[ %
\matrix{ 0 & 0 & 1 & 1 & 0 & 0 & 0 & 0
\cr
0 & 0 & 0 & 0 & 1 & 0 & 0 & 0
\cr
0 & 0 & 0 & 0 & 0 & 1 & 1 & 0
\cr
0 & 0 & 0 & 0 & 0 & 0 & 0 & 1 }
 \right].
\end{equation}
Note that $V_t$ is not empty and $F_t$ is not empty for all $t<\TC$.
On the other hand, $F_{\TC}$ may be empty, in which case, we adopt the
convention that all
columns corresponding to $V_{\TC}$ are included in $\cC_0$.

The above ordering reduces $\bH$ to an essentially upper triangular
matrix. It
is then immediate to construct a basis for its kernel. We will do this by
partitioning the set of variable nodes as the disjoint union $V= U\cup
W$ in
such a way that $U\in\{0,1\}^{m\times m}$ and $\bH_U$ is square with
full rank, and
$W\in\{0,1\}^{m\times(n-m)}$. We then treat $\ux_W$ as independent
variables and $\ux_U$ as dependent ones. The partition is then
constructed by
letting $W= W_1\cup\cdots\cup W_{\TC}$ and $U= U_1\cup\cdots\cup
U_{\TC}$,
whereby for each $t\in\{1,\ldots, \TC\}$, $W_t\subseteq V_t$ is
chosen by
considering the staircase structure of block $\bH_{F_t,V_t}$ and the
corresponding partition over columns $V_t=\cC_0\cup\cC_1\cup\cdots
\cup
\cC_{\ell}$. We let $W_t=\cC_0 \cup\cC'_1\cup\cdots\cup\cC
'_{\ell}$, where
$\cC'_i$ includes all the elements of $\cC_i$ except the first (and is
empty if
$|\cC_i|=1$). Finally, $U_t\equiv V_t\setminus W_t$. With these definitions,
$\bH_{F,U}$ is an $m\times m$ binary matrix with full rank. In
addition, it is
upper triangular with diagonal blocks $\bH_{F_t,U_t}=I_{|U_t|}$ for
$t=1,\ldots,\TC$, where $I_r$ is the $r \times r$ identity matrix.

In order to construct a sparse basis for the clusters when
$\alpha>\ad(k)$ [and hence prove Theorem~\ref{thm:MainSparse}, point
2(a)], we will have to consider matrices $\bH$ (without a
2-core) that contain rows with exactly $2$ nonzero entries (i.e., check
nodes of degree~$2$).
Whenever this happens, the construction must be modified, by
introducing the notion of \emph{collapsed graph}. The basic idea is
that a factor node of degree $2$
constrains the adjacent variables to be identical and hence we can
replace each set of variables
that are thus constrained to be equal
by a single proxy variable (a~``super-node''). This proxy variable node
will have an edge with each factor that
was previously connected to a replaced variable
node, with a small modification:
Since we are operating in $\GF(2)$, we retain a single
edge for edges with odd multiplicity, and drop edges with even multiplicity.
%

\begin{definition}\label{def:collapsed_graph}
The \emph{collapsed graph} $ G_* = (F_*,V_*,E_*)$ of a graph $G =
(F,V,E)$ is the graph of connected components in the subgraph induced
by factor
nodes of degree $2$. Formally,
\begin{eqnarray*}
F_* &\equiv& \bigl\{f \in F \dvtx|\nbr f| \geq3 \bigr\},
\\
V_* &\equiv& \bigl\{S \subseteq V \dvtx d_{G^{(2)}}(i,j) < \infty, \forall
i,j \in S\bigr\},
\\
E_* &\equiv& \bigl\{(S,a)\dvtx S \in V_*, a\in F_*, \bigl|\bigl\{ i\in S \mbox{ s.t. }(i,a)
\in E\bigr\}\bigr| \mbox{ is odd}\bigr\},
\end{eqnarray*}
where $G^{(2)}$ is the subgraph of $G$ induced by factor nodes of degree
$2$. We
let $n_* \equiv|V_*|$, $m_* \equiv|F_*|$. An element of $V_*$ is
referred to
as a \emph{supernode}.
\end{definition}

Note that for a graph $G$ with no 2-core, the collapsed graph $G_*$
also has no
$2$-core. We let $\bQ$ denote the corresponding adjacency matrix of $G_*$.
Finally, we construct a binary matrix ${\mathbb L}$ with rows indexed
by $V$, and
columns indexed by $V_*$, and such that $L_{i,v}=1$ if and only if $i$ belongs
to connected component $v$. We apply peeling to $G_*$, thus obtaining the
decomposition of $V_*$ into $U_*\cup W_*$ as described for the original graph
$G$ above.

The following is the key deterministic lemma on the construction of the basis.
We denote the size of the component of $v\in V_*$ in $G^{(2)}$ by $S(v)$,
and for $v\in V_*$,
$t\ge0$ we let $S(v,t) = \sum_{w\in\Ball_{G_*}(v,t)} S(w)$ be the
sum of
sizes of vertices within distance $t$ from $v$.

%
\begin{lemma}\label{lemmabasic}
Assume that $G_*$ has no $2$-core, then the columns of
\[
{\mathbb L}\left[ %
\matrix{(
\bQ_{F_*,U_*})^{-1}\bQ_{F_*,W_*}
\vspace*{2pt}\cr
I_{(n_*-m_*) \times(n_*-m_*)} }
 \right] %
\]
form an $s$-sparse basis of the kernel of $\bH$, with $s = \max_{v_*\in
V_*}S(v_*,\TC(G_*))$. Here, we have ordered the super-nodes $v_* \in V_*$
as $U_*$ followed by $W_*$,
and the matrix inverse is taken over $\GF[2]$.
\end{lemma}

The proof of Lemma~\ref{lemmabasic} is
presented in the Appendix~\ref{app:Sparse_and_basis}.

%
%
\subsection{Construction of the cluster decomposition}
\label{sec:ClusterConstruction}

When $G$ does not contain a $2$-core [which happens w.h.p.for
$\alpha<\alpha_{\mathrm{d}}(k)$]
the above lemma is sufficient to characterize the space of solutions
$\cS$.
When $G$ contains a $2$-core [w.h.p. for $\alpha>\ad(k)$]
we need to construct the partition of the space of solutions
$\cS_1\cup\cdots\cup\cS_N$.

We let $\GC=(\FC
,\VC,\EC)$ denote the $2$-core of $G$, and $P_G\dvtx\{0,1\}^V\to
\{0,1\}^{\VC}$ be the projector that maps a vector $\ux$ to its restriction
$\ux_{\VC}$. Next, we let $\Hcore\equiv\bH_{\FC,\VC}$ be the
restriction of $\bH$ to the $2$-core, and denote its kernel by $\Score$.
Obviously, for any $\ux\in\cS$, we have $P_G\ux\in\Score$. Further,
%
%
\begin{equation}
\cS= \bigcup_{\uxC\in\Score} \cS(\uxC), \qquad\cS(\uxC)
\equiv \{\ux\in\cS\dvtx P_G\ux= \uxC \},
\end{equation}
with $\{\cS(\uxC)\}_{\uxC\in\Score}$ forming a partition of $\cS$.

It is easy to check $\bH_{F\setminus\FC,V\setminus\VC}$ has full row
rank. For instance, this follows from the fact that the subgraph
induced by $({F\setminus\FC,V\setminus\VC})$ is annihilated by
peeling (cf.
Remark~\ref{rem:peeling_leaves_2core}). Thus,
$\cS(\uxC)$ is nonempty for all $\uxC\in\Score$, and the sets
$\cS(\uxC)$ are simply translations of each other.

It turns out that $\{\cS(\uxC)\}_{\uxC\in\Score}$ is not exactly
the partition of $\cS$ that we seek. In our next lemma, we show that
the set of solutions
of the core $\Score$ can be partitioned in well-separated
core-clusters. Moreover, the core-clusters are
small and have a high conductance. We will form sets in our partition
of $\cS$
by taking the union of $\cS(\uxC)$ over $\uxC$ that lie in a
particular core-cluster.

We write $\ux' \preceq\ux$ for binary vectors $\ux', \ux$ if $x'_i
\leq
x_i$ for all $i$. We write $\ux' \prec\ux$ if $\ux' \preceq\ux$ and
$\ux'
\neq\ux$. We need the following definition:
%
%
\begin{equation}
\Lcore(\ell) \equiv\bigl\{\ux\dvtx\ux\in \Score(G), d(\ux,\underline{0}) \leq
\ell, \nexists\ux' \in \Score(G) \setminus \{\underline{0}\} \mbox{
s.t. } \ux' \prec\ux\bigr\}.
\end{equation}
The set $\Lcore(\ell)$ consists of minimal nonzero solutions of the
$2$-core having
weight at most $\ell$. (Here, the support of a binary vector $\ux$ is the
subset of its coordinates that are nonzero, and the weight of $\ux$
is the size of its support.)

%
\begin{lemma}
\label{lemma:core_few_low_weight}
For any $\alpha\in(\ad(k),\as(k))$, there exists
$\ve= \ve(\alpha,k)>0$ such that the following holds.
Take any sequence $(s_n)_{n \geq1} $ such that
$\lim_{n\rightarrow\infty} s_n = \infty$ and $s_n \leq\ve n$.
Let $G\sim\G(n,k,\alpha n)$.
Then w.h.p., we have:
\textup{(i)}~$\Lcore(\ve n) = \Lcore(s_n)$; \textup{(ii)}~$ |\Lcore(\ve n)| < s_n$;
\textup{(iii)} For any $\ux, \ux' \in\Lcore(\ve n)$, we have
$\ux\wedge\ux' = \underline{0}$, where $\wedge$ denotes bitwise AND.
In other words, different elements
of $\Lcore(\ve n)$ have disjoint supports.
\end{lemma}

Lemma~\ref{lemma:core_few_low_weight} is proved in Section~\ref
{app:proof_of_core_separation_lemma}.

%
\begin{remark}\label{remark:core_groups}
Let
$
\Ev_n
$
be the event that points (i), (ii) and (iii) in Lemma~\ref
{lemma:core_few_low_weight} hold. Assume $\Ev_n$ and $s_n^2 < \ve
n$.
Let $\Sci{1}$ be the set of core solutions with weight less than
$\ve n$. Then $\Sci{1}$ forms a linear space over
$\GF(2)$ of dimension $|\Lcore(\ve n)|$, with $\Lcore(\ve n)$ being
a $s_n$-sparse basis for $\Sci{1}$. Moreover, every element of $\Sci
{1}$ is
$s_n^2$-sparse.
\end{remark}

Let
%
%
\begin{equation}
g \equiv2^{|\Lcore(\ve n)|}. \label{eq:g_defn}
\end{equation}
%
We partition the set $\Score$ of core solutions in
disjoint \emph{core-clusters}, as follows. For $\ux,\ux'\in\Score
$, we
write $\ux\simeq\ux'$ if
$\ux\oplus\ux'\in\spn(\Lcore(\ve n))$. It is immediate to see that
$\simeq$ is an equivalence relation. We define the core-clusters to be
the equivalence classes of $\simeq$. Obviously, the core clusters are
affine spaces that differ by a translation, each containing $g\le2^{s_n}$
solutions. Their number is to be denoted by $N$.
Denote the core-clusters by $\Sci{1}, \Sci{2},
\ldots, \Sci{N}$.
Note that for any $\ux, \ux' \in\Score$ belonging to different
core-clusters, we have $d(\ux, \ux') > n\ve$, that is,
the core-clusters are well separated.
We use the following partition of the solution space (including
noncore variables) $\cS$ into clusters, based on the core-clusters
defined above:
%
%
\begin{equation}
\label{eq:PartitionDefinition} %
\cS= \bigcup_{i=1}^N
\cS_i,\qquad \cS_i\equiv \{\ux\in\cS\dvtx P_G
\ux\in\Sci{i} \}.%
\end{equation}

A version of Lemma~\ref{lemma:core_few_low_weight} was
claimed in \cite{MezRicZec_XOR,CoccoXOR,MM09}. These papers capture the
essence of
the proof but miss some technical details, and make the
erroneous claim that, w.h.p. each pair of core solutions is separated
by Hamming distance $\Omega(n)$.

We next want to study the internal structure of clusters.
By linearity, it is sufficient to consider only one of them, say $\cS_1$,
which we can take to contain the origin~$\underline{0}$. For any $\ux
\in\cS_1$,
we have
$P_G\ux\in\Sci{1}=\spn(\Lcore(\ve n))$, and $\Lcore(\ve n)$ forms a
$s_n$-sparse basis for $\Sci{1}$, which coincides with the projection
of $\cS_1$ onto the core.
Consider the subset of solutions $\ux\in\cS$, such that $P_G\ux
=\uxC$ for
some $\uxC\in\Sci{1}$. The set of
variables that take the same value for all solutions in this set is
strictly larger than the
$2$-core. In order to capture this remark, we define the
\emph{backbone} (variables that are uniquely determined by the core
assignment) and \emph{periphery} (other variables) of a graph $G$.
%

\begin{definition}
\label{def:backboneaugmentation_backbone_periphery}
Define the \emph{backbone augmentation procedure} on $G$ with the
initial check induced
subgraph $G_{\mathrm b}^{(0)}$ as follows. Start with $G_{\mathrm b}^{(0)}$.
For any $t\ge
0$, pick
all check nodes which are not in $G_{\mathrm b}^{(t)}$ and have at most one
neighbor outside $G_{\mathrm b}^{(t)}$. Build $G_{\mathrm b}^{(t+1)}$
by adding
all these
check nodes
and their incident edges and neighbors to $G_{\mathrm b}^{(t)} $. If no such
check nodes
exist, terminate and output $G_{\mathrm b}= G_{\mathrm b}^{(t)}$.

The \emph{backbone} $\GB=(\FB,\VB,\EB)$ of a graph $G= (F, V, E)$ is
the output of backbone augmentation procedure on $G$ with the initial subgraph
$\GC$, the 2-core of the graph $G$.

The \emph{periphery} $\GP$ of a graph $G=(F,V,E)$ is the subgraph
induced by
the factor nodes $\FP= F\setminus\FB$ and variable nodes $\VP=
V\setminus\VB$
that are not in the backbone.\footnote{Notice that there may be a few
variables (w.h.p. at most a constant number) in the
periphery that also are uniquely determined by the core assignment.}
\end{definition}

We can now define our basis for $\cS_1$. This is formed by two
sets of vectors. The first set has a
vector corresponding to each element of $\Lcore(\ve n) $. For each
$\uxC\in
\Lcore(\ve n) $, we construct a sparse solution $\ux\in\cS_1$ such that
$P_G \ux= \uxC$ (Lemma~\ref{lemma:core_sparse_basis} below
guarantees the existence of such a vector, and bounds its sparsity). This
set of vectors forms a basis for the projection of $\cS_1$ onto the
backbone.

For the second set of vectors, let $\HP\equiv\bH_{\FP,\VP}$ be the matrix
corresponding to the periphery graph. We construct a sparse basis for
the kernel of the matrix~$\HP$, following the general procedure
described in Section~\ref{subsec:sparse_basis_const_periphery}.
Namely, we first collapse
the graph and then peel it to order the
nodes. Note that this second set of basis vectors vanishes on
the backbone variables. Lemma~\ref{lemmabasic} is used to
bound its sparsity.

The first set of vectors is characterized as below
(see Section~\ref{app:core_basis} 
for a proof).

%
\begin{lemma}\label{lemma:core_sparse_basis}
Consider any $\alpha\in(\ad(k),\as(k))$. Let $G$ be drawn uniformly
from $\G(n,k,m)$.
Take $\ve(\alpha, k)> 0$ from Lemma~\ref{lemma:core_few_low_weight},
and consider any sequence $(c_n)_{n \geq1}$ such that $\lim_{n
\rightarrow\infty}c_n = \infty$. Then, with high probability, the
following is true. For every $\uxC
\in\Lcore(\ve n) $, there exists $c_n$-sparse $\ux\in\cS_1$ such
that $P_G \ux= \uxC$.
\end{lemma}

%
%
\subsection{Analysis of the construction}
\label{sec:AnalysisConstruction}

The main challenge in proving Theorem~\ref{thm:MainSparse} is bounding
the sparsity of the bases
constructed (either for the full set
of solutions, when $G$ does not have a core, or for the cluster $\cS
_1$, when
$G$ has a core). This involves two type of estimates: the first one
uses Lemma~\ref{lemmabasic}, while the second is stated as
Lemma~\ref{lemma:core_sparse_basis}.
In the first estimate, we need to bound all the quantities involved in
the sparsity upper
bound: the number
of iterations $T$ after which peeling (on the collapsed graph $G_*$)
halts, and
the maximum size $\max_{v\in V_*}S(v,T)$ of any ball of radius $T$ in the
collapsed graph. In particular, we will show that, w.h.p., we have $T=
O(\log\log
n)$, and that $\max_{v \in V_*}S(v,T) \le(\log n)^{\const}$ w.h.p.,
which gives
sparsity $s\le(\log n)^\const$.

Proving these bounds turns out to be a relatively simpler task when $G$ does
not have a $2$-core, partly because the graph in question has no factor nodes
of degree~$2$, and thus the collapse procedure is not needed.
A second reason is that when $G$ has a $2$-core, we need to apply
Lemma~\ref{lemmabasic} to the
periphery subgraph as discussed above. Remarkably, the periphery graph
admits a
relatively explicit probabilistic characterization. We say that a graph is
\emph{peelable} if its core is empty, and hence the peeling procedure halts
with the empty graph. It turns out that, conditional on the degree distribution,
the periphery is uniformly random among all peelable graphs.

Such an explicit characterization is not available, however, when we consider
the subgraph obtained by removing the core (the periphery is obtained
by removing
the entire backbone). Nevertheless, the proof of Lemma~\ref
{lemma:core_sparse_basis}
requires the study of this more complex subgraph. We overcome this problem
by using tools from the theory of local weak convergence \cite
{BenjaminiSchramm,AldousSteele,AldousLyonsUnimodular}.


Given a graph $G=(F,V,E)$, its check-node degree profile
$R=(R_l)_{l\in\naturals}$ is a probability distribution such that, for
any $l\in\naturals$, $mR_l$ is the number of check nodes of degree $l$.
A degree profile $R$ can conveniently be represented by its generating
polynomial $R(x) \equiv\sum_{l\ge0} R_l x^l$. The derivative of this
polynomial is denoted by $R'(x)$. In particular, $R'(1) =
\sum_{l\ge0} l R_{l}$ is the average degree.

Given integers $m$, $n$ and a probability distribution $R=
(R_l)_{l\le k}$ over $\{0,1,\break \ldots, k\}$,
we denote by $\bD(n, R, m)$ the set of \emph{check-node-degree-constrained
graphs}, that is, the set of bipartite graph with $m$ labeled check
nodes, $n$ labeled variable
nodes and check node degree profile $R$.
As for the model $\G(n,k,m)$, we will write $G\sim\bD(n, R, m)$ to
denote a graph drawn uniformly at random from this set. Note that we
have restricted the checks to have
degree no more than $k$. Further, we will only be interested in cases
with $R_0=R_1=0$.

%
\begin{lemma}\label{lemma:UniformPeelable}
Let $G= (F, V, E)\sim\G(m,n,k)$ and let $\GP$ be its periphery. Suppose
that with positive
probability, $\GP$ has $n_{\mathrm p}$ variable nodes, $m_{\mathrm p}$ check
nodes, and
check degree
profile $R^{\mathrm p}$. Then, conditioned on $\GP\in\bD(n_{\mathrm p},
R^{\mathrm p},m_{\mathrm p})$, the periphery
$\GP$ is distributed uniformly over the set $\bD(n_{\mathrm p},
R^{\mathrm
p}, m_{\mathrm p}) \cap\cP
$, where $\cP$ is the set of peelable graphs.
\end{lemma}

There is a small technical issue here in that if $G' \in\bD
(n_{\mathrm
p}, R^{\mathrm p},
m_{\mathrm p})$,
then variable nodes in $G'$ have labels from $1$ to $n_{\mathrm p}$,
whereas $\GP
$ has
variable node labels that form a subset of $\{1, 2, \ldots, n\}$, and similarly
for check nodes. We adopt the convention that the variable and check
nodes in
$\GP$ are relabeled sequentially, respecting the original order, before
comparing
with elements of $\bD(n_{\mathrm p}, R^{\mathrm p}, m_{\mathrm p})$.

The above lemma establishes that the periphery is roughly uniform,
conditional on being peelable. Its proof is in Section~\ref
{subsec:perip_conditional_random}.

Conceptually, we will bound the sparsity, as estimated in Lemma~\ref
{lemmabasic} by proceeding in three steps:
$(1)$ Bound the estimated basis
sparsity $\max_{v}S(v,T)$ for \emph{check node degree constrained
graphs} $\bD(n, R, m)$, in terms of the degree distribution;
$(2)$ Estimate the ``typical'' degree distribution for the periphery, and
prove concentration around this
estimate; $(3)$ Prove that, if $R$ is close to the typical degree
distribution, then $G\sim\bD(n, R, m)$ is peelable with uniformly
positive probability. The latter allows us to transfer the sparsity
estimates from the uniform model $\bD(n, R, m)$ to the actual
distribution of the periphery.

Lemma~\ref{lemma:peelability_implies_good} below accomplishes
steps $(1)$ and $(3)$, while Lemma~\ref{lemma:periphery_peelable} takes
care of
step $(2)$.
In order to state these lemmas, it is convenient to introduce density
evolution (the terminology comes from the analysis of sparse graph codes
\mbox{\cite{Luby98,Luby01,RiUBOOK}}).
%

\begin{definition}
Given $\alpha>0$, a degree profile $R$, and an initial condition
$z_{0}\in[0,1]$, we define
the \emph{density evolution sequence} $\{z_t\}_{t\ge0}$ by letting for
any $t\ge1$,
%
%
\begin{equation}
z_{t} = 1 - \exp \bigl\{-\alpha R'(z_{t-1})
\bigr\}. \label{eq:density_evolution}
\end{equation}
Whenever not specified, the initial condition will be assumed to be
$z_0=1$. The one-dimensional recursion (\ref{eq:density_evolution}) will
be also called \emph{density evolution recursion}.

We say the pair $(\alpha, R)$ is \emph{peelable at rate $\eta$} for
$\eta>0$ if $z_t \leq(1-\eta)^t/\eta$ for all $t \geq0$. We say that
the
pair $(\alpha, R)$ is \emph{exponentially peelable}
(for short \emph{peelable}) if there exists $\eta>0$ such that it is
peelable at rate $\eta$.
\end{definition}

The density evolution recursion \eqref{eq:density_evolution} describes
the large graph asymptotics of a certain belief propagation algorithm
that captures the peeling
process, and will be described Section~\ref{sec:BP_DE}.

The next lemma is proved in Section~\ref{sec:collapse_peeling_fast}.

%
\begin{lemma}\label{lemma:peelability_implies_good}
Consider the set $\bD(n, R, \alpha n)$, where $R=(R_l)_{l\le k}$ is a check
degree profile such that $R_0=R_1=0$. Assume that the pair $(\alpha,
R)$ is
peelable at rate $\eta$. Then there exist constants $N_0 = N_0(\eta
,k) <
\infty$, $\delta= \delta(\eta,k)>0$, $\const_1=\const_1(\eta
,k)<\infty$,
$\const_2=\const_2(\eta,k)<\infty$ such that the following hold, for
$G$ a
random graph drawn from $\bD(n,R, m)$ with $n > N_0$:
\begin{longlist}[(iii)]
\item[(i)] The graph $G$ is peelable with probability at least $\delta$. Further,
if $R_2=0$, one can take $\delta$ arbitrary close to $1$ (in other
words $G$ is
peelable w.h.p.).

\item[(ii)] Conditional on $G$ being peelable, peeling on the collapsed graph
$G_*$ terminates after $T\le\const_1 \log\log n$ iterations, with
probability at least $1-n^{-1/2}$.

\item[(iii)] Letting $T_{\mathrm {ub}}=\lfloor\const_1 \log\log n \rfloor
$, we have
$\max_{v\in V_*}S(v, T_{\mathrm{ub}})\le(\log n)^{\const_2}$, with
probability at
least $1-n^{-1/2}$.
\end{longlist}
\end{lemma}

Our final lemma is proved in Section~\ref{subsec:perip_exp_peelable}
and establishes the
peelability condition for the periphery.

%
\begin{lemma}\label{lemma:periphery_peelable}
For any $\alpha> \ad$ there exist constants $\eta=\eta(k,\alpha)>0$,
$\gamma_*=\gamma_*(k,\alpha)>0$ such that the following holds. Let $G=
(F, V, E)$ be a graph drawn
uniformly at random from the ensemble
$\G(n,k,m)$, $m=n\alpha$, and
let $\GP= (\FP, \VP, \EP)$ be its periphery. Let $\mP\equiv|\FP|$,
$\nP\equiv|\VP|$, $\aP\equiv\mP/\nP$ and denote by $\RP$ the
random check
degree profile of $\GP$. Then, for any
$\ve>0$, w.h.p. we have: \textup{(i)} The pair $(\aP,\RP)$ is peelable at rate
$\eta$; \textup{(ii)}
$n(\gamma_*-\ve)\le n_{\mathrm p}\le n(\gamma_*+\ve)$.
\end{lemma}

%
%
\subsection{Putting everything together}
\label{sec:Together}

At this point, we can formally summarize the proof of our main result,
Theorem~\ref{thm:MainSparse}, that builds on the construction and analysis
provided so
far.

\begin{pf*}{Proof of Theorem~\ref{thm:MainSparse}}
1. For $\alpha<\ad(k)$, w.h.p., the graph $G$ does not
contain a $2$-core (cf. Theorem~\ref{lemma:BasicCore}), hence peeling
returns an
empty graph. Using the construction in Lemma~\ref
{lemmabasic}, we
obtain an $s$-sparse basis, with $s = \max_{v\in V}|\Ball(v,\TC)|$
(notice that
in this case there is no factor node of degree $2$, and hence the collapsed
graph coincides with the original graph). The number of peeling
iterations $\TC$
is bounded by Lemma~\ref{lemma:peelability_implies_good}(ii), using
the fact
that, by definition of $\ad(k)$ the pair $(\alpha,R)$, with $R_k=1$
is peelable
at rate $\eta= \eta(\alpha,k)>0$ for $\alpha< \ad(k)$. Hence, $\TC
\le
\const_1 \log\log n$ w.h.p., for some $\const_1 = \const_1(\alpha,k)<
\infty$.
Finally, by applying Lemma~\ref{lemma:peelability_implies_good}(iii) we obtain
the thesis.

Next, consider point $2$.
The partition into clusters is constructed as per
equation~(\ref{eq:PartitionDefinition}), and in particular the number
of clusters
$N$ is equal to the number of solutions of the core linear system
$\Hcore\ux
=\underline{0}$ divided by $g$ given by equation~\eqref{eq:g_defn}.
Let us
consider the
various claims concerning this partition:

2(a). By construction, it is sufficient to construct a
basis of the
cluster $\cS_1$ containing the origin, cf. Section~\ref
{sec:ClusterConstruction}. The basis has two sets of vectors.

The first set of vectors is given by Lemma~\ref
{lemma:core_sparse_basis}. Their projection onto the core
spans the core solutions in $\Sci{1}$. Since variables in the backbone
are uniquely determined by those on the core, their projection onto
the backbone spans the backbone projection of $\cS_1$.
By Lemma~\ref{lemma:core_sparse_basis}, these vectors are, w.h.p., $c_n$-sparse
for any $c_n\to\infty$.
Lemma~\ref{lemmabasic} provides the second set of
vectors.
These span the
kernel of the adjacency matrix of the periphery, ${\mathbb H}_{\mathrm p}$
and vanish identically in the backbone. In particular, they are
independent from the first set. It is easy to check that the two sets
of vectors together form a basis for the
cluster $\cS_1$.

We are left with the task of proving that the second set of basis
vectors is sparse.
The construction in Lemma~\ref{lemmabasic} proceeds by
collapsing the
periphery graph $\GP$, and applying peeling. We thus
need to bound the sparsity $s = \max_{v\in V}S(v,\TC)$.
Define the event (implicitly indexed by $n$)
\[
\Ev_1 \equiv \bigl\{\bigl(\aP, \RP\bigr)\mbox{ is peelable at rate }
\eta>0\mbox{ and }\nP\geq n\gamma_*/2 \bigr\}.
\]
By Lemma~\ref{lemma:periphery_peelable}, we know
that $\Ev_1$ holds with high probability for suitable choices of
$\eta= \eta(k, \alpha)>0$ and $\gamma_*= \gamma_*(\alpha,k)>0$.
Further $\RP_0 = \RP_1 =0$ with probability~$1$.

From Lemma~\ref{lemma:UniformPeelable}, we know
that $\GP$ is drawn uniformly from the set $\D(\nP, \RP, \mP)\cap
\cP$.
Let $G'$ be drawn uniformly from $\D(\nP, \RP, \mP)$, with $(\nP,
\RP,
\mP)$ distributed as for $\GP$, conditional on $(\alpha_{\mathrm
p},R^{\mathrm p})\in\Ev_1$.
We can then apply Lemma~\ref{lemma:peelability_implies_good} to $G'$.
From point
(i), it follows that $G'$ is peelable with probability at least
$\delta=
\delta(\alpha,k)>0$.
Let $G_*'$ be the result of collapsing $G'$. From points (ii) and
(iii) it follows that, with probability at least $1-
\nP^{-0.5} \geq1 - (n\gamma_*/2)^{-0.5} \rightarrow1$ as $n
\rightarrow
\infty$, we have $\max_{v \in V_*'} S_{G'}(v, \TC) \le\max_{v
\in V_*'} S_{G'}(v, T_{\mathrm{ ub}})\le(\log n)^{\const}$,
for some $\const= \const(\alpha,k)< \infty$. (We use the subscript on
$S$ to indicate the graph
under consideration.)

Since $\Ev_1$ holds for $\GP$ w.h.p., and since
$G'$ is peelable with probability uniformly bounded away from zero, it
follows that the same
bound on the sparsity holds for $\GP$ as well.
In other words, w.h.p., we have that
\[
\max_{v \in\VPstar} S_{\GP}(v, \TC) = (\log
n)^{\const}.
\]
Here, $\VPstar$ is the set of super-nodes resulting from the collapse
of $\GP$.
Finally, using Lemma~\ref{lemmabasic}, we deduce that the second
set of basis vectors obtained from this construction is $s$-sparse for
$s =
(\log n)^{\const}$.

2(b).
By Lemma~\ref{lemma:core_few_low_weight}, w.h.p.,
for any two core solutions $\uxC\in\Sci{1}$, $\uxC'\in
\Sci{b}$, $b\neq1$ we have $d(\uxC, \uxC') \geq n
\ve$. This immediately implies $d(\ux, \ux') \geq n \ve$,
for any two solutions $\ux\in\cS_1$, $\ux'\in\cS\setminus\cS
_1$. By
linearity, we conclude $d(\cS_a, \cS_b) \geq n \ve$ for all $a,b$.

2(c). Let $\NC$ be the number of solutions of the core linear
system $\Hcore\ux=0$. This was proved to concentrate on the
exponential scale
in \cite{DuboisFOCS,Cuckoo}, with $n(\Sigma-\ve)\le\log\NC\le
n(\Sigma+\ve)$ with high probability, and $\Sigma$ given as in the
statement (cf. also \cite{MM09}). The number of clusters is $N = \NC/g$
for $g =
2^{\Lcore(\ve n)}$, cf. equation~\eqref{eq:g_defn}. Using the bound
$|\Lcore
(\ve n)|
\leq s_n$ from Lemma~\ref{lemma:core_few_low_weight}(ii) and
choosing $s_n$ to diverge sufficiently slowly with $n$, we
deduce that $N$ also concentrates on the exponential scale with the same
exponent as $\NC$.
\end{pf*}

%
%
\section{A belief propagation algorithm and density evolution}
\label{sec:BP_DE}

A useful analysis tool is provided by a belief propagation algorithm
[cf. equations~\eqref{eq:BP1} and \eqref{eq:BP2}] that refines the
peeling algorithm introduced in Section~\ref
{subsec:sparse_basis_const_periphery}.
The same algorithm is also of interest in iterative coding; see \cite
{RiUBOOK,MM09}.

We restate the BP update rules for the convenience of the reader.
\[
\vtoc{v} {a}^{t} = \cases{ %
*,&\quad
$\mbox{if $\ctov{b} {v}^{t-1} = *$ for all $b \in\partial v \setminus a$,}$
\vspace*{2pt}\cr
0, &\quad $\mbox{otherwise,}$}
\]
and
\[
\ctov{a} {v}^t = \cases{ %
0, &\quad $\mbox{if
$\vtoc{u} {a}^t = 0$ for all $u \in\partial a \setminus v$,}$
\vspace*{2pt}\cr
*, & \quad$\mbox{otherwise.}$}
\]
The initialization at $t=0$ depends on the context, but it is
convenient to single out two special cases. In the first case, all
messages are initialized to $0$: $\vtoc{v}{a}^0=\ctov{a}{v}^0=0$ for
all $(a,v)\in E$. In the second, they are all initialized to $*$:
$\vtoc{v}{a}^0=\ctov{a}{v}^0=*$ for all $(a,v)\in E$.
We will refer to
these two cases (resp.) as \BPzero\ and \BPstar.
We let $\vtcv^t \equiv( \vtoc{v}{a}^t)_{(a,v)\in E}$ and
$\ctvv^t \equiv( \ctov{v}{a}^t)_{(a,v)\in E}$ denote the vector of messages.

We mention here that \BPstar\ on the a graph $G \in\G(n,k,m)$ turns
out to be trivial (all messages remain $*$). However, we find it
useful to run \BPstar\ on the subgraph induced by variable and check
nodes outside the core. We describe this in detail in Section~\ref{sec:BPFP}.

The belief propagation algorithm introduced here enjoys an important
monotonicity property.
More precisely, define a partial ordering between message vectors by
letting $0\succ*$ and
$\vtcv\succeq\vtcv'$ if $\vtoc{v}{a}\succeq\vtoc{v}{a}'$ and
$\ctov
{a}{v}\succeq\ctov{a}{v}$ for all
$(a,v)\in E$.

%
\begin{lemma}[(\cite{RiUBOOK,MM09})]\label{lemma:Monotonicity}
Given two states $\vtcv^t_1\succeq\vtcv^t_2$,
we have $\vtcv^{t'}_1\succeq\vtcv^{t'}_2$ and $\ctvv^{t'}_1\succeq
\ctvv^{t'}_2$ at all $t'\ge t$.

As a consequence, the iteration \BPzero\ is monotone decreasing
(i.e., $\vtcv^{t+1}\preceq\vtcv^t$) and \BPstar\ is monotone
increasing (i.e., $\vtcv^{t+1}\succeq\vtcv^t$). In particular, both
converge to a fixed point in at most $|E|$ iterations.
\end{lemma}

It is not hard to check by induction over $t$ that \BPzero\ corresponds
closely to the peeling process.

%
\begin{lemma}\label{lemma:PeelingBP}
A variable node $v$ is eliminated in round $t$ of peeling, that is, $v
\in V_t$,
if there is at
most one incoming $0$ message to
$v$ in iteration $t-1$ of \BPzero\ but this was not true in previous
rounds. A factor node
$a$ is eliminated in round $t$ of peeling (i.e., $a \in F_t$), along
with all
its incident edges, if it receives a $*$ message for the first time in iteration
$t$ of \BPzero.
\end{lemma}

Further, the fixed point of \BPzero\ captures the decomposition of $G$
into core, backbone and periphery as follows.

%
\begin{lemma}\label{lemma:PeelingBPFP}\label
{lemma:BP_gives_core_backbone_periphery}
Let $(\vtcv^{\infty}$, $\ctvv^{\infty})$ denote the fixed point of
\BPzero.
For $v \in V$, we have:
\begin{itemize}
\item$v \in\VC$ if and only if $v$ receives two or more incoming $0$
messages under
$\ul{\wh{\nu}}^\infty$,
\item$v \in\VB\setminus\VC$ if and only if $v$ receives exactly one
incoming $0$ message under
$\ul{\wh{\nu}}^\infty$,
\item$v \in\VP$ if and only if $v$ receives no incoming $0$ messages under
$\ul{\wh{\nu}}^\infty$.
\end{itemize}
For $a \in F$, we have
\begin{itemize}
\item$a \in\FC$ if and only if $a$ receives no incoming $*$ message
under
$\ul{{\nu}}^\infty$,
\item$a \in\FB\setminus\FC$ if and only if $a$ receives one
incoming $*$
message under
$\ul{{\nu}}^\infty$,
\item$a \in\FP$ if and only if $a$ receives two or more incoming $*$
messages under
$\ul{{\nu}}^\infty$.
\end{itemize}
Finally, $\GC$ is the subgraph induced by $(\FC, \VC)$ and similarly
for $\GB$
and $\GP$.
\end{lemma}

The proofs of the last two lemmas are based on a straightforward
case-by-case analysis,
and we omit them. (In fact, this
correspondence is well known in iterative coding, albeit in a somewhat
different language \cite{RiUBOOK}.)
%
%
\subsection{Density evolution}
\label{subsec:density_evolution}

It turns out that distribution of BP messages is closely
tracked by density evolution, in the large graph limit.
Before stating this fact formally, it is useful to
introduce a different ensemble $\C(n,R,m)$ that will be used in some
of the proofs.
A graph $G$ in $\C(n,R,m)$ is constructed as follows. We label
variable nodes
$1$ through $n$ and check nodes $1$ through $m$. We choose an arbitrary
partition of the $m$ check nodes into $k+1$ sets with the $l$th set consisting
of $mR_l$ check nodes with degree $l$ each, for $l=0,1, \ldots, k$.
For each
check node of degree $l$, we draw $l$ half-edges distinct from each
other. Each
of these half-edges is connected to an arbitrary variable node.

There is a close relationship between the sets $\D(n,R,m)$ and $\C
(n,R,m)$. Any
element of $\D(n,R,m)$ corresponds to $\prod_{l=2}^k(l!)^{mR_l}$
elements of $\C(n,R,m)$, with
the ambiguity arising due to the ordering of the neighborhood of a
check node
in $\C(n,R,m)$. Conversely, any element of $\C(n,R,m)$ with no double edges
[two or more edges between the same (variable, check) pair] corresponds
to a
unique element of $\D(n,R,m)$. Moreover,
the fraction of elements of $\C(n,R,m)$ that have no double edges is
uniformly bounded away from zero as $n\to\infty$ \cite{BollobasConf}.
This leads to Lemma~\ref{lemma:Conf_implies_Uniform} below.

%
\begin{lemma}\label{lemma:Conf_implies_Uniform}
Let $\Ev$ be a graph property that does not depend on edge labels
[e.g., $\Ev(G) \equiv\{
G \mbox{ is a tree}\}$]. There exists
$\const= \const(k, \alpha_{\mathrm{ max}}) < \infty$ such that the
following is true for any $\alpha\in[0,\alpha_{\mathrm{ max}}]$. Suppose
$\Ev$
holds with probability $1-\eps$ for $G$ drawn uniformly at random from
$\C(n,R,\alpha n)$, for some $ \eps\in[0,1]$. Then $\Ev$ holds with
probability at least $1-\const\eps$ for $G'$ drawn uniformly at
random from
$\D(n,R,\alpha n)$.
\end{lemma}

An important tool in the following will be the notion of almost sure local
convergence of graph sequences. We made this notion precise in
Definition~\ref{def:WC},
following~\cite{DemboMontanariBrazil}.

We now return to the distribution of BP messages and density evolution.

%
\begin{lemma}\label{lemma:DegreeDensityEvolution}
Let $\{z_t\}$ be the density evolution sequence defined by
(\ref{eq:density_evolution}), for a given polynomial $R$,
with $z_0=1$, and define $\hz_t\equiv R'(z_t)/R'(1)$.
Assume $G_n\sim\D(n,R,m)$ or $G_n\sim\C(n,R,m)$ with $m=n\alpha$.

Let
$R_{l_0,l_*}^{(t)}$ be the fraction of check nodes
receiving $l_0$ incoming $0$ messages and $l_*$ incoming $*$ messages after
$t$ iterations of \BPzero\ in $G_n$. Similarly, let $L_{l_0,l_*}^{(t)}$
the fraction of variable nodes
receiving $l_0$ incoming $0$ messages and $l_*$ incoming $*$ messages after
$t$ iterations of \BPzero.

%
Then for any fixed $t\ge0$,
the following occurs almost surely:
%
%
\begin{eqnarray}
\qquad \lim_{n \rightarrow\infty} R_{l_0,l_*}^{(t)} &=&
R_{l_0+l_*} \pmatrix{l_0+l_*
\cr
l_0}
z_t^{l_0}(1-z_t)^{l_*} \qquad\mbox{for }
l_0,l_*\in\{0, 1, \ldots, k\}\label{eq:DE_R},
\\
\lim_{n \rightarrow\infty} L_{l_0,l_*}^{(t)} &=& \prob \{
X_0=l_0,X_*=l_* \} \qquad \mbox{for all } l_0,l_*
\in\naturals, \label{eq:DE_L} %
\end{eqnarray}
where $X_0\sim\Po(R'(1)\alpha\hz_t)$, $X_*\sim\Po(R'(1)\alpha
(1-\hz_t))$
are two independent Poisson random variables.
\end{lemma}

%
\begin{pf}
Notice that both $\D(n,R,m)$ and $\C(n,R,m)$, $m=n\alpha$ converge
locally to unimodular bipartite trees. More precisely,
if rooted at random variable nodes, they converge to Galton--Watson
trees with root
offspring distribution $\Po(R'(1)\alpha)$ at variable nodes,
and equal to the size-biased version of $R$ at check nodes.
The proof of the analogous statement in the case of nonbipartite
graphs can be found in \cite{DemboMontanariBrazil}, Proposition~2.6.
It uses an explicit calculation to show that the empirical
distribution of local neighborhoods converges in expectation, and
a martingale concentration argument to verify the assumptions of
Borel--Cantelli, and hence deduce almost sure convergence. The same
proof extends---with minimal changes---to bipartite (factor) graphs.

Messages are local functions of the graph, hence their distribution
converges to the one on the limit tree. In particular, incoming
messages on the same node are asymptotically independent because they
depend on distinct subtrees. The message distribution can be
computed through a standard tree recursion (see \cite{RiUBOOK,MM09})
that coincides with the density evolution recursion (\ref
{eq:density_evolution}).
\end{pf}
%

Using the correspondence in Lemma~\ref{lemma:PeelingBP} between
\BPzero\
and the
peeling algorithm, we can use density evolution to track the peeling
algorithm.

%
\begin{lemma}\label{lemma:concentration_DE}
Given a factor graph $H$, let $n_1(H)$ denote the number of variable
nodes of degree $1$, and $n_{2+}(H)$ the number of variable nodes of
degree $2$ or larger in $H$. For $l\in\naturals$, let $m_l(H)$ be the
number of factor nodes of degree $l$ in $H$.

Consider synchronous peeling for $t \geq1$ rounds on a
graph $G\sim\D(n,R,\alpha n)$ or $G\sim\C(n,R,\alpha n)$, with
$R_0= R_1=0$, and let $J_t$ denote the residual graph after $t$ iterations.
Let $\omega\equiv\alpha R'(1)$.
Then for any $\delta>0$, there exists
$N_0 = N_0(\delta,k,t, \alpha)$ such that with probability at least $1-1/n^2$
%
%
\begin{eqnarray}
\biggl\llvert \frac{m_l(J_t)}{n} - \alpha R_l z_t^l
\biggr\rrvert &\leq&\delta
\nonumber
\\[-8pt]
\\[-8pt]
\eqntext{\mbox{for $l\in\{2, 3, \ldots, k\}$},}
\\
\biggl\llvert \frac{n_1(J_t)}{n} - \omega\hat{z}_t \exp(-\omega
\hat{z}_t) \bigl(1-\exp\bigl(-\omega(\hat{z}_{t-1}-
\hat{z}_t)\bigr) \bigr) \biggr\rrvert &\leq&\delta,\label{eq:N1Estimate}
\\
\biggl\llvert \frac{n_{2+}(J_t)}{n} - 1+\exp(-\omega\hat{z}_{t}) (1 +
\omega\hat{z}_{t} ) \biggr\rrvert &\leq&\delta.
\end{eqnarray}
\end{lemma}

%
\begin{pf}
For the sake of simplicity, let us consider $n_{1}(J_t)$.
By Lemma~\ref{lemma:PeelingBP}, a~node $v$ has degree
$1$ in the residual graph $J_t$ if and only if there is one incoming
$0$ message to $v$ at
time $t$, and there were two or more incoming $0$ messages to $v$ at time
$t-1$. By Lemma~\ref{lemma:DegreeDensityEvolution}, the number of
incoming $0$ messages to $v$ at time $t$
converges in distribution to $Z_1 \sim\Po(\omega\hz_t)$. Using
monotonicity of the algorithm, and again Lemma~\ref
{lemma:DegreeDensityEvolution},
the number of incident edges such that the message incoming to $v$ at time
$t-1$ is $0$ but changes to $*$ at time $t$, converges to $Z_2 \sim
\Po(\omega(\hz_{t-1} - \hz_t))$, and is asymptotically independent of
the number of $0$ messages (converging to $Z_1$). Therefore,
$n_{1,t}/n$ converges as $n\to\infty$ to
\[
\prob[Z_1 = 1] \prob[Z_2 \geq1] = \omega
\hat{z}_t \exp(-\omega\hat{z}_t) \bigl(1-\exp\bigl(-\omega(
\hat{z}_{t-1}-\hat{z}_t)\bigr)\bigr).
\]

This establishes that the estimate (\ref{eq:N1Estimate}) holds with
high probability. In order to obtain the desired probability bound,
one can use a standard concentration of measure argument \cite
{RiUBOOK,PanconesiBook}.
Namely, we first condition on the degrees of the check nodes.
Since the unconditional distributions $\D(n,R,m)$ and $\C(n,R,m)$ are
recovered by a random
relabeling of the check nodes, such conditioning is irrelevant.
We then regard $n_1(J_t)$ as a function of the independent random
variables $X_1, \ldots, X_m$ whereby $X_a$ is the neighborhood of the
$a$th check node. We denote by $\Ev_n$ the event that all the balls
$\Ball_G(v,2t)$ of radius $t$ in $G$ have
size smaller than $(\log n)^C$. We have
\[
 \bigl|\E\bigl\{n_1(J_t)|X_1,
\ldots,X_{a-1},X_a;\Ev_n\bigr\}-\E\bigl\{
n_1(J_t)|X_1,\ldots ,X_{a-1},X_a';
\Ev_n\bigr\} \bigr|\le (\log n)^C. %
\]
The desired probability estimate then follows by applying Azuma's
inequality (in a form that allow for exceptional events; see, e.g., \cite{PanconesiBook}, Theorem~7.7) and bounding
$\prob(\Ev_n^{\mathrm c})$ (see, e.g., Section~\ref
{sec:PeelabilityGood}).
\end{pf}
%
%
\subsection{BP fixed points}
\label{sec:BPFP}

For our purposes, it is important to characterize the fixed point of the
\BPzero\ algorithm introduced above. Indeed, the structure of
this fixed point is directly related to the decomposition of $G$ into
core, backbone and periphery (cf. Lemma~\ref
{lemma:BP_gives_core_backbone_periphery}), which is in turn crucial for our
definition of clusters. Let us start from an easy remark on density
evolution.

%
\begin{lemma}\label{lemma:DEConvergence}
Let $\{z_t\}_{t\ge0}$ be the density evolution
sequence defined by equation~(\ref{eq:density_evolution}) with initial
condition $z_0=1$.
Then $t\mapsto z_t$ is monotone decreasing, and hence has a limit
$Q\equiv\lim_{t\to\infty} z_t$ which is given by
%
%
\begin{equation}
Q = \sup \bigl\{z \mbox{ s.t. } z = 1-\exp\bigl\{-\alpha
R'(z)\bigr\} \bigr\}. %
\end{equation}
\end{lemma}

%
\begin{pf}
Monotonicity follows from the fact that $z\mapsto
f(z) \equiv1-\break  \exp\{-\alpha R'(z)\} $ is monotone increasing, and that $z_1=
1-\exp\{-\alpha R'(1)\} <z_0$,
whence $z_2=f(z_1)\le f(z_0)=z_1$, and so on.
\end{pf}
%
Notice that the definition of $Q$ given in this lemma is consistent
with the one in Theorem~\ref{thm:Main}, that corresponds to the special
case of regular, degree-$k$ check nodes, that is, $R(x) = x^k$.
We further let $\hQ\equiv R'(Q) /R'(1)$.

We know that both \BPzero\ and density evolution converge
to a fixed point. Since density evolution tracks \BPzero\ for any
bounded number of iterations,
it would be tempting to conclude that a description of the \BPzero\
fixed point is
obtained by replacing $z_t$ by $Q$ and $\hz_t$ by $\hQ$ in Lemma~\ref
{lemma:DegreeDensityEvolution}. This is, of course, far from obvious
because it requires an inversion of the limits $n\to\infty$ and
$t\to\infty$. Despite this caveat, this substitution is essentially
correct.
%

\begin{lemma}\label{lemma:FP_concentration_DE}
Assume $G_n\sim\G(n,k,m)$ with $m=n\alpha$, and
$\alpha\in[0,\ad(k))\cup(\ad(k),\infty)$.

Let
$R_{l_0,l_*}^{(\infty)}$ be the fraction of check nodes
receiving $l_0$ incoming $0$ messages and $l_*$ incoming $*$ messages
at the fixed point of \BPzero. Similarly, let $L_{l_0,l_*}^{(\infty)}$
the fraction of variable nodes
receiving $l_0$ incoming $0$ messages and $l_*$ incoming $*$ messages
at the fixed point of \BPzero.

The following occurs with probability $1$:
%
%
\begin{eqnarray}
\lim_{n \rightarrow\infty} R_{l_0,l_*}^{(\infty)}& =&
\pmatrix{k
\cr
l_0} Q^{l_0}(1-Q)^{l_*} \qquad\mbox{for }
l_0\in\{0, 1, \ldots, k\}, l_* = k-l_0,\label{eq:REstFP}
\\
\lim_{n \rightarrow\infty} L_{l_0,l_*}^{(\infty)}& =& \prob \{
X_0=l_0,X_*=l_* \}\qquad \mbox{for all } l_0,l_*
\in\naturals,\label{eq:LEstFP} %
\end{eqnarray}
where $X_0\sim\Po(k\alpha\hQ)$, $X_*\sim\Po(k\alpha(1-\hQ))$
are two independent Poisson random variables.
\end{lemma}

Given Lemma~\ref{lemma:DegreeDensityEvolution} above,
Lemma~\ref{lemma:FP_concentration_DE} says that the messages change
very little beyond a large constant number of iterations.
A hint at the fact that Lemma~\ref{lemma:FP_concentration_DE} is
significantly more challenging than Lemma~\ref
{lemma:DegreeDensityEvolution} is given by the assumption in the
former that $\alpha\neq\ad(k)$. In fact, this turns out to be a
necessary assumption, because it implies an important correlation decay
property.

Molloy \cite{Molloy} established the analog of
equation~(\ref{eq:LEstFP}) for $\sum_{\ell_0\ge2,\ell_*\ge
0}L_{l_0,l_*}^{(\infty)} $, which corresponds to the relative size of
the core.
We find that the complete theorem presents new challenges: keeping
track of the backbone turns out to be hard. One hurdle is that the
``estimated backbone'' after $t$ iterations of \BPzero\ (i.e., the subset
of variable nodes that receive exactly one $0$ message) does not evolve
monotonically in $t$. In contrast, the ``estimated core''
(i.e., the subset
of variable nodes that receive two or more $0$ messages)
can only shrink. Another hurdle is that, unlike the periphery (cf.
Section~\ref{secperrr}), it turns out that the backbone is \emph
{not} uniformly random conditioned on the degree sequence.

The proof of Lemma~\ref{lemma:FP_concentration_DE} is quite long
and will be presented in Section~\ref{subsubsec:proof_FP_conc_DE}. The
basic idea is to run BP starting from the initialization with $0$
messages coming from vertices in the core and $*$ messages
everywhere else. This corresponds to \BPstar\ on the noncore $\GNC$
[i.e., the subgraph
induced by $(F\setminus\FC, V \setminus\VC)$], since messages
outside the noncore do not change: Messages within the core and from
core variables to noncore checks stay fixed to 0. Messages from noncore
checks to core variables stay fixed to *. We refer to this algorithm
simply as \BPstar, with the understanding that \BPstar\ is actually run
on $\GNC$.

It is not hard to check by induction over $t$ that \BPstar\ corresponds
to the backbone augmentation procedure.

%
\begin{lemma}\label{lemma:PeelingBPstar}
Consider the backbone augmentation procedure with the initial subgraph
$\GC$.
A factor node $a$ is added to the backbone in round $t$
of backbone augmentation, that is, $a \in G_{\mathrm b}^{(t)}\setminus
G_{\mathrm
b}^{(t-1)}$
(cf. Definition~\ref{def:backboneaugmentation_backbone_periphery}) if
all but one incoming message to $a$ in iteration $t$ of \BPstar\ are
$0$, but this was not the case in previous iterations.

A variable node $v$ is added to the backbone in round $t$,
of backbone augmentation, that is, $v \in G_{\mathrm b}^{(t)}\setminus
G_{\mathrm b}
^{(t-1)}$ if there is one incoming $0$ message to
$v$ in iteration $t$ of \BPstar\ but this was not true in previous iterations.
\end{lemma}

It then follows immediately from Lemma~\ref
{lemma:BP_gives_core_backbone_periphery} that \BPzero\ and \BPstar\
converge to the same fixed point. Denote the messages at this fixed
point by $\vtoc{v}{a}^{0,\infty}$.

Denote by $\vtoc{v}{a}^{*,t}$ the messages produced in iteration $t$ of
\BPstar, and $\vtoc{v}{a}^{0,t}$ the messages produced by
\BPzero. Monotonicity of BP update implies
$\vtoc{v}{a}^{0,t}\succeq\vtoc{v}{a}^{0,\infty}\succeq\vtoc{v}{a}^{*,t}$.
The proof consists in showing that the fraction of $0$ messages in
$\{\vtoc{v}{a}^{0,t}\}_{(a,v)\in E}$ is, for large fixed $t$, close to
the fraction of $0$ messages in $\{\vtoc{v}{a}^{*,t}\}_{(a,v)\in E}$.
The challenge is that no analog of Lemma~\ref
{lemma:DegreeDensityEvolution} is available for \BPstar.

Our final lemma is a straightforward consequence of Lemmas \ref
{lemma:DegreeDensityEvolution} and \ref{lemma:FP_concentration_DE} above.

%
\begin{lemma}\label{lemma:BP_few_late_changes}
Consider any $k \geq3$, any $\alpha\in(0, \ad) \cup(\ad, \as)$ and
any $\delta> 0$. There exists $T< \infty$ such that the following occurs.
Let $G_n\sim\G(n,k,\alpha n)$. Then, eventually (in $n$) almost
surely, the
fraction of (check-to-variable or
variable-to-check) messages that change after iteration $T$ of
\BPzero\ is smaller than $\delta$.
\end{lemma}

%
\begin{pf}
Let $N^{t}(n)$ be the
fraction of variable-to-check messages that are equal to $0$ after $t$
iterations
on $G_n$ (with $t=\infty$ corresponding to the fixed point). Then
equations~(\ref{eq:DE_R}) and (\ref{eq:REstFP}) imply that
\[
\bigl |N^{t}(n)- z_t \bigr|\le\frac{\delta}{3k},\qquad
\bigl|N^{\infty}(n)- Q \bigr|\le\frac{\delta}{3k} %
\]
holds eventually almost surely.
Using Lemma~\ref{lemma:DEConvergence}, there exists $T$ large
enough so that, for $t\ge T$, $|z^t-Q|\le\delta/(3k)$. By the triangle
inequality $|N^t(n)-N^{\infty}(n)|\le\delta/k$. The thesis for
variable-to-check messages follows
since, by monotonicity of \BPzero, $N^t(n)-N^{\infty}(n)$ is exactly
equal to the fraction of messages that change value from iteration $t$
to the fixed point. Each change in a variable-to-check message can lead
to a change
in at most $k-1$ check-to-variable messages. Thus, the fraction of
check-to-variable messages
that change after iteration $T$ is smaller than $\delta$.
\end{pf}
%
%
%
\subsection{Proof of Lemma \texorpdfstring{\protect\ref{lemma:FP_concentration_DE}}{4.8}}
\label{subsubsec:proof_FP_conc_DE}

Throughout this section, the notion of convergence adopted is
\emph{convergence locally} (cf. Definition~\ref{def:WC}).

For $n \geq0$, draw a graph $G_n$ uniformly at random from $\G(n,k,
\alpha n)$.
Consider equation~(\ref{eq:LEstFP}). Since the total number of incoming
messages is equal to the vertex degree, which is $\Po(k\alpha)$, it is
sufficient to control the distribution of $0$ incoming messages.
In particular, we define
\[
L^{(t)}_{\ell+} \equiv \sum
_{\ell_*=0}^{\infty}\sum_{l_0=\ell}^{\infty}
L^{(t)}_{l_0,l_*}, %
\]
that is the fraction of nodes that receive $\ell$ or more $0$ incoming
messages.

We prove a series of lemmas, leading to the desired estimate for
$L^{(t)}_{\ell+}$.

An upper bound on $L^{(\infty)}_{\ell+}$ is relatively easy to obtain.

%
\begin{lemma}
\label{lemma:L1inf_ub}
With probability $1$ with respect to the choice of $(G_n)_{n\geq0}$,
we have for all $l \geq0$,
\[
\lim\sup_{n \rightarrow\infty} L_{\ell+}^{(\infty)} \leq\prob
\bigl\{ \Po (k\alpha\hQ)\ge\ell\bigr\}.
\]
\end{lemma}

\begin{pf}
Using Lemma~\ref{lemma:DegreeDensityEvolution} (and using the fact that
$L_l \leq\const\exp(- l/\const)$ for all $l$ holds eventually almost
surely, for some $\const< \infty$) we have,
\[
\lim_{n \rightarrow\infty} L_{\ell+}^{(t)} = \prob\bigl\{
\Po(k\alpha \hz _t)\ge\ell\bigr\}
\]
holds w.p. $1$.
From Lemma~\ref{lemma:Monotonicity}, it follows that $L_{\ell+}^{(t)}$
is monotone decreasing. Thus, we have
\[
\lim\sup_{n \rightarrow\infty} L_{\ell+}^{(\infty)} = \prob\bigl
\{\Po (k\alpha\hz_t)\ge\ell\bigr\}
\]
w.p. $1$.

Fix an arbitrary $\delta> 0$.
Lemma~\ref{lemma:DEConvergence} implies that, for $t$ large
enough,
\[
\bigl[\prob\bigl\{\Po(k\alpha\hz_t)\ge\ell\bigr\} - \prob\bigl\{
\Po(k\alpha \hQ)\ge \ell\bigr\} \bigr] \leq\delta,
\]
which implies that
\[
\lim\sup_{n \rightarrow\infty} L_{\ell+}^{(\infty)} \leq\prob
\bigl\{ \Po (k\alpha\hQ)\ge\ell\bigr\} + \delta
\]
holds almost surely.
Since $\delta$ is arbitrary, we obtain the claimed result.
\end{pf}

The lower bound on $L_{\ell+}^{(\infty)}$ cannot be obtained by the
same approach. We go therefore through a detour.

Let $\mu_n \equiv\mu(G_n)$ be the measure on rooted factor graphs
with marks
(called ``networks'' in \cite{AldousLyonsUnimodular}), constructed as
follows: Choose
a uniformly random variable node $i\in V_n$ as root. Mark variable
nodes with mark $\mathsf{ c }$ if they are in the 2-core of~$G_n$.

%
\begin{lemma}\label{lemma:Gstar_convergesto_Tstar}
The sequence
$\{\mu_n\}_{n\ge0}$ converges locally
to the measure on random rooted tree with marks, $\T_*(\alpha,k)$,
defined as
follows. Construct a random bipartite Galton--Watson tree rooted at
$\root$ with offspring
distribution $\Po(k\alpha)$ at variable nodes and deterministic
$k-1$ at factor notes. 
Let $\VC(\T_*)$ be the maximal subset of its vertices
such that each variable node has degree at least $2$ and each factor
node has degree $k$ in the induced subgraph.
Mark with $\mathsf{ c }$ all vertices in $\VC(\T_*)$.
\end{lemma}

%
%
\begin{pf}
It is immediate to see that the sequence $\{\mu_n\}_{n\ge0}$ is
tight almost surely with respect to the choice of $(G_n)_{n\geq0}$,
that is, that for any $\ve\ge0$ there exists a compact set $\K$
such that $\prob\{H_*(n)\in\K\}\ge1-\ve$.
(E.g., take $\K$ to be the set of graphs that have maximum
degree $\Delta_t$ at distance $t$ for a suitable sequence $t\mapsto
\Delta_t$.)
Therefore \cite{AldousLyonsUnimodular},
any subsequence of $ \{\mu_n\}$ admits a further subsequence that converges
locally weakly to a limiting measure on rooted networks.
This subsequence can be constructed through a diagonal argument:
First, construct a subsequence $\{\mu_{n^t_s}\}_{s\ge0}$ such that the
depth-$t$ subtree
converges. Refine it to get a subsequence $\{\mu_{n^{t+1}_s}\}_{s\ge
0}$ such that the depth-$(t+1)$
subtree converges and so on. Finally, extract the diagonal subsequence
$\{\mu_{n^s_s}\}_{s\ge0}$.

We will prove the thesis by a standard weak convergence argument \cite
{Kallenberg}:
We will show that for any subsequence of $\{\mu_n)\}_{n\ge0}$, there
is a sub-subsequence that converges
locally weakly to the measure on $\T_*(\alpha,k)$.

Consider indeed any
sub-subsequence that converges locally weakly to limiting random rooted
graph with marks, which we denote by $\O_*$.
Define the \emph{unmarking} operator $\U$ that maps a marked rooted
graph to the corresponding unmarked rooted graph.
We have that $\U(\O_*)\ed\U(\T_*)$ (here $\ed$ denotes equality in
distribution) from local weak convergence of
random graphs to Galton--Watson trees (see, e.g.,
\mbox{\cite{AldousSteele,DemboMontanariBrazil}}). We will hereafter couple
the two trees in such a way that $\U(\O_*)= \U(\T_*)$.

Recall that a stopping set is
any subset of variable nodes of a factor
graph, such that each variable node has degree at least $2$ in the
induced subgraph.
The $2$-core of the factor graph is the maximal stopping set and is a
superset of any stopping set. These notions are well defined for
infinite graphs as well.

Now, the marks in $\T_*$ correspond to the core by definition. The
marks in $\O_*$ form a stopping set, since the measure on $\O_*$ is the
local weak
limit of $\mu_n$, and in any graph drawn from $\mu_n$, w.p. 1 a vertex
is marked only if at
least two of its neighboring checks have all marked neighboring
variable nodes. Moreover, one can show that both $\T_*$ and $\O_*$ are
unimodular. Indeed $\T_*$ is unimodular since the unmarked tree is
clearly unimodular, and the marking process does not make any
reference to the root. Unimodularity of $\O_*$ is clear since it is
the local weak limit of a marked random graph
\cite{AldousLyonsUnimodular}. Thus, in order to prove our thesis
it suffices to show that the density of marks is the same in $\T_*$
and $\O_*$.
(Because the subset of nodes that is marked in $\T_*$ contains the
subset marked in $\O_*$ and the density of their difference is equal to
the difference of the densities. Finally, for unimodular network, if
a mark type has density $0$, then the set of marked nodes is empty by
union bounds.)

Let
\[
\Ev\equiv \Bigl\{ \lim_{n \rightarrow\infty} \bigl|V_{\mathrm
c}(G_n)\bigr|/n
= \prob \bigl\{ \operatorname{Poisson}(k\alpha\hQ)\ge2 \bigr\} \Bigr\},
\]
where $Q$ and $\hQ$ are defined as at the beginning of Section~\ref
{sec:BP_DE}.
It was proved in~\cite{Molloy} that $|V_{\mathrm c}(G_n)|/n \stackrel
{\mathrm{a.s.}}{\longrightarrow} \prob \{ \operatorname
{Poisson}(k\alpha\hQ
)\ge
2 \}$, that is, the event $\Ev$ occurs with probability 1.
Now let the set of marked vertices in $\O_*$ be denoted by $\hVC(\O
_*)$. It is easy to see that if $\Ev$ holds, the density of marks in
$\O
_*$ is given by
%
%
\begin{equation}
\prob\bigl\{\root\in\hVC(\O_*)\bigr\} = \prob \bigl\{\operatorname
{Poisson}(k\alpha \hQ)\ge 2 \bigr\}.\label{eq:FiniteCore} %
\end{equation}

Proceeding analogously to the proof of \cite{BalPerPete06},
Proposition~1.2, we obtain
%
%
\begin{equation}
\prob\bigl\{\root\in\VC(\T_*)\bigr\} = \prob \bigl\{\operatorname
{Poisson}(k\alpha \hQ)\ge 2 \bigr\}. \label{eq:InfCore} %
\end{equation}
The sketch of this step is the following. Let $\Ev_t$ be the event
that $\root$ belongs to a ``depth $t$ core,'' where the requirement of
``degree at least 2 in the subgraph''
applies only to variables up to depth $t-1$.
The probability on the left-hand side is just $\prob\{\Ev\}$ for
$\Ev=\bigcap_{t\ge1}\Ev_t$. Since $\Ev_t$ is a decreasing sequence,
$\prob
\{\Ev\}=\lim_{t\to\infty}\prob\{\Ev_t\}$.
On the other hand, $\prob\{\Ev_t\}$ can be computed explicitly through
a tree calculation and converges to $\prob\{\Po(\alpha k \hQ)\ge2\}$
as $t
\rightarrow\infty$ yielding (\ref{eq:InfCore}).

Finally, the thesis follows by comparing equations~(\ref
{eq:FiniteCore}) and
(\ref{eq:InfCore}),
and recalling that $\prob(\Ev) =1$.
\end{pf}
%
We next construct a random tree $\tT_*(\alpha,k)$ with marks on the
directed edges as follows.
Marks take values in $\{0,*\}$ and to each undirected edge we
associate a
mark for each of the two directions. We will refer to the direction
toward the root as to the ``upward'' direction, and to the opposite
one as to the ``downward'' direction.
The marks correspond to fixed point BP messages,
and we will call them messages as well in what follows.
First, consider only edges directed upward.
This is a multitype GW tree. At the root generate
$\operatorname{Poisson}(k\alpha)$ offsprings, and mark each of the edges to $0$
independently with probability $\hQ$, and to $*$ otherwise.
At a nonroot variable node, if the parent edge is marked $0$, generate
$\Po(k\alpha(1-\hQ))$ descendant edges marked $*$ and $\Po_{\ge
1}(k\alpha\hQ)$ descendant edges
marked $0$ [here $\Po_{\Ev}(\lambda)$ denotes a Poisson random variable
with parameter $\lambda$ conditional to~$\Ev$]. If the parent edge is
marked $*$, generate
$\Po(k\alpha(1-\hQ))$ descendant edges marked $*$ and
no descendant edges marked $0$.
At a factor node, if the parent edge is marked $0$, generate $k-1$
descendant edges marked $0$. If the parent node is marked~$*$, generate
$M\sim\Binom_{\le k-2}(k-1,Q)$ descendants marked $0$, and $k-1-M$
descendants marked $*$.

For edges directed downward,
marks are generated recursively following the usual BP
rules, cf. equations~(\ref{eq:BP1}), (\ref{eq:BP2}), starting from
the top
to the bottom.
It is easy to check that with this construction,
the marks in $\tT_*(\alpha,k)$ correspond to a
BP fixed point.

We extend the unmarking operator $\U$ by allowing it to act on graphs
with marks on edges (and removing the marks).

%
\begin{lemma}
$\U(\tT_*)$ and $\U(\T_*)$ have the same distribution.
\end{lemma}
%
%
\begin{pf}
For this, we construct $\U(\tT_*)$ (which is $\tT_*$ without the marks
revealed) in a ``breadth first'' manner as follows: First, we draw a
$\Po(\alpha k)$ number of factor descendants for the root node. Let
$a$ be a factor descendant of the root. Then $a$ has $k-1$ variable
node descendants. The message $\nuai{a}{\root}$ is 0 with probability~$\hQ$. It immediate
to check from our construction and $\hQ=Q^{k-1}$ that:

\emph{Fact}~1: Conditional on the degree of the root
$\deg(\root)=d_1$, the $d_1(k-1)$ upward messages incoming to the check
nodes $a\in\droot$ are independent, with $\prob\{\nu_{v\to a}=0\}=Q$.

Now, we draw the number of descendants for each neighbor of $a$.
Using fact~1, together with the definition of $\tT$, one can check that:

\emph{Fact}~2: Conditional on the degree of the root
$\deg(\root)=d_1$, the number of descendants of each of the
$d_1(k-1)$ variable nodes $v$ at the first generation is an
independent $\Po(k\alpha)$ random variable. Further, the upward
messages toward these variable nodes are independent with
$\prob\{\hat{\nu}_{b\to v}=0\}=\hQ$.

This argument (outlined for simplicity for the first generation) can
be repeated almost verbatim at any generation. Denote by
$\tT_{*,d}$ the first $d$ generations of $\tT_{*,d}$ (with variable
nodes at the leaves). One then proves by induction that at any $d$,
conditional on $\U(\tT_{*,d})$, the number of descendants of
the variable nodes in the last generation are
i.i.d. $\Po(k\alpha)$,
and given these, the corresponding upward messages are i.i.d.
$\prob\{\hat{\nu}_{b\to v}=0\}=\hQ$.
This implies the thesis.
\end{pf}

%
\begin{lemma}
$\tT_*$ is unimodular.
\end{lemma}
\begin{pf}
We already established unimodularity of $\U(\tT_*)$ [since
$\U(\tT_*)=\U(\T_*)$ is a unimodular Galton--Watson tree]. To
establish the
claim, let $\tT_*'$ be the random tree whose distribution has
Radon--Nikodym derivative $\deg(\root)/\E\{\deg(\root)\}$ with respect
to that of $\tT_*$.
We need to show that moving the root to a uniformly random descendant
variable node of the root (via one check) in $\tT_*'$, leaves the
distribution of $\tT_*'$ unchanged
(cf. \cite{AldousLyonsUnimodular}, Section~4).

Draw $\tT_*'$ at random, weighted by the degree of the root
$\root$. In this argument, we make the root explicit by denoting the
tree by $(\tT_*', \root)$. Reveal the degree $d_1=\deg(\root)$ of
the root. We have $d_1>0$ almost surely. Take a uniformly random
neighboring check $a \in\partial\root$, and a uniformly random
descendant $i$ of $a$ (we know that $a$ has $k-1$
descendants). Reveal the number of descendants of $i$. Let this
number be $d_2-1$, so that $i$ has $d_2$ neighbors in total. Note
that we do not reveal any of the messages in $\tT_*'$. At this
point, consider the incoming messages to the variable nodes $\root$
and $i$ except for $\nuai{a}{\root}$ and $\nuai{a}{i}$, and the
incoming messages to the check $a$ except for $\nuia{\root}{a}$ and
$\nuia{i}{a}$. Call this vector of messages $M$. The messages in $M$
are independent, with probability $\hQ$ of for each incoming message
to variable nodes to be 0, and probability $Q$ for incoming messages
to $a$ to be 0.\footnote{The argument establishing this is essentially
the one above, where we showed that $\U(\tT_{*})=\U(\T_*)$.} The
messages $\nuai{a}{\root}$, $\nuai{a}{i}$, $\nuia{\root}{a}$ and
$\nuia{i}{a}$ are deterministic functions of~$M$. Finally, notice
that $d_1$ and $d_2$ are independent, and identically distributed as
$1+ \Po(\alpha k)$. At this point, it is clear that $(\tT_*',i)$ is
distributed identically to $(\tT_*', \root)$, which establishes unimodularity.
\end{pf}
%

%
\begin{lemma}\label{lemma:Tstar_tTstar_identical}
Let $\F$ be a map from ``trees with marked edges'' to ``trees with marked
variable nodes'' defined as follows: $\F(\T)$ is obtained from $\T$ by
putting a $\mathsf{ c }$
mark on vertex $i$ if and only if at least two incoming edges have a
$0$ mark.

Then $\F(\tT_*(\alpha,k))\ed\T_*(\alpha,k)$.
\end{lemma}

\begin{pf}
It is easy to check that the subset of variable nodes in
$\tT_*$ that receive two or more incoming $0$'s forms a stopping set
(since the set of messages is at a BP fixed point). But the density of
marked nodes in $\F(\tT_*)$ (i.e., the probability of the root being
marked) is $\prob \{ \operatorname{Poisson}(k\alpha\hQ)\ge
2 \}$,
which is
exactly the same as the density of marked nodes in $\T_*$ (recall that
$\T_*$ is also unimodular, cf. proof of Lemma~\ref
{lemma:Gstar_convergesto_Tstar}). On the other hand, the set of
marked nodes in
$\T_*$ is the core by definition
and hence includes the marked nodes in $\F(\tT_*)$. We deduce that
the set of vertices that are marked in $\T_*$ but not in $\F(\tT_*)$
has vanishing density and, therefore, $\F(\tT_*(\alpha,k))\ed\T
(\alpha,k)$.
\end{pf}

We let $B$ be the subset of variable nodes $v$
of $\tT_*(\alpha,k)$ such that at least one message
incoming to $v$ is equal to $0$.
Then this set has density
%
%
\begin{equation}
\prob\{\root\in B\} = \prob \bigl\{ \operatorname{Poisson}(k\alpha
\hQ)\ge 1 \bigr\} \equiv\hQ. %
\end{equation}
In light of Lemma~\ref{lemma:Tstar_tTstar_identical}, we further denote
the set of variable nodes in $\tT_*$ having two or more incoming $0$
messages by $\VC(\tT_*)$.

Consider running \BPstar\ on $U(\tT_*)$ [this is BP starting with zeros
from the variable nodes in $\VC(\tT_*)$ and $*$ elsewhere]. Let the
trees with marks on edges obtained after $t$ iterations be denoted by
$\tT_*^t$.

Denote by $\tmu^t_n$ the measure on the rooted factor graph with marks
on the edges constructed as follows: Choose a uniformly random variable
node $i \in V(G_n)$. Mark the edges (in each direction) with the
messages corresponding to \BPstar\ run for $t$ iterations.

%
\begin{lemma}\label{lemma:tmut_convergesto_tTt}
The measures $(\tmu^t_n)_{n \geq0}$ converge locally to the measure on
$\tT_*^t$.
\end{lemma}

\begin{pf}
This result is immediate from Lemmas \ref
{lemma:Gstar_convergesto_Tstar} and \ref{lemma:Tstar_tTstar_identical}.
\end{pf}

The following is immediate from the construction of $\tT_*$.
%

\begin{remark}
If $\root\in B$, then there exists a subtree of $\tT_*$ rooted at
$\root$
with the following properties: (i) If $j$ is a variable node
in the subtree, either $j\in\VC(\tT_*)$ or at least one descendant
factor node
is in the subtree; (ii) If $a$ is a factor node in the subtree, all
its descendants are also in the subtree.
\end{remark}

We call the subtree just defined a \emph{witness} for $\root$
(there might be more than one in principle). 
Notice that a priori a
witness can be finite [if it ends up with nodes in $\VC(\tT_*)$], or infinite.

%
\begin{lemma}\label{lemma:backbone_finite_witness}
Almost surely any node $i\in B$ has a finite witness. Thus, $\lim_{t
\rightarrow\infty} \tT_*^t = \tT_*$.
\end{lemma}

\begin{pf}
It is sufficient to prove that the following event has zero probability:
$\root\in B$ and $\root$ only has infinite witnesses.
Suppose $\root\in B$. We will look for a minimal witness for $\root$.
If $\root\in\VC(\tT_*)$, then it is itself a witness and we are done.
If not then, there is exactly one incoming $0$ message, say from factor
$a$. Then factor $a$ has $k-1$ incoming $0$ messages from descendants.
The subtrees corresponding to these descendants are independent.
Consider a descendant $i$ of $a$. We have
\begin{eqnarray*}
\prob\bigl(i \in B \setminus\VC(\tT_*)\bigr) &=& \prob \bigl\{
\Po_{\ge
1}(\alpha k\hQ)=1 \bigr\}
\\
&=&\exp(-\alpha k \hQ) \alpha k \hQ/\bigl(1-\exp(-\alpha k \hQ)\bigr)
\\
&=&\exp(-\alpha k \hQ) \alpha k Q^{k-2}.
\end{eqnarray*}
Conditioned on $i \in B \setminus\VC(\tT_*)$, the node $i$ has
exactly $k-1$ descendant variable nodes (via one check node). Thus,
conditioned on $\root\in B$, the minimal witness is a Galton--Watson
tree with
offspring distributed as $Z$, whereby
$Z=(k-1)$ with probability $\exp(-\alpha k \hQ) \alpha k Q^{k-2}$, and
$Z=0$ otherwise.
The branching factor of this tree is $\exp(-\alpha k \hQ) \alpha k
(k-1) Q^{k-2} < 1$
(cf. Lemma~\ref{lem:bdd_branch_fact} below). The lemma follows.
\end{pf}

%
\begin{lemma}\label{lemma:L1inf_lb}
Consider the setting of Lemma~\ref{lemma:FP_concentration_DE}.
We have
\[
\lim\inf_{n\to\infty} L_{1+}^{(\infty)} \geq\prob
\bigl\{\Po(k \alpha \hQ)\ge 1\bigr\},
\]
almost surely with respect to the choice of $G_n$.
\end{lemma}

%
\begin{pf}
Let $B_t$ be the subset of
variable nodes in $\tT_*^t$ that receive at least one~$0$ message.
Let $y_t$ be the density of nodes in $B_t$.
From Lemma~\ref{lemma:backbone_finite_witness}, we have immediately
%
%
\begin{equation}
\lim_{t\to\infty}y_t= \prob\bigl\{\Po(k\alpha\hQ)\ge1
\bigr\}. \label{eq:YTInfty}
\end{equation}
%

Let $B_t(n)\subseteq V(G_n)$ be the
subset of nodes having at least one incoming
$0$ after $t$ iterations of \BPstar.
Let $y_{t}(n)$ be the fraction of these nodes, that is, $y_t(n)\equiv
|B_t(n)|/n$. From Lemma~\ref{lemma:tmut_convergesto_tTt}, we have
%
%
\begin{equation}
\lim_{n\to\infty}y_{t}(n) = y_t
\end{equation}
almost surely.
%
By equation~(\ref{eq:YTInfty}), we have $\lim_{n\to\infty} y_{t}(n)
\ge
\prob\{\Po(k\alpha\hQ)\ge1\}-\delta$ for all $t\ge T(\delta)$.
By monotonicity of \BPstar, we have $\lim\inf_{n\to\infty}
L_{1+}^{(\infty)}\ge\lim_{n\to\infty}
y_t(n)\ge\prob\{\Po(k\alpha\hQ)\ge1\}-\delta
$, which implies the thesis.
\end{pf}
%

%
\begin{lemma}
Consider the setting \label{lemma:L2inf_lb}of Lemma~\ref{lemma:FP_concentration_DE}.
We have, for all
$\ell\ge2$,
\[
\lim\inf_{n\to\infty} L_{\ell+}^{(\infty)} \geq\prob
\bigl\{\Po(k \alpha\hQ )\ge\ell\bigr\},
\]
almost surely with respect to the choice of $G_n$.
\end{lemma}
%
%
\begin{pf}
The proof is very similar to that of the previous lemma.
Let $C(\ell;n)\subseteq V$ be the
subset of variable nodes in $G_n$ that are in the core and have at
least $\ell$
neighboring check nodes in the core. Then we have
(by monotonicity of \BPstar)
%
%
\begin{equation}
L^{(\infty)}_{\ell+}\ge\frac{|C(\ell;n)|}{n}. \label{eq:Ll1}
\end{equation}
On the other hand, let $y(\ell)$ be the density of variable nodes
in $\tT_*$ that receive two or more $0$ messages and have at least
$\ell$ neighboring check nodes in the set
\[
\{a\dvtx\mbox{For each $i \in\partial a$, node $i$ receives two or more $0$
messages} \}.
\]
It follows from Lemmas \ref{lemma:Gstar_convergesto_Tstar}
and~\ref{lemma:Tstar_tTstar_identical} that
%
%
\begin{equation}
\lim\inf_{n\to\infty}\frac{1}{n}\bigl |C(\ell;n)\bigr| = y(
\ell). \label{eq:Ll2} %
\end{equation}

On the other hand, it is easy to check that the construction of
$\tT_*$ implies that $y(\ell)$ coincides with the density of nodes
receiving $\ell$ or more $0$ messages (here the assumption $\ell\ge2$
is crucial). Hence, $y(\ell) = \prob\{\Po(k\alpha\hQ)\ge\ell\}$,
which together with equations~(\ref{eq:Ll1}), (\ref{eq:Ll2}) yields
the thesis.
\end{pf}
%

\begin{pf*}{Proof of Lemma~\ref{lemma:FP_concentration_DE}}
Equation~(\ref{eq:LEstFP}) follows from Lemmas \ref{lemma:L1inf_ub},
\ref{lemma:L1inf_lb} and~\ref{lemma:L2inf_lb}. Equation~(\ref
{eq:REstFP}) follows from a completely analogous argument.
\end{pf*}

Recall that $\tmu^t_n$ is the measure on the rooted factor graph with
marks on the edges constructed as follows: Choose a uniformly random
variable node $i \in V(G_n)$. Mark the edges (in each direction) with
the messages corresponding to \BPstar\ run for $t$ iterations. Recall
that $\tmu^*_n$ is defined similarly with marks corresponding to the BP
fixed point.
Denote by $\tmu^t_n(d)$, the measure obtained from $\tmu^t_n$ by
restricting the depth of the rooted graph to $d$.

%
\begin{lemma}\label{lemma:few_late_changes}
For any $d \geq0$ and any $\delta> 0$, there exists $t< \infty$ such
that almost surely,
\[
\lim\sup_{n \rightarrow\infty} \bigl\|\tmu_n^t(d) -
\tmu_n^*(d)\bigr\|_{\TV
} < \delta.
\]
\end{lemma}

\begin{pf}
Consider running \BPstar\ on $G_n$. From Lemma~\ref
{lemma:tmut_convergesto_tTt}, we know $\tmu_n^t$ converges locally to
the measure on $\tT_*^t$. From Lemma~\ref
{lemma:backbone_finite_witness}, we know $\lim_{t \rightarrow\infty}
\tT_*^{(t)} = \tT_*^\infty$. In particular, the fraction of $0$
variable-to-check messages in $\tT_*^{(t)}$ converges to $Q$ (i.e., the
fraction of $0$ variable-to-check messages in $\tT_*^{(\infty)}$). But
from Lemma~\ref{lemma:FP_concentration_DE}, the fraction of $0$
variable-to-check messages in $\tmu_n^*$ converges eventually almost
surely to the same value, and similarly for check-to-variable messages
the fraction of $0$ messages converges to $\hQ$ [using the fact that
$L_l \leq\const\exp(- l/\const)$ for all $l$ holds eventually almost
surely, for some $\const< \infty$]. Using monotonicity of \BPstar, we
deduce that for any $\eps> 0$, there exists $t$ large enough such that,
%
%
\begin{equation}
\lim\sup_{n \rightarrow\infty} \{\mbox{Number of message changes after
iteration $t$ in $G_n$}\}/n \leq\eps \label{eq:few_late_changes}
\end{equation}
holds almost surely.
Now, we can choose $\eps$ small enough such that eventually (in $n$)
almost surely, for any set of $\eps n$ edges in $G_n$, the union of
balls of radius $d$ around these edges contains no more than $\delta n$
nodes. Combining with equation~\eqref{eq:few_late_changes}, at least $(1-
\delta)$ fraction of nodes have all messages in a ball of radius $d$
unchanged after iteration $t$, almost surely. This yields the result.
\end{pf}

\begin{pf*}{Proof of Theorem~\ref{thm:BP_FP_localweaklimit}}
From Lemma~\ref{lemma:tmut_convergesto_tTt}, we know $\tmu_n^t$
converges locally to the measure on $\tT_*^t$. From Lemma~\ref
{lemma:backbone_finite_witness}, we know $\lim_{t \rightarrow\infty}
\tT_*^t = \tT_*^\infty$. Combining with Lemma~\ref
{lemma:few_late_changes}, we obtain that
%
%
\begin{equation}
\lim\sup_{n \rightarrow\infty} \bigl\|\tmu_n^*(d) - \mu\bigl(
\tT_*^\infty (d)\bigr)\bigr\| _{\TV} < \delta
\end{equation}
almost surely.
Since $\delta$ is arbitrary, we obtain, for every $d$, that
%
%
\begin{equation}
\lim\sup_{n \rightarrow\infty} \bigl\|\tmu_n^*(d) - \mu\bigl(
\tT_*^\infty (d)\bigr)\bigr\| _{\TV} =0
\end{equation}
holds almost surely. The result follows.
\end{pf*}

%
\section{Proof of Lemma
\texorpdfstring{\protect\ref{lemma:peelability_implies_good}}{3.11}:
Peelability implies a sparse basis} \label{sec:collapse_peeling_fast}

\subsection{\texorpdfstring{Proof of
Lemma \protect\ref{lemma:peelability_implies_good}\textup{(i)}
and \textup{(ii)}}{Proof of Lemma \protect\ref{lemma:peelability_implies_good}(i)
and (ii)}}
\label{sec:Peel12}

Let us begin by describing the proof strategy.

Instead of analyzing peeling on the collapsed graph $G_*$, we analyze a
different peeling process. We first run synchronous peeling on $G$ for
a large
constant $\tau$ number of iterations. We then collapse the resulting
graph, as discussed in Section~\ref
{subsec:sparse_basis_const_periphery}, that is,
coalescing variables connected to each other via degree 2 factors (cf.
Definition~\ref{def:collapsed_graph}). Finally, we run synchronous
peeling on
the collapsed graph until it gets annihilated. We show that this
process takes
at least as many iterations as synchronous peeling on $G_*$ (Lemma~\ref
{lemma:collapse_peeling_is_faster} below).
In order to bound the number of iterations under this new two-stages
process, we proceed as follows.
We choose the constant $\tau$
such that the residual graph $J_\tau$ is subcritical, and hence consists
of trees and unicyclic components of size $O(\log n)$ w.h.p. As a
consequence, the
collapsed graph---to be denoted by $\colp(J_\tau) $---contains only
checks of degree $3$ or
more, and consists of trees and unicyclic components of size $O(\log
n)$. It is
not hard to show that it takes only $O(\log\log n)$ additional rounds of
peeling to annihilate $\colp(J_\tau)$ under this condition (see
Lemma~\ref{lemma:peeling_fast_on_trees} below).

Several technical lemmas follow, which are proved in the Appendix~\ref
{app:proofs_of_technical}, except Lemma~\ref
{lemma:collapse_peeling_is_faster}, which we prove below.
At the end of the subsection, we provide a
proof of Lemma~\ref{lemma:peelability_implies_good}, parts (i) and
(ii).

Consider the peeling algorithm and define $\peel$ to be the peeling operator
corresponding to one round of synchronous peeling (cf. Table~\ref
{table:sync_peeling}). Thus, for a bipartite graph $G$, the
residual graph
after $t$ rounds of peeling is $\peel^t(G)$. Denote by
$\peel^{\infty}(G)$ the graph produced by the peeling procedure after
it halts: this is the empty graph if $G$ is peelable, and the core of
$G$ otherwise. Recall that $\TC(G)$ denotes the number of rounds of peeling
performed before halting at $\peel^{\infty}(G)$. Further, define
$\colp
$ to be the collapse
operator as per Definition~\ref{def:collapsed_graph}. For instance
$G_* =
\colp(G)$. The next lemma bounds from above the number of rounds of peeling
required to annihilate $G_*$, in terms of the modified peeling process
(consisting of $\tau$ rounds of peeling, followed by collapse, and
then peeling
until annihilation).

%
\begin{lemma}\label{lemma:collapse_peeling_is_faster}
For any constant $\tau\geq0$ and any peelable bipartite graph $G$,
\[
\TC\bigl(\colp(G)\bigr) \leq\TC\bigl(\colp\bigl(\peel^\tau(G)\bigr)
\bigr) + \tau.
\]
\end{lemma}

Peelability of a pair $(\alpha, R)$ immediately implies some useful properties.

%
\begin{lemma}\label{lemma:properties_of_peelable}
For a factor degree profile $(\alpha, R)$ that is peelable at rate
$\eta>0$,
we have:
\begin{longlist}[(ii)]
\item[(i)]$2\alpha R_2 \leq1- \eta$.
\item[(ii)]
$\alpha\leq1$.
\end{longlist}
\end{lemma}

Notice that the factor graph induced by degree $2$ check nodes is
in natural correspondence with an ordinary graph (replace every check node
by an edge) which is uniformly random given the number of edges. The
average degree of this graph is $2\alpha R_2$, and
Lemma~\ref{lemma:properties_of_peelable}(i) implies that it
is subcritical, as we would expect for a
peelable degree distribution.

Lemma~\ref{lemma:peelability_implies_good} is stated for the ensemble
$\D(n,R,m)$, $m=n\alpha$. However, in parts of the proof of this lemma,
we find
it convenient to work instead with the ensemble $\C(n,R,m)$ introduced
in Section~\ref{subsec:density_evolution}.

We need to characterize the residual graph $J_t$ after $t$ rounds of peeling.
Lemmas~\ref{lemma:peeling_uniform_n1_n2} and
\ref{lemma:concentration_DE} achieves this
for $G\sim\C(n,R,m)$. Together, they show
essentially that density evolution provides an accurate
characterization of
$J_t$. Using these Lemmas, we are able to deduce [see proof of
Lemma~\ref{lemma:peelability_implies_good}(i) and (ii) below] that
$J_\tau$ consists
of small trees and unicyclic components w.h.p., for large enough $\tau
$. Finally,
using Lemma~\ref{lemma:Conf_implies_Uniform}, we apply the same
results to $G\sim\D(n,R,m)$.

Recall that $n_1(G)$ denotes the number of variable nodes of degree
$1$ in $G$, and $n_{2+}(G)$ denotes the number of variable nodes of
degree $2$ or more in $G$. Let
%
%
\begin{equation}\qquad
\label{eq:config_n1_n2_defined} \C\bigl(n,R,m;n_1',n_2'
\bigr) \equiv\bigl\{ G\dvtx G \in\C(n,R,m), n_1(G)=n_1',
n_{2+}(G)=n_2'\bigr\}.
\end{equation}
In the lemma below, we slightly modify the peeling process, choosing to
retain all variable nodes $V$ in the residual graph
(check nodes are eliminated as usual). With a slight abuse of
notation, we keep denoting by $J_t$ the residual graph, although this
is obtained from $J_t$ by adding a certain number of isolated variable
nodes.

%
\begin{lemma}\label{lemma:peeling_uniform_n1_n2}
Consider a graph $G$ drawn uniformly at random from $\C(n,\break R, m)$. For any
$t\in\mathbb{N}$, consider synchronous peeling for $t$ rounds on $G$,
resulting
in the residual graph $J_t$. Suppose that for some $(\tR, \tm,\tn_1,
\tn_2)$, we have $J_t \in\C(n,\tR,\tm;\tn_1,\tn_2)$ with positive
probability. Then,
conditioned on $J_t \in\C(n,\tR,\tm;\tn_1,\tn_2)$, the residual graph
$J_t$ is
uniformly random within $\C(n,\tR,\tm; \tn_1,\tn_2)$.
\end{lemma}

Our final technical lemma bounds the number of peeling rounds needed to
annihilate a tree or unicyclic component.

%
\begin{lemma}
\label{lemma:peeling_fast_on_trees}
Consider a factor graph $G = (F,V,E)$ with no check nodes of degree
$1$ or $2$, and that is a tree or unicyclic. Then $G$ is peelable and
$\TC(G) \leq2\lceil\log_2 |V|\rceil$.
\end{lemma}
%

%
\begin{pf*}{Proof of Lemma~\ref{lemma:peelability_implies_good}\normalfont{(i)}
and \normalfont{(ii)}}
A standard calculation (see, e.g., \cite{DemboFSS}, or Section~\ref
{subsec:even_subgraphs_of_core} which carries through a similar calculation)
shows that, for a uniformly random graph
$\C(n,\tR,\tm;\tn_1,\tn_2)$, with $\tn_1,\tn_2\ge n\ve$ and
with $\tm\tR'(1)\ge\tn_1+2\tn_2+n\ve$ for some $\ve>0$,
the asymptotic degree distribution of variable nodes is
\begin{eqnarray*}
\prob\{D=0\} &=& q_0,
\\
\prob\{D=1\} &=& q_1,
\\
\prob\{D=\ell\}& = & (1-q_0-q_1) \prob\bigl\{
\Po_{\ge
2}(\lambda)=\ell\bigr\}\qquad \mbox{for all $\ell\ge2$} %
\end{eqnarray*}
for suitable choices of $q_0$, $q_1$, $\lambda$ depending on the
ensemble parameters.
Further, by a standard breadth-first search argument, the neighborhood
of a vertex $v$ is dominated stochastically by a
(bipartite) Galton--Watson tree,\vadjust{\goodbreak} with offspring distribution equal to
the size-biased version of $\tR$ at check nodes, and equal to of
$\prob
\{D= \cdot\}$
at variable nodes.

Consider $G\sim\C(n,R,m)$.
Using Lemmas \ref{lemma:concentration_DE} and \ref
{lemma:peeling_uniform_n1_n2},
it is possible to estimate the degree distribution,
of $J_t$.
A lengthy but straightforward calculation shows that the corresponding
branching factor is $\theta(J_t) = \alpha R'(z_t)$. Now, notice that
\[
R'(z) = 2 R_2 + \sum_{l=3}^k
l(l-1) z^{l-2} \leq2 R_2 + k(k-1) z
\]
for $z \leq1$. Choose $\tau= \tau(\eta, k) < \infty$ such that
$z_\tau
\leq
\eta/(3\alpha k (k-1))$. Then we have $\alpha R'(1) \rho'(z_\tau)
\leq
2 \alpha
R_2 + \eta/3$. But Lemma~\ref{lemma:properties_of_peelable} tells us
that $2
\alpha R_2 \leq1- \eta$. It follows that $\alpha R'(z_\tau) \leq1 -
2 \eta/3$.

In particular, the branching factor $\theta=\theta(J_\tau)$ associated
with the random graph
$J_{\tau}$ satisfies $\theta\leq1 - \eta/3$, with probability at
least $1-1/n^2$.
Following a standard argument \cite{Bollo} where we
explore the neighborhood of $v$ by breadth first search, we obtain that
with probability
at least $1 - 1/n^{1.7}$ for
$n \geq N_1(\eta,k)$, the connected component containing $v$ is a tree or
unicyclic, with size less than $\const_4 \log n$, for some $\const_4 =
\const_4(\eta,k) < \infty$. Applying a union bound, we obtain that for
$n \geq
N_2 = N_2(\eta, k)$, with probability at least $1/n^{0.7}$, the event
$\Ev_n$
occurs, where
%
%
\begin{eqnarray}
&&\Ev_n \equiv\{\mbox{All connected components in $J_\tau$
are trees or unicyclic}
\nonumber
\\[-8pt]
\\[-8pt]
\nonumber
 &&\hspace*{132pt}\mbox{and have size at most $\const_4 \log n$} \}.
\end{eqnarray}
Then, from Lemma~\ref{lemma:Conf_implies_Uniform}, we infer that $\Ev
_n$ occurs
with probability at least $1/n^{0.6}$ for $G\sim\D(n,R,m)$
provided $n \geq N_3$, where $N_3 = N_3(k) < \infty$. We stick to
$G\sim\D(n,R,m)$ for the rest of this proof.

We now analyze the peeling process starting with $J_{\tau}$ and
consider only what happens on $\Ev_n$ since it occurs with sufficiently
large probability.
Let us consider first point (i).
Clearly, tree components are peelable. If $R_2=0$, then there are no
factors of
degree 2, and unicyclic components are also peelable (Lemma~\ref
{lemma:peeling_fast_on_trees}). Thus, the entire graph is
annihilated by
peeling w.h.p., as claimed. If $R_2 > 0$, then the number of unicyclic
components of size smaller than $M$ is asymptotically Poisson with parameter
$\const_5< \infty$ uniformly bounded in $M$ (this follows, e.g., by
\cite{Wormaldshortcycles81};
see also \cite{WormaldRegular,Bollo}). It follows that with probability
at least $
\exp(-\const_5)/2$ for $n\geq N_4$, there are no unicyclic components
of size smaller than $M$. The expected number of unicyclic components
of size $M$ or larger is upper bounded by $\sum_{\ell\ge M}
\theta^{\ell}/(2\ell)\le\theta^M/(1-\theta)$, and for $M$ large
enough no unicyclic component of this sizes exists, with probability
at least $1-\exp(-\const_5)/4$.
Considering these two contributions, the graph contains no cycle with
probability at least $\exp(-\const_5)/4$ for $n\geq N_4$, and hence
it is peelable. This completes part (i).

For (ii), notice that in collapsing a connected component of $J_\tau
$, the
number of variable nodes does not increase. Further, a tree component
collapses to a
tree and a unicyclic component collapses either to a tree or a unicyclic
components. Thus, we can use Lemma~\ref{lemma:peeling_fast_on_trees} with
$N\leq\const_4 \log n $ to obtain the a bound of $(\const_1/2) \log
\log n \leq\const_1 \log\log n - \tau$ on the
number of additional peeling rounds needed, with probability at least $1-
1/n^{0.6}$. Since the probability of peelability is uniformly bounded
away from
zero as $n\to\infty$, the probability that the same bound on the
number of peeling rounds holds conditioned on peelability is at least
(for some $\delta>0$) $1- 1/(\delta
n^{0.6}) \geq1- 1/n^{0.5}$ for $n \geq N_5$, as required.
\end{pf*}

\subsection{\texorpdfstring{Proof of Lemma \protect\ref
{lemma:peelability_implies_good}\textup{(iii)}}
{Proof of Lemma \protect\ref
{lemma:peelability_implies_good}(iii)}}
\label{sec:PeelabilityGood}

The following lemma bounds the size of a supercritical Galton--Watson tree,
observed up to finite depth. The proof is in the Appendix~\ref
{app:proofs_of_technical}.

%
\begin{lemma}\label{lemma:GW_bound}
Consider a Galton--Watson branching process $\{Z_t\}_{t=0}^\infty$ with
$Z_0=1$ and with offspring distribution $\prob\{Z_1=j\}= b_j$, $j\ge
0$. Suppose $b_r \leq
(1-\delta)^r / \delta$ for all $r \geq0$, for some $\delta> 0$.
Also, assume
that the branching factor satisfies $\theta\equiv\sum_{j=1}^\infty
jb_j =
\E[Z_1] > 1$. Then there exists $\const= \const(\delta) > 0$ such
that the
following happens.

For any $\beta> 3$ and $T \in\mathbb{N}$, we have
%
%
\begin{equation}
\prob \Biggl[ \sum_{t=0}^T
Z_t > (\beta\theta)^T \Biggr] \leq2 \exp\bigl(-\const(
\beta/3)^{T}\bigr). \label{eq:GWBound}
\end{equation}
\end{lemma}

\begin{pf*}{Proof of Lemma~\ref{lemma:peelability_implies_good}\normalfont{(iii)}}
From Lemma~\ref{lemma:properties_of_peelable}(ii), we know that
$\alpha\leq
1$. The following occurs in the collapse process: Let $G^{(2)}=
(F^{(2)}, V, E^{(2)})$ be the subgraph of $G$ induced by the degree $2$ factor
nodes (with isolated vertices retained). We have $F_* = F \setminus F^{(2)}$.
All variable nodes that belong to a single connected component of
$G^{(2)}$
coalesce into a single supernode $v' \in V_*$ in $G_*$, with a
neighborhood that
consists of the union of the individual neighborhoods restricted to
$F_*$ (cf.
Definition~\ref{def:collapsed_graph}). As mentioned above, $G^{(2)}$
is a
random factor graph
with $\alpha R_2 n$ factor nodes of degree 2, and is in one-to-one
correspondence with a uniformly random graph. For $v' \in V_*$, we
denote by
$S(v')$ the number of variable nodes in $V$ in the component $v'$.
Lemma~\ref{lemma:properties_of_peelable}(i) implies that the branching
factor of $G^{(2)}$ obeys $2
\alpha R_2 \leq1- \eta$, that is, $G^{(2)}$ is subcritical. This
leads to the
following claim that follows immediately from a well-known result on
the size
of the largest connected component in a subcritical random graph
\cite{Bollo}.

\textit{Claim}~1: There exists $\const_2 = \const_2 (\eta)< \infty$,
$N_2=N_2(\eta) < \infty$ such that the following occurs for all $n >
N_2$. No
component $v' \in V_*$ is composed of more than $\const_2 \log n$ variable
nodes, that is, $ \max_{v'\in V_*}S(v') \leq\const_2 \log n$, with
probability at least $1-1/n$.

Let $G^{\sim2} \equiv(F_*, V, E \setminus E^{(2)})$, that is,
$G^{\sim2}$ is
the subgraph of $G$ induced by factors of degree greater than $2$ (with
isolated vertices retained).

From Poisson estimates on the node degree
distribution, we get the following.

\textit{Claim}~2: There exists $\const_3 = \const_3 (\eta,k)<
\infty$,
$N_3=N_3(\eta,k) < \infty$ such that the following occurs. For all $n
> N_3$,
no variable node $v \in V$ has degree larger than $\const_3 \log n$ in
$G^{\sim
2}$, that is, $\deg_{G^{\sim2}}(v) \leq\const_3 \log n$ for all $v
\in V$, with
probability at least $1-1/n$.

Note that we used $\alpha< 1$ [from Lemma~\ref
{lemma:properties_of_peelable}(i)] to avoid dependence on
$\alpha
$ in the
above claim.

Let
\begin{eqnarray*}
\Ev_n &\equiv& \bigl\{S\bigl(v'\bigr) \leq
\const_2 \log n \mbox{ for all } v' \in V_* \bigr\} \\
&&{}\cap
\bigl\{ \deg_{G^{\sim2}}(v) \leq\const_3 \log n \mbox{ for all }
v \in V \bigr\}.
\end{eqnarray*}
Using claims 1 and 2 above and a union bound, we deduce that $\Ev_n$
holds with probability at least $1-2/n$
for $n > N_4$, for some $N_4 = N_4(\eta,k)< \infty$.

Clearly, $G^{\sim2}$ is independent of $G^{(2)}$.
In particular, for $v \in V$ that is part of supernode $v' \in V_*$, we know
that $|S(v')|$ is independent of $G^{\sim2}$.
There is a slight dependence between the degree of
different variable nodes, but assuming $\Ev_n$, the effect of this is
small if we
only condition on $\polylog(n)$ nodes in $G_*$. This enables our bound
on the
size of balls in $G_*$.

Recall that the distribution of random variable $X_1$ is dominated by the
distribution of $X_2$, if there exists a coupling between $X_1$ and
$X_2$ such
that $X_1 \leq X_2$ with probability 1. In bounding the size of a ball of
radius $T_{\mathrm{ub}}$, we are justified in replacing degree distributions
by dominating
distributions, and in assuming that there are no loops.
%
%

Fixing a vertex $v\in V_*$, we construct the ball
$\Ball_{G_*}(v,T_{\mathrm{ ub}})$ sequentially through a breadth-first
search.
Choose $\ve= \eta/2$. For $n$ large enough, the distribution of
$|S(v')|$ is
dominated by the distribution of the number of nodes in a
Galton--Watson tree
with offspring distribution $\Po(2 \alpha R_2 + \ve)$. The
distribution of
$\deg_{G^{\sim2}}(v)$ is dominated by $\Po(\alpha(\sum_{l=3}^k l
R_l )+
\ve)$. In particular, the degree distribution of $G_*$ is dominated by
a geometric distribution $b_r \leq(1-\delta)^r/\delta$ for some
$\delta=
\delta(\eta, k)>0$.
Assuming $\Ev_n$, this also holds conditionally on the nodes revealed so
far, as long as the number of these is, say, $\polylog(n)$.

Thus, assuming $\Ev_n$, the number of nodes in a ball of
radius $T_{\mathrm{ ub}}= \const_1 \log\log n $ is dominated by the number
of nodes
in a Galton--Watson tree of depth $T_{\mathrm{ ub}}$ with offspring
distribution
$(b_r)_0^\infty$ satisfying $b_r \leq(1-\delta)^r/\delta$ for some
and $\theta\equiv\sum_{j=1}^\infty jb_j < \const_5$, for
$n \geq N_5$.
We deduce
from Lemma~\ref{lemma:GW_bound} that
%
%
\begin{equation}
\prob \Bigl[ \max_{v' \in V_*} \bigl|\Ball_{G_*}
\bigl(v', T_{\mathrm{ ub}}\bigr)\bigr | \leq (\log n)^{\const_6} \big|
\Ev_n \Bigr] \geq1- 1/n
\end{equation}
for some $\const_6=\const_6(\eta, k)<\infty$, where $|\Ball_{G_*}(v',
T_{\mathrm{ ub}})|$
denotes the number of super-nodes in $\Ball_{G_*}(v', T_{\mathrm{
ub}})$. But
given $\Ev_n$,
the size of components $v' \in V_*$ is uniformly bounded by $\const_2
\log n$.
Thus, conditioned on $\Ev_n$, we have $\max_{v' \in V_*} S(v',
T_{\mathrm{
ub}}) \leq
\const_2 (\log n)^{\const_6+1}$ with probability at least $1-1/n$. At this
point, we recall that $\prob[\Ev_n] > 1- 2/n$, and the result follows.
\end{pf*}

\section{Characterizing the periphery}
\label{secperrr}

Consider a factor graph $G$ when
it has a nontrivial $2$-core. Recall the definitions of the $2$-core,
backbone and periphery of a graph from Section~\ref{sec:ClusterConstruction}.
First, we note some of the properties of these subgraphs that will be
useful in
the proof of the main lemmas of this section.

As a matter of notation, for a bipartite graph $G$ chosen uniformly at random
from the set $\G(n, k, m)$ we denote by $\GP$ the periphery of $G$ and
by $G_{\mathrm p}$
(lower case subscript)
a subgraph of $G$ that is a potential candidate for being the periphery
of $G$.
Similarly, we denote by $\GB$ the backbone of $G$ and by $G_{\mathrm
b}$ a
subgraph of
$G$ that is a potential candidate for being the backbone of $G$.

\subsection{Proof of Lemma
\texorpdfstring{\protect\ref{lemma:UniformPeelable}}{3.9}: Periphery is conditionally a
uniform random graph}
\label{subsec:perip_conditional_random}

Lemma~\ref{lemma:UniformPeelable} states that if we fix the number of
nodes and the
check degree profile of the periphery of a graph $G$ chosen uniformly
at random from
the set $\G(n, k, m)$ then the periphery, $\GP$, is distributed
uniformly at random conditioned on being peelable. Since
the original graph $G$ is chosen uniformly at random, in order to prove
this lemma it is
enough to count, for each possible choice of the periphery $\GP$, the
number of
graphs $G$ that have the periphery $\GP$.

Before proving Lemma~\ref{lemma:UniformPeelable}, we first introduce
the concept
of a ``rigid'' graph and establish a monotonicity property for the backbone
augmentation procedure which was defined in Section~\ref
{sec:ClusterConstruction}. We use the notation $G \subseteq G'$ if
$G$ is a subgraph of $G'$.

%
\begin{lemma}\label{lem:monotonocity_bba}
Let $G= (F, V, E)$ be a bipartite graph and let $G_{\mathrm s}$ be the
subgraph of $G$ induced by some $F_{\mathrm s} \subseteq F$. Let
$F_{\mathrm
l}$ and $F_{\mathrm u}$ be subsets of $F$ such that $F_{\mathrm l}
\subseteq
F_{\mathrm u}$ and $F_{\mathrm l} \subseteq F_{\mathrm s}$. Let $\Bl
^{(0)}$ be the
subgraph induced by $F_{\mathrm l}$ (so $\Bl^{(0)}\subseteq G_{\mathrm s}$)
and let $\Bu^{(0)}$ be the subgraph induced by $F_{\mathrm u}$. Denote by
$\Bl^{(\infty)}$ the output of
the backbone augmentation process on $G_{\mathrm s}$ with the initial graph
$\Bl^{(0)}$ and
by $\Bu^{(\infty)}$ the output of the backbone augmentation process on
$G$ with the initial graph $B_{\mathrm s}^{(0)}$. Then $\Bl^{(\infty)}
\subseteq\Bu^{(\infty)}$.
\end{lemma}
%

The proof of Lemma~\ref{lem:monotonocity_bba} can be found in
the Appendix~\ref{app:periphery}.


\begin{definition}
Define a graph to be \emph{rigid} if its backbone is the whole graph.
We denote by
$\cR(n,k,m)$ the class of rigid graphs with $n$ variable nodes, and $m$
check nodes each of degree $k$.
\end{definition}

%
\begin{lemma}\label{th:bp_partition_uniqueness}
Consider a bipartite graph $G= (F, V, E)$ from the ensemble $\G
(n,k,m)$. For
some set of check nodes $F_{\mathrm b}\subseteq F$ denote by
$G_{\mathrm b}=
(F_{\mathrm b}, V_{\mathrm b},
E_{\mathrm b})$ the
subgraph induced by $F_{\mathrm b}$, and denote by $G_{\mathrm p}=
(F_{\mathrm p},
V_{\mathrm p}, E_{\mathrm p})$ the
subgraph of $G$ induced by the pair $(F_{\mathrm p}\equiv F\setminus
F_{\mathrm
b}, V_{\mathrm p}
\equiv
V\setminus V_{\mathrm b})$. Assume $G_{\mathrm b}$ and $G_{\mathrm p}$
satisfy the
following conditions:
\begin{itemize}
\item$G_{\mathrm p}$ is peelable,
\item$G_{\mathrm b}$ is rigid,
\item$|\da|\ge2, \forall a\in F_{\mathrm p}$.
\end{itemize}
Then $G_{\mathrm b}$ is the backbone of $G$ (and $G_{\mathrm p}$ is the
periphery).
\end{lemma}
%
%
\begin{pf}
If $G_{\mathrm b}$ is empty, the lemma is trivially true. Assume
$G_{\mathrm
b}$ is\break 
nonempty. We
prove this lemma in two steps. In the first step we prove that
$G_{\mathrm b}$
is a subgraph of $\GB$, the backbone of $G$. In the second step we show
that $\GB$ cannot
contain anything outside $G_{\mathrm b}$.

Since $G_{\mathrm b}$ is rigid, it contains a nonempty 2-core
$(G_{\mathrm
b})_\core$ and
the output
of the backbone augmentation procedure with initial graph $(G_{\mathrm
b})_\core$ is
$G_{\mathrm b}$ itself. Furthermore, $(G_{\mathrm b})_\core$ is part
of $\GC
$, the 2-core
of the
original graph $G$, since by definition a 2-core is the maximal
stopping set (cf. Definition~\ref{def:stopping_set})
and $(G_{\mathrm b})_\core$ is a stopping set in $G$. Hence, the
monotonicity of the
backbone augmentation procedure implies that $G_{\mathrm b}\subseteq
\GB$.

In the second step, we prove that $\GB$ cannot contain any node outside
$G_{\mathrm b}$.
First, note that $G_{\mathrm p}$ cannot contain any check node from the
2-core of
the original
graph~$G$. We prove this by contradiction. Suppose instead that $\tilde
{F}$ is the nonempty set of all the
check nodes from the 2-core of $G$ that are in $G_{\mathrm p}$. Let
$\tilde{V}$
be the
set of neighbors of $\tilde{F}$ in $G_{\mathrm p}$. The nodes in
$\tilde
{V}$ are
also part of the
2-core of $G$ and have degree at least 2 in the 2-core of $G$. Furthermore,
there is no edge incident from $F_{\mathrm b}$ to $V_{\mathrm p}$
because, by
definition,
$G_{\mathrm b}$ is check-induced. In particular, in the 2-core of
$G$, there is no other edge incident on variables in $\tilde{V}$
beyond the
ones coming from $\tilde{F}$. Hence, in the nonempty subgraph $\tilde
{G} \subseteq
G_{\mathrm p}$ induced by the check nodes in $\tilde{F}$ and all their
neighbors every
variable node has degree at least 2. This subgraph is then, by
definition,
a stopping set in $G_{\mathrm p}$. But by assumption $G_{\mathrm p}$ is
peelable and
cannot contain a stopping set. This is a contradiction that rules out
the existence
of a nonempty set $\tilde{F}$. Hence, the 2-core of $G$ is contained
entirely in $G_{\mathrm b}$ (recall that both $G_{\mathrm b}$ and the 2-core
are check-induced).

Let $B^{(\GC)}$ and $B^{(G_{\mathrm b})}$ be the output of the backbone
augmentation
procedure on $G$, once with initial subgraph given by the 2-core of $G$
and once
with the initial subgraph given by $G_{\mathrm b}$ (which contains the
2-core of
$G$). By
monotonicity, $B^{(\GC)} \subseteq B^{(G_{\mathrm b})}$. But the process
with the
initial subgraph $G_{\mathrm b}$ terminates immediately since, by
assumption, all check
node outside $G_{\mathrm b}$ have at least two neighbors in $G_{\mathrm p}$.
Therefore, $\GB
=B^{(\GC)}
\subseteq B^{(G_{\mathrm b})} = G_{\mathrm b}$. This completes our proof.
\end{pf}
%
%
It is easy to see that the converse of Lemma~\ref
{th:bp_partition_uniqueness} is
also true, as stated below.
%

\begin{remark}\label{rem:periphery_partition_conv}
If $G_{\mathrm b}=\GB$ is the backbone of $G$, then the subgraphs
$G_{\mathrm
b}$ and
$G_{\mathrm p}= G \setminus G_{\mathrm b}=\GP$ satisfy the condition
of Lemma~\ref{th:bp_partition_uniqueness}. Here, $G\setminus
G_{\mathrm b}$
denotes the subgraph of $G$ induced by
$(F\setminus F_{\mathrm b}, V\setminus V_{\mathrm b})$.
\end{remark}

Notice that the fact that the graph $G\setminus G_{\mathrm b}$ is peelable
follows from the connection between the peeling algorithm and \BPzero\
stated in Lemmas \ref{lemma:PeelingBP} and \ref{lemma:PeelingBPFP}. We
stated that the messages coming out
of the backbone are always $0$. From the check node update rule, an
incoming $0$
message to a check node can be dropped without changing any of the
outgoing messages as long as there is
at least one other incoming message. By definition, there is no edge between
variable nodes in the periphery and check nodes in the backbone. Furthermore,
all the check nodes in the periphery have at least two neighbors in the
periphery. Therefore, \BPzero\ on the periphery has the same messages as the
corresponding messages of \BPzero\ on the whole graph. In particular,
the fixed point
of \BPzero\ on the periphery is all $*$ messages which shows that the periphery
subgraph is peelable.
We now prove Lemma~\ref{lemma:UniformPeelable}.
%
\begin{pf*}{Proof of Lemma~\ref{lemma:UniformPeelable}}
Our goal is to characterize the probability of observing the periphery
of $G$
to be $G_{\mathrm p}= (F_{\mathrm p}, V_{\mathrm p}, E_{\mathrm p})$.
We use the
shorthand notation
$G\setminus G_{\mathrm p}$
to denote the subgraph of $G$ induced by the check-variable nodes pair
$(F\setminus F_{\mathrm p}, V\setminus V_{\mathrm p})$. Let $G_{\mathrm b}=
(F\setminus F_{\mathrm p},
V\setminus
V_{\mathrm p}, E_{\mathrm b}) = G\setminus G_{\mathrm p}$ and
$E_{\mathrm{pb}}=\{
(i,a)|i\in
V\setminus V_{\mathrm p},
a \in F_{\mathrm p}\}$ be a set of edges that satisfy the condition
$\deg_{E_{\mathrm p}}(a)+\deg_{E_\mathrm{pb}}(a) = k$ for all $a \in
F_{\mathrm p}$. As
before, we
denote by $\GB$ and $\GP$ the actual periphery and backbone of the
graph $G$.
Define the set of rigid graphs on $\nb$ variable nodes, $\mb$ check
nodes and
check degree $k$, $\cR(\nb,k,\mb)$, as
%
%
\begin{eqnarray}
\quad\hspace*{-1pt}&&\cR(\nb,k,\mb)
\nonumber
\\[-8pt]
\\[-8pt]
\nonumber
&&\hspace*{-2pt}\qquad= \bigl\{G_{\mathrm b}= (F_{\mathrm b},
V_{\mathrm b}, E_{\mathrm b}) \dvtx|F_{\mathrm b}| = \mb,
V_{\mathrm b}= \nb, |\da| = k\ \forall a\in F_{\mathrm b},
G_{\mathrm b}\mbox{ is rigid}\bigr\}.
\end{eqnarray}

By Lemma~\ref{th:bp_partition_uniqueness},
%
%
\begin{eqnarray}
&&\bigl\{G \in\G(n,k,m) \dvtx\GP=G_{\mathrm p}, \GB=G_{\mathrm b}\bigr\}
\nonumber
\\[-8pt]
\\[-8pt]
\nonumber
&&\qquad = \bigl\{G \in\G(n,k,m) \dvtx G_{\mathrm p}\subseteq G, G\setminus
G_{\mathrm
p}= G_{\mathrm b} , G_{\mathrm p}\in\cP, G_{\mathrm b}
\in\cR\bigr\},
\end{eqnarray}
and in particular,
%
%
\begin{eqnarray}
\label{eq:set_graph_gp}&& \bigl\{G \in\G(n,k,m) \dvtx\GP=G_{\mathrm p}\bigr\}
\nonumber
\\[-8pt]
\\[-8pt]
\nonumber
&&\qquad= \bigl\{G
\in\G(n,k,m) \dvtx G_{\mathrm p} \subseteq G, G_{\mathrm p}\in\cP, G
\setminus G_{\mathrm p}\in\cR\bigr\}.
\end{eqnarray}
From equation
\eqref{eq:set_graph_gp}, and counting all the choices for the subgraph
$G_{\mathrm b}=
G\setminus G_{\mathrm p}$, and the edges that connect $G_{\mathrm p}$ and
$G_{\mathrm b}$,
%
%
\begin{eqnarray}
&& \bigl|\bigl\{G \in\G(n,k,m) \dvtx\GP=G_{\mathrm p}\bigr\}\bigr|
\nonumber
\\
&&\qquad = \sum_{G_{\mathrm b}} \sum_{E_{\mathrm{pb}}}
\bigl|\bigl\{G \in\G(n,k,m) \dvtx G_{\mathrm p} \subseteq G, G_{\mathrm p}\in
\cP, G\setminus G_{\mathrm p}= G_{\mathrm b}, G_{\mathrm
b}\in\cR,\\
&&\hspace*{235pt}{} E
\setminus(E_{\mathrm p}\cup E_{\mathrm b}) = E_{\mathrm{pb}}\bigr\}\bigr|.\nonumber
\end{eqnarray}
For fixed $G_{\mathrm p}$ and $G_{\mathrm b}$, we can count the number
of ways
these two subgraphs
can be connected to each other. Letting $\bar{R}$ be the degree
profile of
$G_{\mathrm p}$, we have
%
%
\begin{eqnarray}
\label{eq:card_fixed_periphery}
&&\bigl |\bigl\{G \in\G(n,k,m) \dvtx\GP=G_{\mathrm p}\bigr\}\bigr|
\nonumber
\\[-8pt]
\\[-8pt]
\nonumber
&&\qquad = \sum_{G\setminus G_{\mathrm p}} \prod_{l=2}^k
{n-|V_{\mathrm p}| \choose k-l}^{|F_{\mathrm p}
|\bar{R}_l} \bI(G\setminus G_{\mathrm p}
\in\cR) \bI(G_{\mathrm
p}\in\cP).
\end{eqnarray}
We can rewrite this as
%
%
\begin{eqnarray}
&& \bigl|\bigl\{G \in\G(n,k,m) \dvtx\GP=G_{\mathrm p}\bigr\}\bigr|
\nonumber
\\[-8pt]
\\[-8pt]
\nonumber
&&\qquad = \prod_{l=2}^k {n-|V_{\mathrm p}|
\choose k-l}^{|F_{\mathrm p}|\bar{R}_l}\bigl |\cR\bigl(n - |V_{\mathrm p}|, k, m -|F_{\mathrm p}|\bigr)\bigr|
\bI(G_{\mathrm p}\in\cP).
\end{eqnarray}
It is clear that the cardinality of the set $\cR(\nb,k,\mb)$ is a
function of
only $\nb$ and~$\mb$. Hence,
%
%
\begin{equation}
\label{eq:num_graph_fix_periphery}
 \bigl|\bigl\{G \in\G(n,k,m) \dvtx\GP=G_{\mathrm p}\bigr\}\bigr| = Z
\bigl(n_{\mathrm p}, k, R^{\mathrm
p}, m_{\mathrm p}\bigr)
\bI(G_{\mathrm p} \in\cP),
\end{equation}
for some function $Z( \cdot,\cdot,\cdot,\cdot)$.
Since the graph $G$ itself was chosen uniformly at random from the set $
\G(n,k,m)$, this shows that conditioned on $(n_{\mathrm p}, R^{\mathrm p},
m_{\mathrm p})$, all
graphs $G_{\mathrm p}
\in\cP$ with $n_{\mathrm p}$ variable nodes, $m_{\mathrm p}$ check
nodes, and
check degree
profile $R^{\mathrm p}$ are equally likely to be observed.
\end{pf*}
%
%

\subsection{Proof of Lemma
\texorpdfstring{\protect\ref{lemma:periphery_peelable}}{3.12}: Periphery is exponentially peelable}
\label{subsec:perip_exp_peelable}

Let $G= (F, V, E)$ be a graph drawn uniformly at
random from $\G(n,k,\alpha n)$, and let $\GP= (\FP, \VP, \EP)$ be its
periphery. Recall the connection between \BPzero\ and the peeling
algorithm from
Section~\ref{sec:BP_DE}.
Let $Q$ be defined as in Theorem~\ref{thm:Main}, that is, $Q$ is the
largest positive solution of
$Q=1-\exp\{-k\alpha Q^{k-1}\}$.
In light of Lemma~\ref{lemma:FP_concentration_DE}, we define the
asymptotic degree profile pair of the periphery, $(\ba, \bR(x))$ as
follows
(recall that, from Lemma~\ref{lemma:PeelingBPFP}, the periphery does
include check
nodes receiving at most $k-2$ messages of type $0$).
%

\begin{definition}\label{def:periphery_deg_est}
%
%
\begin{eqnarray}
\bar{R}(x) & \equiv&\frac{1}{ 1 - Q^{k} - k (1-Q)Q^{k-1} } \cdot \sum_{l=2}^{k}
\pmatrix{k \cr l} (1- Q)^l Q^{k - l } x^l,
\\
\ba& \equiv&\alpha \biggl(\frac{1-Q^k - k
(1-Q)Q^{k-1}}{1-Q} \biggr).
\end{eqnarray}
\end{definition}

Unlike the backbone where all check nodes are of degree $k$, the
periphery can
have check nodes of degrees between $2$ and $k$. Among these, check
nodes of degree
$2$ are of importance to us since they can potentially form long strings.
Strings are particularly unfriendly structures for the peeling algorithm;
peeling takes linear time to peel such structures. In the next lemma,
we define
a parameter $\theta$ as a function of $Q$, which is the estimated
branching factor of the subgraph of the periphery induced by check
nodes of
degree $2$. Lemma~\ref{lem:bdd_branch_fact} proves that this branching factor
is less than one for all $\alpha\in(\ad(k), 1]$.

%
\begin{lemma}\label{lem:bdd_branch_fact}
Let $\theta\equiv\alpha k (k-1) (1-Q) Q^{k-2}$ with $Q$ as defined in
Theorem~\ref{thm:Main}. Then $\theta< 1$ for
all $\alpha\in(\ad(k), 1]$.
\end{lemma}

Proof of this lemma can be found in the Appendix~\ref{app:periphery}.

%
\begin{lemma}\label{lemma:periphery_peelable_av}
Let $Q$ be defined as in Theorem~\ref{thm:Main}. Then there exists
$\eta_1 = \eta_1(\alpha,k) > 0$ such that the pair $(\ba,
\bar{R})$ defined in Definition~\ref{def:periphery_deg_est} is
peelable at rate~$\eta_1$.
Further, $0\le f(z,\ba,\bR)\le(1-\eta_1)z$ for all $z\in(0,1]$.
\end{lemma}

%
\begin{pf}
In view of the density evolution recursion (Definition~\ref
{eq:density_evolution}), define
\[
f(z) = 1- \exp\bigl(-\ba\bar{R}'(z)\bigr).
\]
We prove the lemma by showing
that $f'(0) = \theta<1$ and that $f(z) < z$ strictly
for $z\in(0, 1]$.

Using the definitions of $\ba$ and $\bar{R}(z)$, the function $f(z)$
can be written as
%
%
\begin{equation}
f(z) =1 - \exp \bigl(-\alpha k \bigl( \bigl(Q+(1- Q)z \bigr)^{k-1} -
Q^{k-1} \bigr) \bigr).
\end{equation}
By a straightforward calculation, and using
Lemma~\ref{lem:bdd_branch_fact}, we get
%
%
\begin{equation}
\label{eq:6_df} f'(0) = \ba\bar{R}'(0) \exp\bigl(-\ba
\bar{R}'(0)\bigr) = \alpha k (k-1) (1-Q) Q^{k-2} =
\theta<1.
\end{equation}
Assume $0 \le y \le1$ to be fixed point of $f$, that is,
%
%
\begin{equation}
y =1 - \exp \bigl(-\alpha k \bigl( \bigl(Q+(1- Q)y \bigr)^{k-1} -
Q^{k-1} \bigr) \bigr).
\end{equation}
Using the identity $Q = 1- \exp(- \alpha k Q^{k-1})$ and after some
calculation, we get
%
%
\begin{equation}
\label{eq:6_Q_contradiction} Q + (1-Q) y =1 - \exp \bigl(-\alpha k \bigl(Q+(1- Q)y
\bigr)^{k-1} \bigr).
\end{equation}
Equation~\eqref{eq:6_Q_contradiction} shows that $Q + (1-Q) y $ is a
fixed point of the original
density evolution recursion (\ref{eq:density_evolution}) with $R(x) =
x^k$. Since, by definition, $Q$ is the largest fixed point of that
recursion, $y = 0$ is the only fixed point of $f(z) = 1- \exp(-\ba
\bar{R}'(z))$ in the interval $[0, 1]$. Since $f'(0)<1$, we have
$f(z)<z$ for all $z\in(0,1]$ and, therefore, $f(z)/z<1$ for all
$z\in[0,1]$.
The claim follows by taking $\eta_1= 1-\sup_{z\in[0,1]}f(z)/z$,
with $\eta_1>0$ by continuity of $z\mapsto f(z)/z$ over the compact $[0,1]$.
\end{pf}
%

We can now prove Lemma~\ref{lemma:periphery_peelable}.
\begin{pf*}{Proof of Lemma~\ref{lemma:periphery_peelable}}
For any $\ve> 0$, by Lemmas \ref
{lemma:BP_gives_core_backbone_periphery} and \ref
{lemma:FP_concentration_DE}, we know that
%
%
\begin{eqnarray} \label{eq:parameters_near_baraR}
|\aP- \ba| &<& \ve,
\nonumber
\\[-8pt]
\\[-8pt]
\nonumber
\bigl| \RP_l - \bar{R}_l\bigr| &<& \ve\qquad \mbox{for } l \in\{2,
\ldots, k\},
\end{eqnarray}
hold w.h.p.

As before, let $f(z, \alpha, R)=1- \exp\{-\alpha R'(z) \}$. Using
$\RP
_0=\RP_1=0$ we obtain
that the function $f(z, \alpha, R)/z$ is an analytic function over
set $[0,1]^{k+2}$. By Lemma~\ref{lemma:periphery_peelable_av},
$f(z,\ba,\bR)/z\le1-\eta_1$. It follows that, for $\ve>0$ small enough,
$\partial f(z,\ba,\bR)/\partial z\le1-(\eta_1/2)$ using continuity
$\partial f/\partial z$ with respect to the other arguments of $f$. We
infer that the periphery is w.h.p. peelable at rate $\eta=\eta_1/2$.
This proves part (i). Part (ii) follows immediately from Lemma~\ref
{lemma:FP_concentration_DE}.
\end{pf*}

%
%

\section{Proof of Lemma \texorpdfstring{\protect\ref{lemma:core_few_low_weight}}{3.5}}
\label{app:proof_of_core_separation_lemma}

We find it convenient to work within the configuration model: we
assume here that $G$ is drawn uniformly at random from $\C(n,k,m)$.
The following fact is an immediate consequence of Lemma~\ref
{lemma:peeling_uniform_n1_n2}.

%
\begin{fact}\label{fact:core_uniform_nC_mC}
Assume $G$ is drawn uniformly at random from $\C(n,k,m)$, and denote by
$n_{\mathtt{C}},\mC$ the number of
variable and check nodes in the core of $G$.
Suppose $(n_{\mathtt{C}}=n_{\mathrm c},
\mC
= m_{\mathrm c})$ occurs with positive probability. Then
conditioned on $(n_{\mathtt{C}}=n_{\mathrm
c}, \mC=m_{\mathrm c})$, the core is drawn uniformly from
$\C(n_{\mathrm c}, k, m_{\mathrm c};0,\break n_{\mathrm c})$ [recall the
definition of
this ensemble in
equation~\eqref{eq:config_n1_n2_defined}].
\end{fact}

In words, the core is drawn uniformly from
$\C(n_{\mathtt{C}}, k, \mC)$
conditioned on all variable nodes having degree $2$ or more.

Now, it has been proved \cite{DemboFSS} that, w.h.p.
%
%
\begin{eqnarray}
\bigl|n_{\mathtt{C}}/n - \bigl(1 - \exp (-\alpha k \wh{Q}) (1+\alpha k \wh{Q})
\bigr) \bigr| &=& o(1), \label{eq:nC_limit}
\\
\bigl|\mC/n - \alpha Q^k \bigr| &=& o(1),\label{eq:mC_limit}
\end{eqnarray}
where $(Q, \wh{Q})$ is as defined in Theorem~\ref{thm:Main}. The above
bounds also follow from Lemmas \ref
{lemma:BP_gives_core_backbone_periphery} and \ref
{lemma:DegreeDensityEvolution}.

The kernel of the core system $\Score$ contains all vectors
$\ux$ with the following property. Let $V_{(1)} \subseteq\VC$ be the
subset of variables taking value $1$ in $\ux$
(i.e., the support of~$\ux$). Then the subgraph of $\GC$ induced by
$V_{(1)}$ has no check node with odd degree.

We will refer to such subgraphs as to \emph{even} subgraphs.
Explicitly, even subgraphs are variable-induced subgraphs such that no
check node has odd degree. We want characterize the even subgraphs of
$\GC$ having no more than $n\ve$ variable nodes, in terms of their size
and number. Lemma~\ref{lemma:even_subgraphs_characterization} in
Section~\ref{subsec:even_subgraphs_of_core} below allows us to do
this provided certain conditions are met. Our next lemma tells us that
the core meets these conditions w.h.p.

\begin{lemma}\label{lemma:core_is_good}
Fix $k$ and consider any $\alpha\in(\ad(k),\as(k))$. There exists
$\delta= \delta(\alpha,k)>0$ such that the following happens. Let $G$
be drawn uniformly from
$\C(n,k,\alpha n)$. Let $n_{\mathtt{C}}$ be the (random) number of variable nodes
in the core, $\mC$ be the number of check nodes in the core and $\aC
\equiv\mC/n_{\mathtt{C}}$. Let $\etaC
$ be the unique positive solution of
%
%
\begin{equation}
\frac{\etaC(e^{\etaC} -1)}{e^{\etaC} - 1 - \etaC}
= \aC k \label{eq:etaC_defn}
\end{equation}
and let $\theta_{\mathtt{2C}}\equiv
\etaC(k-1)/(e^{\etaC} -1)$. For any $\delta'>0$,
we have, w.h.p.:
\begin{longlist}[(iii)]
\item[(i)] $\theta_{\mathtt{2C}} \leq1-
\delta$.

\item[(ii)] $\aC\in[2/k+\delta,1]$.

\item[(iii)] $n_{\mathtt{C}}/n \geq(1 - \exp
(-\alpha k \wh{Q})(1+\alpha k \wh
{Q}))-\delta'$.
\end{longlist}
\end{lemma}

The discussion in Section~\ref{subsec:even_subgraphs_of_core} throws
light on the definitions of $\etaC$ and
$\theta_{\mathtt{2C}}$ used.

\begin{pf*}{Proof of Lemma~\ref{lemma:core_is_good}}
From equations~\eqref{eq:nC_limit}, \eqref{eq:mC_limit}, we deduce that
$\etaC= \alpha k\wh{Q} +o(1)$ w.h.p., leading to
\[
\theta_{\mathtt{2C}} = \alpha k (k-1) Q^{k-2}(1-Q) + o(1) \leq1-
\delta
\]
for sufficiently small $\delta$, using Lemma~\ref{lem:bdd_branch_fact}.
Thus, we have established point (i).

Point (iii) and the lower bound in point (ii) are easy consequences
of equations~\eqref{eq:nC_limit}, \eqref{eq:mC_limit}. The upper
bound in
point (ii), $\aC\leq1$ w.h.p., follows directly from the fact that
for $\alpha< \as$, the system $\bH x = \ub$ has a solution for all
$\ub\in\{0,1\}^m$ w.h.p.
\end{pf*}

\begin{pf*}{Proof of Lemma~\ref{lemma:core_few_low_weight}}
Consider first $G\sim\C(n,k,m)$.
Applying Fact~\ref{fact:core_uniform_nC_mC} and Lemma~\ref
{lemma:core_is_good}, we deduce that, conditional on the number
of nodes, the core is $\GC\sim\C(n_{\mathrm c}, k, m_{\mathrm
c};0,n_{\mathrm
c})$ and satisfies
the conditions of Lemma~\ref{lemma:even_subgraphs_characterization}
proved below. By Lemma~\ref{lemma:even_subgraphs_characterization},
the elements of $\Lcore(\ve n)$ are in correspondence with simple
loops in the subgraph of $\GC$ induced by degree-$2$ variable
nodes. The sparsity bounds follows from Lemma~\ref
{lemma:even_subgraphs_characterization}.
The clam that they are, with high probability, disjoint, follows
instead from the fact that this random subgraph is subcritical (since
$2\alpha R_2<1$), and hence decomposes in trees and unicyclic components.

Using Lemma~\ref{lemma:Conf_implies_Uniform}, we deduce that the result
holds also for the
$G\sim\G(n,k,m)$ as required.
\end{pf*}

\subsection{Characterizing even subgraphs of the core}
\label{subsec:even_subgraphs_of_core}

This section aims at characterizing the small even subgraphs of the
core $\GC$.
For the sake of simplicity, we shall drop the subscript
$\mathtt{C}$ throughout
the subsection.

Fix $k$. Consider some $\alpha>2/k$. Let $\eta_* >0$ be defined
implicitly by
%
%
\begin{equation}
\frac{\eta_*(e^{\eta_*} -1)}{e^{\eta_*} - 1 - \eta_*}
= \alpha k. \label{eq:eta_star_defn}
\end{equation}
For $\alpha\in(2/k, \infty)$, we have $\eta_*(\alpha) > 0$ and
$\eta_*$ is an increasing function of $\alpha$ at fixed~$k$ \cite{DemboFSS}.

Consider a graph $G=(F,V,E)$ drawn uniformly at random from $\C(n, k,\break 
\alpha n; 0,n)$.
The rationale for this definition of $\eta_*$ is that the asymptotic
degree distribution of variable nodes in $G$ is
$\Po(\eta_*)$ conditioned on the outcome being greater than or equal
to $2$ [to be denoted below $\Po_{\ge2}(\eta_*)$].

We are interested in even subgraphs of $G$.

Consider the subgraph $G_2=(F, V^{(2)}, E^{(2)})$ of $G$ induced by
\textit{variable} nodes of degree $2$ (with all factor nodes retained).
The asymptotic branching factor this subgraph turns out to be $\theta_2
\equiv\eta_*(k-1)/(e^{\eta_*} -1)$.
We impose the condition $\theta_2 \leq1 -\delta$ for some $\delta>0$
(since this is true of the core). Note that $\theta_2$ is a decreasing
function of $\eta_*$, and hence a decreasing function of $\alpha$, for
fixed $k$.

First, we state a technical lemma that we find useful.

%
\begin{lemma}\label{lemma:nbrhood_bound}
Consider any $k$, any $\alpha\in(2/k, 1]$ and $\eps\in(0,1]$.
Then there exists $N_0 \equiv N_0(k, \eps) < \infty$ and $\const=
\const(k) < \infty$ such that the following occurs for all $n >
N_0$. Consider a graph $G=(F,V,E)$ drawn uniformly at random from
$\C(n, k, m;0,n)$, $m=n\alpha$. With probability at least $1-1/n$,
there is
no subset of variable nodes $V' \subseteq V$ such that $|V'| \leq\eps
n$ and the sum of the degrees of nodes in $V'$ exceeds $ \const\eps
\log(1/\eps) n$.
\end{lemma}

\begin{pf}
Let $\deg(i)$ be the degree of variable node $i \in V$. Let $X_i \sim
\Po_{\ge2}(\eta_*)$ be i.i.d.
for $i\in V$.
Then $(\deg(i))_{i=1}^n$ is distributed as $(X_i)_{i=1}^n$,
conditioned on $\sum_{i=1}^n X_i = mk$. Consider $V' = \{1, 2, \ldots,
l \}$. We have
\begin{eqnarray*}
\prob \Biggl\{\sum_{i=1}^l \deg(i) \geq
\gamma l \Biggr\} &= & \prob \Biggl\{\sum_{i=1}^l
X_i \geq\gamma l \bigg| \sum_{i=1}^n
X_i = mk \Biggr\}
\\
&\leq& \frac{
\prob \{\sum_{i=1}^l X_i \geq\gamma l  \}}{
 \prob \{\sum_{i=1}^n X_i = mk  \}}.
\end{eqnarray*}

Now, $n \E[X_i] = n \alpha k = mk$, by our choice of $\eta_*$ in
equation~\eqref{eq:eta_star_defn}. Since $\alpha\leq1$, we deduce that
$\eta_* \leq\const_1 = \const_1(k) < \infty$.
Using a local central limit theorem (CLT) for lattice random variables
(Theorem~5.4 of \cite{HallCLT}) we obtain $\prob \{\sum_{i=1}^n X_i
= mk  \} \geq\const_2 n^{-1/2}$ for some $\const_2 = \const_2(k)
>0$. A standard Chernoff bound yields
$\prob \{\sum_{i=1}^l X_i \geq\gamma l  \} \leq\exp \{ -
l \gamma\const_3  \}$, for some $\const_3(k) \in(0,1]$, provided
$\gamma> 2 \alpha k$. Thus, we obtain
%
%
\begin{equation}
\prob \Biggl\{\sum_{i=1}^l \deg(i) \geq
\gamma l \Biggr\} \leq n^{1/2} \exp \{ - l \gamma\const_3 \}
/ \const_2,
\end{equation}
provided $\gamma> 2 \alpha k$. We use $\gamma= \const' (1+ \log
(1/\eps
))$ with $\const' = 2 \alpha k/\const_3$. Take $l = \eps n$. The number
of different subsets of variable nodes of size $l$ is
${n\choose l} \leq(e/\eps)^l$ for $n \geq N_1$ for some $N_1 = N_1
(\eps) < \infty$. A union bound gives the desired result.
\end{pf}

%
\begin{lemma}
\label{lemma:even_subgraphs_characterization}
Fix $k\ge3$, and $\delta>0$ so that for any $\alpha\in
[2/k+\delta,1]$, we have $\theta_2(\alpha, k) \leq1-\delta$. Then,
for any $\delta'>0$, there exists $\eps=
\eps(\delta, k) >0$, $\const= \const(\delta, \delta', k) < \infty
$ and
$N_0 = N_0(\delta, \delta', k)< \infty$ such
that the following occurs for every $n > N_0$. Consider a graph $G=(F,V,E)$
drawn uniformly at random from $\C(n, k, \alpha n;0,n)$. With
probability at
least $1-\delta'$, both the following hold:
\begin{longlist}[(ii)]
\item[(i)] Consider minimal even
subgraphs consisting of only degree $2$ variable nodes. There are no
more than
$\const$ such subgraphs. Each of them is a simple cycle consisting of
no more
than $\const$ variable nodes.

\item[(ii)] Every even subgraph of $G$ with
less than $\eps n$ variable nodes contains only degree $2$ variable nodes.
\end{longlist}
\end{lemma}
\begin{pf}
\emph{Part} (i):
Reveal the $m k$ edges of $G$ sequentially. The expected number of
nodes in
$V^{(2)}$, conditioned on the first $t$ edges revealed forms a
martingale with
differences bounded by $2$. Then, from Azuma--Hoeffding inequality
\cite{PanconesiBook}, we deduce that $|V^{(2)}|$ concentrates around
its expectation:
\[
\prob \bigl( \bigl| \bigl|V^{(2)}\bigr| - \E\bigl[\bigl|V^{(2)}\bigr|\bigr] \bigr| \geq\zeta
\sqrt {n} \bigr) \leq\exp\bigl( - \wh{\const}_1 \zeta^2
\bigr)
\]
for all $\zeta>0$, where $\wh{\const}_1 = \wh{\const}_1(k) > 0$.
The expectation can be computed for instance using the
Poisson representation as in the proof of Lemma~\ref
{lemma:nbrhood_bound}, yielding $|\E|V^{(2)}| -n
\eta_*^2/(2(e^{\eta_*}-1-\eta_*))| \le n^{3/4}$, for all $\alpha<1$,
$n\ge\wh{N}_0(k)$.
We deduce that for any $\delta_1 = \delta_1(\delta,k)>0$, we have
%
%
\begin{equation}
\prob \bigl(\bigl | \bigl|V^{(2)}\bigr|/n - \eta_*^2/\bigl(2
\bigl(e^{\eta_*}-1-\eta_*\bigr)\bigr)\bigr | \geq\delta_1 n \bigr)
\leq1/n
\end{equation}
for all $n>\wh{N}_1$, where $\wh{N}_1 = \wh{N}_1(\delta, k)< \infty$.

Now, condition on $|V^{(2)}|=n^{(2)}$, for some $n^{(2)}$ such that
%
%
\begin{equation}
\bigl| n^{(2)}/n - \eta_*^2/\bigl(2\bigl(e^{\eta_*}-1-
\eta_*\bigr)\bigr) \bigr| < \delta_1 n. \label{eq:n2_constraint}
\end{equation}
Note that by choosing $\delta_1$ small enough, we can ensure $n^{(2)} =
\Omega(n)$.
We are now interested in the check degree distribution $R^{(2)}$ in
$G_2$. Reveal the $2 n^{(2)}$ edges of $G_2$ sequentially. Consider $l
\in\{0,1, \ldots, k\}$. The expected number of check nodes with degree
$l$ in $G_{2}$, conditioned on the edges revealed thus far, forms a
martingale with differences bounded by $2$. Let $Z \sim\mbox
{Binom}(k,2n^{(2)}/(mk))$. We have $\E[R^{(2)}_l] = \prob(Z=l) +
O(1/n)$. Arguing as above for each $l \leq k$, we finally obtain
%
%
\begin{equation}
\prob \Biggl( \sum_{l=0}^k \bigl|
R^{(2)}_l - \prob(Z=l) \bigr| \geq \delta_1 n
\Biggr) \leq1/n
\end{equation}
for all $n>\wh{N}_2$, where $\wh{N}_2 = \wh{N}_2(\delta, k)< \infty$.

Now condition on both $n^{(2)}$ satisfying equation~\eqref{eq:n2_constraint}
and $R^{(2)}$ satisfying
\[
\sum_{l=0}^k\bigl | R^{(2)}_l
- \prob(Z=l)\bigr | < \delta_1.
\]
Let $\zeta$ be the branching factor of $G_2$ (i.e., of a graph that is
uniformly random conditional on the degree profile $R^{(2)}$).
Under the above conditions on $n^{(2)}$ and $R^{(2)}$, a straightforward
calculation implies that $\zeta$ is bounded above by $\theta_2 +
\delta
_2$, for some
$\delta_2 = \delta_2(\delta_1,k)$ such that $\delta_2 \rightarrow
0$ as
$\delta_1 \rightarrow0$. Thus, by selecting appropriately small
$\delta
_1$, we
can ensure that $\delta_2 \leq\delta/2$, leading to a bound of
$1-\delta/2$ on
the branching factor for all $n^{(2)}$, $R^{(2)}$ within the range specified
above.

Now we condition also on the degree sequence, that is, the sequence of check
node degrees in $G_2$. The factor graph $G_2$ can be naturally
associated to a graph, by replacing each variable node by an edge and
each check node by a vertex. This graph is distributed according to
the standard (nonbipartite) configuration model. Using \cite
{Wormaldshortcycles81},
Theorem~4, we obtain that the number of cycles of
length $l \in\{1, 2, \ldots, l_0\}$ for a constant $l_0$ are
asymptotically independent Poisson random variables, with
parameters\footnote{The model in \cite{Wormaldshortcycles81}
is slightly different from the configuration model for its treatment
of self-loops and double edges. However, the results and proof can be
adapted to the configuration model.}
\[
\lambda_l = \zeta^l /(2l)\qquad \mbox{for } \zeta= \Biggl[
\sum_{d=1}^k d(d-1)R^{(2)}(d)
\Biggr] \bigg/ \Biggl[\sum_{d=1}^k d
R^{(2)}(d) \Biggr].
\]
More precisely, for any constants $c_1, c_2, \ldots, c_{l_0} \in\N
\cup\{0\}$, we have
\[
\prob\bigl[\Ev_n(\uc)\bigr] = \prod_{l=1}^{l_0}
\prob\bigl(\Po(\lambda_l) = c_l\bigr) + o(1) ,
\]
where $\Ev_n(\uc)$ is the event that there are $c_l$ cycles of length
$l$ for $l \in\{1, 2, \ldots, l_0\}$ with all cycles disjoint from each
other, and $\uc=(c_l)_{l=1}^{l_0}$. Choosing $l_0$ large enough, we have
\[
\sum_{\uc\in\cN} \prob\bigl[\Ev_n(\uc)\bigr]
\geq1-\exp \Biggl(-\sum_{l=1}^\infty
\lambda_l \Biggr) - \delta/4 = 1-(1-\zeta)^{-1/2} -
\delta'/4,
\]
where $\cN= \{\uc\dvtx\uc\neq\underline{0}, c_l \leq l_0 \mbox{
for }
l\in\{1,2, \ldots, l_0\}\}$, for $n$ large enough.

On the other hand, we know that the probability of having no cycles in
$G_2$ is
$ (1-\zeta)^{-1/2} + o(1)$ under our assumption of $\zeta\leq1-
\delta/2$.
The argument for this was already outlined in the proof of Lemma~\ref
{lemma:peelability_implies_good}, cf. Section~\ref{sec:Peel12}:
the Poisson approximation of \cite{Wormaldshortcycles81} is used to
estimate the probability of having no cycles of length smaller than
$M$, while a simple first moment bound is sufficient for cycles of
length $M$ or larger.
Thus, with probability at least $1-\delta'/3$, we have no more than
$l_0^2$ cycles, disjoint and each of length no more than $l_0$.
Choosing $\const= l_0^2$, we obtain part (i) with probability at
least $1-\delta'/2$ for large enough $n$.

\emph{Part} (ii):
Let $m \equiv\alpha n$. Let $\N(G;l,j)$ be the number of even
subgraphs of $G$ induced by $l$ variable nodes such that the sum of
the degrees of the $l$ variable nodes is $2(l+j)$. We are interested
in $l \leq\eps n$ (we will choose $\eps$ later) and $j>0$. In
particular, we want to show that, for any $\delta'>0$,
%
%
\begin{equation}
\prob \Biggl\{\sum_{l=1}^{\eps n} \sum
_{j=1}^{mk/2} \N(G;l,j)>0 \Biggr\} \leq
\delta'/2.
\end{equation}
This immediately implies the desired result from linearity of
expectation and Markov inequality.

From Lemma~\ref{lemma:nbrhood_bound}, we deduce that
%
%
\begin{equation}
\prob \Biggl\{\sum_{l=1}^{\eps n} \sum
_{j=\eps' n}^{mk/2} \N(G;l,j)>0 \Biggr\} \leq1/n,
\end{equation}
for some $\eps'(\eps,k)$ with the property that $\eps' \rightarrow0$
as $\eps\rightarrow0$. Thus, we only need to establish
%
%
\begin{equation}
\sum_{l=1}^{\eps n} \sum
_{j=1}^{\eps' n} \E\bigl[\N(G;l,j)\bigr] \leq
\delta'/3, \label{eq:sufficient_bound}
\end{equation}
for all $n$ large enough, since the claim then follows from Markov inequality.

A straightforward calculation \cite{RiUBOOK,MM09} yields
\[
\E\bigl[\N(G;l,j)\bigr] = \frac{ {n\choose l} \term_1 \term_2 \term_3 }{
 {mk\choose2(l+j)} \term_4 },
 \]
where
\begin{eqnarray*}
\term_1 &=& \operatorname{coeff}\bigl[ \bigl(e^y -1-y
\bigr)^l ; y^{2(l+j)} \bigr],
\\
\term_2 &=& \operatorname{coeff}\bigl[ \bigl(e^y -1-y
\bigr)^{n-l} ; y^{mk -2(l+j)} \bigr],
\\
\term_3 &= &\operatorname{coeff}\biggl[ \biggl( \frac{(1+y)^k + (1-y)^k
}{2}
\biggr)^m ; y^{2(l+j)} \biggr],
\\
\term_4 &=& \operatorname{coeff}\bigl[ \bigl(e^y-1-y
\bigr)^n ; y^{mk} \bigr].
\end{eqnarray*}

It is useful to recall the following probabilistic representation of
combinatorial coefficients.
%

\begin{fact}\label{fact:coeff_poisson}
For any $\eta> 0$, we have
%
%
\begin{equation}
{\operatorname{coeff}\bigl[ \bigl(e^y -1 -y\bigr)^N;
y^M \bigr]} = \eta^{-M}\bigl(e^\eta-1 -\eta
\bigr)^N \prob \Biggl[\sum_{i=1}^N
X_i =M \Biggr],
\end{equation}
where $X_i \sim\Po_{\geq2} (\eta)$ are i.i.d. for $i\in\{1,\ldots,
M\}$.
\end{fact}

Consider $\term_4$. By definition, cf. equation~\eqref{eq:eta_star_defn},
$\eta_*$ is such that for
$X_i \sim\Po_{\geq2} (\eta_*)$ we have $\E[X_i] = \alpha k =
mk/n$. Moreover,
$\alpha\in[2/k+\delta,1]$ implies $\eta_* \in[\const_1, \const_2]$
for some $\const_1= \const_1(\delta, k) > 0$ and $\const_2= \const_2(
k) < \infty$. From $\eta_* \leq\const_2$ and using a local CLT for
lattice random variables \cite{HallCLT}, it follows that $\prob[\sum_{i=1}^n X_i = mk] \geq\const_3 /\sqrt{n}$ for some $\const_3 =
\const
_3(\delta,k)>0$. Thus, using Fact~\ref{fact:coeff_poisson}, we have
%
%
\begin{equation}
\term_4 \geq\eta_*^{-mk}\bigl(e^{\eta_*} -1 -\eta_*
\bigr)^n \const_3 n^{-1/2}. \label{eq:term4_bound}
\end{equation}

Now, consider $\term_2$. Again use $\eta= \eta_*$ in Fact~\ref
{fact:coeff_poisson}. From $\eta_* \geq\const_1$ and again using a
local CLT for lattice r.v.'s \cite{HallCLT}, we obtain $\prob[\sum_{i=1}^{n-l} X_i = mk-2(l+j)] \leq\const_4 /2\sqrt{n-l} \leq\const
_4/\sqrt{n}$ for some $\const_4= \const_4(\delta, k) < \infty$, since
$l \leq\eps n$. Thus, Fact~\ref{fact:coeff_poisson} yields
%
%
\begin{equation}
\term_2 \leq\eta_*^{-mk + 2(l+j)}\bigl(e^{\eta_*} -1 -\eta_*
\bigr)^{n-l} \const _4 n^{-1/2}. \label{eq:term2_bound}
\end{equation}

Fact~\ref{fact:coeff_poisson} yields that $\term_1$ can be bounded
above as
%
%
\begin{equation}
\term_1 \leq\eta^{-2(l+j)}\bigl(e^{\eta} -1 -\eta
\bigr)^{l} \label{eq:term1_bound}
\end{equation}
for any $\eta>0$. We will choose a suitable $\eta$ later.

Finally, for $\term_3$, similar to Fact~\ref{fact:coeff_poisson}, we
can deduce that
\[
\term_3 \leq \biggl( \frac{(1+\xi)^k +(1-\xi)^k }{2} \biggr)^m \xi
^{-2(l+j)}
\]
for all $\xi> 0$.
Now, it is easy to check that
\[
\frac{(1+\xi)^k +(1-\xi)^k }{2} \leq\exp \left\{\pmatrix{k
\cr
2} \xi^2 \right\},
\]
by comparing coefficients in the series expansions of both sides.
Choosing $\xi= \sqrt{(l +j)/(m{k\choose2})}$, we obtain
%
%
\begin{equation}
\term_3 \leq \biggl( \frac{em {k\choose2} }{l+j} \biggr)^{l+j}.
\label{eq:term3_bound}
\end{equation}
Finally, we have
%
%
\begin{equation}
\pmatrix{n
\cr
l} \leq\frac{n^l}{l!}, \qquad \pmatrix{mk
\cr
2(l+j)} \geq
\frac{(mk-2(l+j))^{2(l+j)}}{(2(l+j))!}. \label{eq:binom_coeff_bound}
\end{equation}
Putting together equations~\eqref{eq:term4_bound}, \eqref{eq:term2_bound},
\eqref{eq:term1_bound}, \eqref{eq:term3_bound} and \eqref
{eq:binom_coeff_bound}, we obtain
\begin{eqnarray*}
\E\bigl[\N(G;l,j)\bigr]  &\leq&\const_6 \cdot \frac{(e^{\eta} -1 - \eta)^l}{\eta^{2(l+j)}}
\cdot \frac{\eta_*^{2(l+j)}}{(e^{\eta_*} -1 - \eta_*)^l} \\
&&{}\times \biggl( \frac{e (k-1) (1 + C_5((l+j)/n))}{2(l+j)k} \biggr)^{l+j}
\cdot \frac{(2(l+j))!}{l! \alpha^l m^j},
\end{eqnarray*}
for some $\const_6 = \const_6(k, \delta) < \infty$.
Now, $N! \geq\const_7 \sqrt{N} (N/e)^N$ for all $N \in\mathbb{N}$,
for some $\const_7 >0$. Using this with $N=l+j$, we obtain
\begin{eqnarray*}
\biggl( \frac{e}{l+j} \biggr)^{l+j} \cdot\frac{(2(l+j))!}{l!} &\leq&
\frac{\sqrt{l+j}}{C_7} \cdot\frac{(2(l+j))!}{l!(l+j)! } \\
&\leq& \frac{\sqrt{l+j}}{C_7 l^j} \cdot
\pmatrix{2(l+j)
\cr
(l+j)} \leq \frac{ \const_8 2^{2(l+j)}}{l^j},
\end{eqnarray*}
for some $\const_8 < \infty$. Plugging back, we get
\[
\E\bigl[\N(G;l,j)\bigr]  \leq\const_9 (\term_5)^l
(\term_6)^j,
\]
where
\begin{eqnarray*}
\term_5 &=& 2 \theta_2 \frac{(e^{\eta} -1 - \eta)}{\eta^2} \bigl( 1 +
C_5\bigl((l+j)/n\bigr) \bigr),
\\
\term_6 &=& \frac{4(k-1) \eta_*^2}{m l \eta^2}.
\end{eqnarray*}
Without loss of generality, assume $\delta\leq0.1$. Now, we choose
$\eps= \eps(\delta, k)>0$
such that $\eps+ \eps' \leq\delta/(10\const_{5})$. We
choose $\eta= \eta(k)>0$ such that $(e^{\eta} -1 - \eta)\eta^{-2}
\leq
(1+\delta/10)/2$ [note that $(e^{\eta} -1 - \eta)\eta^{-2}
\rightarrow
1/2$ as
$\eta\rightarrow0$]. This leads to $\term_5 \leq1- \delta/2$ for all
$l \leq
\eps n$ and $j \leq\eps' n$, when we use $\theta_2 \leq1- \delta$. Also,
$\term_6 \leq\const_{10}/n$ for all $l$, $j$, for some $\const_{10}=
\const_{10}(k) < \infty$. Thus,
\[
\E\bigl[\N(G;l,j)\bigr] \leq\const_9 (1-\delta/2)^l
\biggl(\frac{\const
_{10}}{n} \biggr)^j.
\]
Summing over $j$ and $l$, we obtain
%
%
\begin{equation}
\sum_{l=1}^{\eps n} \sum
_{j=1}^{\eps' n} \E\bigl[\N(G;l,j)\bigr] \leq
\frac{\const_{11}}{n}
\end{equation}
for some $\const_{11}= \const_{11}(k, \delta) < \infty$. This implies
equation~\eqref{eq:sufficient_bound} for large enough $n$ as required.
\end{pf}

\section{Proof of Lemma \texorpdfstring{\lowercase{\protect\ref{lemma:core_sparse_basis}}}{3.8}:
A sparse basis for low-weight core solutions}
\label{app:core_basis}

For each $\uxC\in\Lcore(\ve n)$, we need to find
a sparse solution $\ux\in\cS_1$ that matches $\uxC$ on the
core. From Lemma~\ref{lemma:core_few_low_weight},
we know that w.h.p., $\uxC$ consists of all zeros except for a small
subset of variables. Indeed, we know from Lemma~\ref
{lemma:even_subgraphs_characterization}
that these variables correspond to a cycle of degree-$2$ variable
nodes. Although this is not used in the following, we shall
nevertheless refer to the set of variable nodes corresponding to an
element of $\Lcore(\ve n)$ as a cycle. Denote by $L_1$ the cycle
corresponding to $\uxC$. Recall that the noncore
$\GNC= (\FNC, \VNC, \ENC)$ is the subgraph of $G$ induced by $\FNC
= F
\setminus\FC$ and $\VNC= V \setminus\VC$. Suppose we set all
noncore variables to $0$. The set of violated checks consists of those
checks in $\FNC$ that have an odd number of neighbors in $L_1$.
We show that w.h.p., each such check can be satisfied by changing a
small number of noncore variables in its neighborhood to 1. To show
that this is possible, we make use of the belief propagation
algorithm described in Section~\ref{sec:BP_DE}.

Our strategy is roughly the following. Consider a violated check $a$.
We wish to set an odd number of its noncore neighboring variables to
$1$. But then, this may cause further checks to be violated, and so on.
A key fact comes to our rescue. If check node $a$ receives an incoming
$*$ message in round $T$, then we can find a subset of noncore variable
nodes in a $T$-neighborhood of $a$ such that if we set those variables
to $1$, check $a$ will be satisfied (with an odd number of neighboring
ones in the noncore) without causing any new violations. We do this for
each violated check. Now w.h.p., for suitable $T$, all violated checks
will receive at least one incoming $*$ by time $T$ (note that each
noncore check receives an incoming $*$ at the BP fixed point). Thus, we
can satisfy them all by setting a small number of noncore variables to $1$.

%
\begin{lemma}\label{lemma:NC_uniform_conditioned_on_GC_ECNC}
Consider $G$ drawn uniformly from $\G(n,k,m)$. Denote by $F^{(l)}
\subseteq\FNC$ the checks in the noncore having degree $l$ with
respect to the noncore, for $l\in\{1,2, \ldots, k\}$. Condition on the
core $\GC$, and $F^{(l)}$ for $l\in\{1,2, \ldots, k\}$. 
%
\begin{itemize}
\item
Then $\ECNC$ and $\GNC$ are independent of each other.
Here $\ECNC$ denotes the edges between core variables $\VC$ and noncore
checks $\FNC$.
\item
The edges in $\ECNC$ are distributed as follows: For each $a \in\FNC$,
if $a \in F^{(l)}$, its neighborhood in $\GC$ is a uniformly random
subset of $\VC$ of size $k-l$, independent of the others.
\item
Clearly, $(\GC, (F^{(l)})_{l=1}^k)$ uniquely determine the parameters
$(\nNC, \RNC, \mNC)$ of the noncore. The noncore $\GNC$ is drawn
uniformly at random from $\D(\nNC, \RNC,  \mNC)$ conditioned on being
peelable, that is, $\GNC$ is drawn uniformly at random from $\D(\nNC,
\RNC, \mNC)\cap\cP$.
\end{itemize}
\end{lemma}
\begin{pf}
Each $G\in\G(n,k,m)$ with the given $(\GC, (F^{(l)})_{l=1}^k)$ has a
$\GNC$ corresponding to a unique element of $\D(\nNC, \RNC, \mNC
)\cap
\cP$ and $\ECNC$ corresponding to a subset of $\VC$ of size $k-l$ for
each $a \in F^{(l)}$, for $l \in\{ 1, \ldots, k\}$. The converse is
also true. This yields the result.
\end{pf}

\begin{pf*}{Proof of Lemma~\ref{lemma:core_sparse_basis}}
Take any sequence $(s_n)_{n \geq1} $ such that\break 
$\lim_{n\rightarrow\infty} s_n = \infty$ and $s_n \leq\ve n$.
If points (i), (ii) and (iii) in Lemma~\ref
{lemma:core_few_low_weight} hold, let
$V_{\mathrm{ cycle}}$ denote the
union of the supports of the solutions in $\Lcore(s_n)$.
Let
\begin{eqnarray*}
\Ev_1 & \equiv&\Ev_{1,a} \cap\Ev_{1,b}\cap
\Ev_{1,c},
\\
\Ev_{1,a}& \equiv&\bigl\{ \mbox{Points (i), (ii) and (iii) in
Lemma~\ref{lemma:core_few_low_weight} hold} \bigr\},
\\
\Ev_{1,b}& \equiv&\bigl\{ \mbox{$\bigl|F^{(l)}\bigr| \geq n/
\const_2$ for all $l \in\{1, 2, \ldots, k\}$} \bigr\},
\\
\Ev_{1,c}& \equiv&\{ \mbox{No variable in $V_{\mathrm{ cycle}}$ has
degree exceeding $\log s_n$} \}.
\end{eqnarray*}
(Note that these events are implicitly indexed by $n$.)
We argue that $\Ev_1$ holds w.h.p. for an appropriate choice of
$\const_2 =\const_2(k,\alpha)< \infty$. Indeed, Lemma~\ref
{lemma:core_few_low_weight} implies that $\Ev_{1,a}$ holds w.h.p.
Lemma~\ref{lemma:FP_concentration_DE} implies that $\Ev_{1,b}$ holds
w.h.p. for sufficiently large $\const_2$.
Finally, Lemma~\ref{lemma:NC_uniform_conditioned_on_GC_ECNC} and a
subexponential tail bound on the Poisson distribution ensure $\Ev
_{1,c}$ holds w.h.p.

Assume that $\Ev_1$ holds.
Let sets of variable nodes on the disjoint cycles corresponding to
elements of $\Lcore(\ve n)$ be denoted by $L_i$ for $i\in\{1, 2,
\ldots
, |\Lcore(\ve n)|\}$. Consider a cycle $L_i$. Denote by $a_{ij}$, $j
\in\{1, 2, \ldots, Z_i\}$, the checks in the noncore having an odd
number of neighbors in $L_i$. (Thus, $Z_i$ is the number of such
checks.) Call these \emph{marked} checks. Given $\Ev_1$, we know that
$Z_i \leq s_n \log s_n$, and that there are no more than $s_n^2 \log
s_n$ marked checks in total:
\[
\sum_{i=1}^{|\Lcore(\ve n)|} Z_i \leq
s_n^2 \log s_n.
\]

Define
\[
\Ev_2 \equiv\bigl\{ \mbox{No more than $n /s_n^3$
messages change after $T_n$ iterations of \BPzero} \bigr\}.
\]
By Lemma~\ref{lemma:BP_few_late_changes}, the event $\Ev_2$ holds
w.h.p.
provided $\lim_{n \rightarrow\infty} T_n = \infty$ and $s_n$ grows
sufficiently slowly with $n$ [for the given choice of $(T_n)_{n\geq1} $].

Let
\begin{eqnarray*}
&&B_{ij} \equiv\{ \mbox{Not all messages incoming to check
$a_{ij}$ have converged}\\
&&\hspace*{99pt} \mbox{to their fixed-point value in
$T_n$ iterations} \}.
\end{eqnarray*}

We wish to show that
%
%
\begin{equation}
\label{eq:no_bad_marked_checks} \bigcap_{i,j} B_{ij}^{\mathrm c}
\end{equation}
holds w.h.p.
We have
\[
\prob \biggl(\bigcup_{i,j} B_{ij}
\biggr) \leq \E_{(\GC, \ECNC)} \biggl[ \E \biggl[\ind[\Ev_1,
\Ev_2]\sum_{i,j} \ind[B_{ij}]
\Big| \GC, \ECNC \biggr] \biggr] + \prob\bigl[\Ev _1^{\mathrm
c}
\bigr]+ \prob\bigl[\Ev_2^{\mathrm c}\bigr].
\]

Given $\Ev_2$, we know that the number of checks for which an incoming
message changes after $T_n$ is no more than $n/s_n^3$. Suppose $a_{ij}
\in F^{(l)}$ is a marked check. Then we have
\[
\E \bigl[ \ind[\Ev_1, \Ev_2] \ind[B_{ij}] |
\GC, \ECNC \bigr] \leq\frac{n}{s_n^3 |F^{(l)}|} \leq\frac{1}{\const_2 s_n^3},
\]
since all check nodes in $F^{(l)}$ are equivalent with respect to the
noncore, from Lemma~\ref{lemma:NC_uniform_conditioned_on_GC_ECNC}.
We already know that under $\Ev_1$, the number of marked checks is
bounded by $s_n^2 \log s_n$. This leads to
\[
\prob \biggl(\bigcup_{ij}B_{ij} \biggr)
\leq\frac{\log s_n}{\const_2
s_n} + \prob\bigl[\Ev_1^{\mathrm c}\bigr]+
\prob\bigl[\Ev_2^{\mathrm c}\bigr] \stackrel{n \rightarrow \infty}
{\longrightarrow} 0,
\]
implying equation~\eqref{eq:no_bad_marked_checks} holds w.h.p.

Condition on $\GC$ and $\ECNC$. This identifies the marked
checks. Lemma~\ref{lemma:NC_uniform_conditioned_on_GC_ECNC} guarantees
us that all checks in $F^{(l)}$ are equivalent with respect to
$\GNC$. Suppose $\Ev_1$ holds. Define a ball of radius $t$ around a
check node as consisting of the neighboring variable nodes, and the
balls of radius $t$ around each of those variables. Similar to the
proof of
Lemma~\ref{lemma:peelability_implies_good}(iii), we can show that
%
%
\begin{equation}
\bigl|\Ball_{\GNC}(a_{ij}, T_n)\bigr| \leq
\const_3^{T_n} \label{eq:max_ball_size}
\end{equation}
holds with probability at least $1 - \const_4 \exp(- 2^{T_n}/\const_4)$,
for some $\const_3 = \const_3(\alpha, k) < \infty$ and
$\const_4 = \const_4(\alpha, k)< \infty$, for all marked checks
$a_{ij}$. Thus, the probability that this bound on ball size holds
simultaneously for all marked checks, by union bound, is at least $1 -
s_n^2 \log s_n \const_4 \exp(- 2^{T_n}/\const_4) \rightarrow1$ as $n
\rightarrow1$ provided $T_n \rightarrow\infty$ and $s_n$ grows
sufficiently slowly with $n$.

Suppose equation~\eqref{eq:no_bad_marked_checks} and $\Ev_1$ hold. Consider
any marked check $a_{ij}$ adjacent to $v \in L_i$ for any $L_i$. It
receives at least one incoming $*$ message at the \BPzero\ fixed point
and since $B_{ij}=0$, this is also true after $T_n$ iterations of
\BPzero. Hence, there is a subset of variables $V^{(ij)} \subseteq
\Ball
_{\GNC}(a_{ij},T_n)$, such that setting variables in $V^{(ij)}$ to $1$
satisfies $a_{ij}$ without violating any other checks. Define
\[
V^{(i)} \equiv\bigl\{v\dvtx v \mbox{ occurs an odd number of times in the
sets } \bigl(V^{(ij)}\bigr)_{j=1}^{Z_i} \bigr\}.
\]
It is not hard to verify that the vector $\ux_{{\mathrm c},i}$
with variables in $L_i \cup V^{(i)}$ set to one and all other variables
set to zero, is a member of $\cS_1$. If equation~\eqref{eq:max_ball_size}
holds for all marked checks, then we deduce that $|V^{(i)}| \leq\const
_3^{T_n} s_n \log s_n \leq c_n$ for $T_n$ and $s_n$ growing
sufficiently slowly with $n$. Thus, $\ux_{{\mathrm c},i} \in\cS_1$ is
$c_n$-sparse assuming these events, each of which occurs w.h.p. We
repeat this construction for every $L_i$.
\end{pf*}

\begin{appendix}\label{app}
\section{Proof of Lemma \texorpdfstring{\lowercase{\protect\ref{lemmabasic}}}{3.4}}
\label{app:Sparse_and_basis}

%
\begin{lemma}\label{lemma:BasicConstructionDistance}
Assume that $G$ has no $2$-core, and let
\[
\bK\equiv\left[ %
\pmatrix{\bH_{F,U}^{-1}
\bH_{F,W}
\vspace*{2pt}\cr
I_{(n-m)\times
(n-m)} }
 \right],
\]
where $U$ and $W$ are constructed as in Lemma~\ref{lemma:structural},
we order
the variables as $U$ followed by $W$, and the matrix inverse is taken over
$\GF[2]$. Then the columns of $\bK$ form a basis of the kernel of
$\cS
$, which
is also the kernel of $\bH$. In addition, if $\bK_{i,j}=1$, then $d_G(i,j)
\leq\TC$.
\end{lemma}

\begin{pf}
A standard linear algebra result shows that $\bK$ is a basis for the kernel
of $\bH$. The bottom identity block of $\bK$ corresponds to the $(n-m)$
independent variables $w \in W$, and in this block a $1$ only occurs if
the row
and column correspond to the same variable, that is, for $i,j \in W$,
$\bK_{i,j}=1$
implies $i=j$, and thus $d_G(i,j)=0$. To prove the distance claim for the
upper block of $\bK$, we proceed by induction on $\TC$. For a variable
$u\in
U$ that is peeled along with factor node $a\in F$, we will reference
$u$ via
the factor node it was peeled with as $u_a$.
\begin{itemize}
\item\textit{Induction base}:
For $\TC=1$, $\bH_{F,U} = I_m$, and thus
\[
\bK= \left[ \pmatrix{\bH_{F,W}
\vspace*{2pt}\cr
I_{(n-m)\times
(n-m)} }
 \right].
\]
Since $\TC=1$, note that ever variable node must be connected to no
more than~$1$ factor node. Thus, $(\bH_{F,W})_{a,i}=1$ implies that factor node
$a$ was
connected to independent variable node $i$. Thus, variables $i$ and
$u_a$ are
both adjacent to factor~$a$, and consequently $d_G(u_a,i)=1$.

\item\textit{Inductive step}:
Assume that $\TC=T+1$ and consider the graph $\peel(G) =
(F_\peel,V_\peel,E_\peel)$ (recall that $\peel$ denoted the peeling
operator).
By construction $\TC(\peel(G))=T$, and thus by the
inductive hypothesis the columns of
\begin{eqnarray*}
\bK_{\peel(G)} &\equiv&\left[ \matrix{\tilde{\bK}
\vspace*{2pt}\cr
I_{((n-n_1)-(m-m_1))\times((n-n_1)-(m-m_1))} }
 \right]\\
 & \equiv&\left[
\matrix{\bH_{F_\peel,U_\peel}^{-1}
\bH_{F_\peel,W_\peel}
\vspace*{2pt}\cr
I_{((n-n_1)-(m-m_1))\times((n-n_1)-(m-m_1))} }
 \right],
\end{eqnarray*}
form a basis for the kernel of $\bH_{\peel(G)}$,
where $F_\peel$, $U_\peel$, and $W_\peel$ refer to the set of factor
nodes of
the factor graph $\peel(G)$, and their corresponding partition, respectively.
In addition, $(\bK_{\peel(G)})_{a,i} = 1$ only if $d_{\peel
(G)}(u_a,i)\leq T$.
To extend this basis to a basis for the kernel of $\bH$, note that
\begin{eqnarray*}
\bK\equiv\left[ %
\matrix{\bH_{F,U}^{-1}
\bH_{F,W}
\vspace*{2pt}\cr
I_{(n-m)\times
(n-m)} }
 \right] &=& \left[
\matrix{
\pmatrix{
\bH_{F_1,U_1} & \bH_{F_1,U_\peel}
\vspace*{2pt}\cr
0 & \bH_{F_\peel,U_\peel}}^{-1} \pmatrix{
\bH_{F_1,W_1} & \bH_{F_1,W_\peel}
\vspace*{2pt}\cr
0 & \bH_{F_\peel,W_\peel} }
\vspace*{2pt}\cr
I_{(n-m)\times(n-m)}}
 \right]
\\
&=& \left[ \matrix{ \pmatrix {I_{|U_1|} & -\bH_{F_1,U_\peel}\bH_{F_\peel,U_\peel}^{-1}
\vspace*{2pt}\cr
0 & \bH_{F_\peel,U_\peel}^{-1}}
\pmatrix{ \bH_{F_1,W_1} & \bH_{F_1,W_\peel}
\vspace*{2pt}\cr
0 & \bH_{F_\peel,W_\peel} }
\vspace*{2pt}\cr
I_{(n-m)\times(n-m)}}
 \right]
\\
&=& \left[ \matrix{
\pmatrix{\bH_{F_1,W_1} & \bH_{F_1,W_\peel}+\bH_{F_1,U_\peel}\tilde{\bK}
\vspace*{2pt}\cr
0 & \tilde{\bK} }
\vspace*{2pt}\cr
I_{(n-m)\times(n-m)}}
 \right].
\end{eqnarray*}
By construction if $(\bH_{F_1,W_1})_{a,i}=1$, then $d_G(u_a,i)=1\leq T$.
Consider the $(a,i)$ entry of the matrix $B \equiv
\bH_{F_1,W_\peel}+\bH_{F_1,U_\peel}\tilde{\bK}$. A necessary
condition for
$B_{a,i}=1$ is the existence of an edge between check node $a \in F_1$ and
independent variable node $i\in W\setminus W_1=W_\peel$ [i.e.,
$(\bH_{F_1,W_\peel})_{a,i}=1$], or the existence of both an edge
between $a \in
F_1$ and dependent variable node $j \in U_\peel$ that is in the basis for
independent variable $i$ [i.e., ($(\bH_{F_1,U_\peel})_{a,j}=1$,
$\tilde{\bK}_{j,i}=1$)].

We note that if $d_{\peel(G)}(u_a,i)\leq T$, then $d_G(u_a,i)\leq T$
also, since
$E_\peel\subset E$. Thus, if $(\bH_{F_1,U_\peel})_{a,j}=1$,
$\tilde{\bK}_{j,i}=1$, then $d_G(u_a,i) \leq T+1$. Similarly, if
$(\bH_{F_1,W_\peel})_{a,i}=1$, then $d_G(u_a,i)=1$ as in the base case.
Thus, if $\bK_{i,j}=1$, then $d_G(i,j) \leq T+1 = \TC$.\quad\qed
\end{itemize}
\noqed\end{pf}

A direct result of this is the sparsity bound given below.

%
\begin{lemma}\label{lemma:BasicConstructionHelper}
For $\bK$ constructed as in Lemma~\ref
{lemma:BasicConstructionDistance}, the
columns of $\bK$ form an $s$-sparse basis for the kernel of $\bH$, with
\[
s \leq\max_{i\in V} \bigl|\Ball_G(i,\TC)\bigr|.
\]
\end{lemma}
\begin{pf}
By Lemma~\ref{lemma:BasicConstructionDistance}, $d_G(a,i)\leq\TC$ is a
necessary condition for $\bK_{a,i}=1$. Thus, for all $i\in W$, the $i$th
column of $\bK$ can only contain $1$'s on the entries that correspond to
variables at distance at most $\TC$ from $i$. The result follows by
taking a
union bound over all $i \in W$.
\end{pf}

\begin{pf*}{Proof of Lemma~\ref{lemmabasic}}
Let
\[
\widehat{\bK} = {\mathbb L}\left[ \matrix{
\bQ_{F_*,U_*}^{-1}\bQ_{F_*,W_*}
\vspace*{2pt}\cr
I_{(n-m)\times(n-m)}}
 \right],
\]
where the matrix inverse is taken over $\GF[2]$. If $G_* \neq G$, then
all degree $2$ check nodes constrain their adjacent variable nodes to
the same
value. Therefore, all variables in the same connected component take on
the same value in
a satisfying solution, that is, for all $v_*\in V_*$, if $\bH\ux= 0$,
then for all
$i \in v_*$, either $x_i=0$ or \mbox{$x_i=1$}. Consequently,
$\bH\ux=0$ if and only if $\ux= {\mathbb L}\ux_*$ for some $\ux_*$
such that $\bQ\ux_*=0$
Thus, $\{\ux^{(1)},\ldots,\ux^{(N)}\}$ is a basis for the kernel of
$\bH$ if and only if $\ux^{(i)} = {\mathbb L}\ux_*^{(i)}$ and
$\{\ux_*^{(1)},\ldots,\ux_*^{(N)}\}$ is a basis for the kernel of
$\bQ$.

Finally notice that ${\mathbb L}\ux_*$ has $|v_*|$ nonzero entries for each
$v_*\in V_*$ such that $\ux_{*,v_*}\neq0$. Thus, the sparsity bound
follows as a direct extension of the bound from Lemma~\ref
{lemma:BasicConstructionHelper}, and the columns of $\widehat\bK$
form an $s$-sparse basis
for the kernel of $\bH$, with
\[
s \leq\max_{v_*\in V_*}S\bigl(v_*,\TC(G_*)\bigr).
\]
\upqed\end{pf*}

\section{Proofs of technical lemmas in Section \texorpdfstring{\lowercase{\protect\ref
{sec:collapse_peeling_fast}}}{5}}
\label{app:proofs_of_technical}

\begin{pf*}{Proof of Lemma~\ref{lemma:properties_of_peelable}}
Let $\omega\equiv\alpha R'(1)$. Define $f(z) \equiv1 - \lambda
(1-\rho(z)) =
1- \exp(-\alpha R'(1)\rho(z))$. We obtain
%
%
\begin{equation}
f'(0) = 2 \alpha R_2. \label{eq:fprime0_lb}
\end{equation}
Now, we know that $z_t \rightarrow0$ as $t \rightarrow\infty$, it
follows that
$\lim_{t \rightarrow\infty} z_{t+1}/z_t \rightarrow f'(0)$.
We then deduce from peelability at rate $\eta$ that
%
%
\begin{equation}
f'(0)\leq1- \eta .\label{eq:fprime0_ub}
\end{equation}
Combining equations \eqref{eq:fprime0_lb} and \eqref{eq:fprime0_ub}, we
obtain the
desired result (i).

In order to prove (ii) notice that, for the pair to be peelable,
need $z\le1-\exp(-\alpha R'(z))$ for all $z\in[0,1]$, that is,
%
%
\begin{equation}
R'{}^{-1}(x)\le1-e^{-\alpha x}\qquad
\mbox{for all $x\in\bigl[0,R'(1)\bigr]$,} %
\end{equation}
where $R'{}^{-1}$ is the inverse mapping of $z\mapsto
R'(z)$. We next integrate the above over $[0,R'(1)]$, using
%
%
\begin{eqnarray}
\int_0^{R'(1)} R'{}^{-1}(x)
\, \de x &=& \int_{0}^1w R''(w)
\,\de w = R'(1)-1,
\\
\int_0^{R'(1)} \bigl(1-e^{-\alpha x}\bigr) \,
\de x &=& R'(1)- \frac{1}{\alpha}\bigl(1-e^{-\alpha R'(1)}\bigr).
\end{eqnarray}
We thus obtain
%
%
\begin{equation}
1\ge\frac{1}{\alpha} \bigl(1-e^{-\alpha R'(1)} \bigr),
\end{equation}
which yields $\alpha\le1-e^{-\alpha R'(1)}<1$.
\end{pf*}

\begin{pf*}{Proof of Lemma~\ref{lemma:peeling_uniform_n1_n2}}
We use the notation $\mathbf(G)=(m_l(G))_{l=2}^k$ whereby $m_l(G)$ is the
number of check nodes of degree $l$in $G$.
Let
\begin{eqnarray*}
n_{1}^{(t)}& \equiv& n_1(J_t),\qquad
n_{2}^{(t)} \equiv n_2(J_t),\qquad
\bm^{(t)} \equiv\bm(J_t),
\\
\alpha^{(t)}& \equiv& \Biggl(\sum_{l=2}^{k}
m_l^{(t)} \Biggr)\bigg/n,\qquad R_l^{(t)} \equiv
m_l^{(t)}\bigg/ \Biggl(\sum_{l'=2}^{k}
m_{l'}^{(t)} \Biggr) \qquad\mbox{for }l\in\{2,3, \ldots, k\}.
\end{eqnarray*}
Note that $R^{(t)}$ defined above is, in fact, the check degree profile
of $J_t$.

As above, let $\peel(\cdot)$
denote the operator corresponding to one round of synchronous peeling
[so that
$J_t = \peel^t(G)$]. Define the set
\[
S(G; \hat{\bm}, \hat{n}_1, \hat{n}_2) \equiv \bigl\{
\wh{G}\dvtx n_1(\wh{G})=\hat{n}_1, n_2(
\wh{G})=\hat{n}_2, \bm (\wh{G}) = \hat{\bm}, \peel(\wh{G})=G \bigr\}.
\]

We prove the result by induction. By definition, we know that $J_0= G$
is drawn
uniformly from the $\C(n,R,\alpha n)$. Suppose,
conditioned on $\bm^{(t)}, n_1^{(t)}, n_2^{(t)}$, the graph $J_t$ is
drawn uniformly from
$\C(n,R^{(t)}, \alpha^{(t)}n)$. Let the probability of each possible~$J_t$ [with
parameters $(\bm^{(t)}, n_1^{(t)}, n_2^{(t)})$] be denoted by $q(\bm
^{(t)}, n_1^{(t)}, n_2^{(t)})$.
Consider a candidate graph $G'$ with parameters $(\bm', n_1', n_2')$.
We have
\begin{eqnarray*}
\prob\bigl[J_{t+1}=G'\bigr] &=& \sum
_{J_t\dvtx\peel(J_t)=G'} \prob[J_t]
\\
&= &\sum_{\hat{\bm}, \hat{n}_1, \hat{n}_2} \sum_{J_t \in S(G'; \hat{\bm}, \hat{n}_1, \hat{n}_2)}
\prob[J_t]
\\
&=& \sum_{\hat{\bm}, \hat{n}_1, \hat{n}_2} q(\hat{\bm}, \hat{n}_1,
\hat{n}_2) \bigl|S\bigl(G'; \hat{\bm}, \hat{n}_1,
\hat{n}_2\bigr)\bigr|.
\end{eqnarray*}
A straightforward count yields
\begin{eqnarray*}
&&\bigl|S\bigl(G'; \hat{\bm}, \hat{n}_1,
\hat{n}_2\bigr)\bigr|\\
&&\qquad = \pmatrix{n - n_1' -
n_2'
\cr
\hat{n}_1} \cdot\Delta! \cdot
\coeff \bigl[ \bigl(e^z-1\bigr)^{n_1'}\bigl(e^z
\bigr)^{n_2'} ; z^{\Delta- \hat{n}_1} \bigr] \cdot\bI\bigl[
\hat{n}_2 = n_1'+n_2'
\bigr],
\end{eqnarray*}
where $\Delta\equiv\sum_{l=1}^k (\hat{m}_l - m_l')l$. Thus,
$\prob[J_{t+1}=G']$ depends on $G'$ only through $(\bm', n_1', n_2')$.
\end{pf*}

To simplify the proof of Lemma~\ref{lemma:peeling_fast_on_trees}, we first
prove a simple technical lemma.

%
\begin{lemma}\label{lemma:tree_leaves_collapsed}
Let $G = (F,V,E)$ be a factor graph that is a tree with no check node
of degree $1$ or $2$,
rooted at a variable node $v$, with $|V|>1$.
Then $|\{u\in V\dvtx\deg(u) \leq1, u \neq v\}| \geq|V|/2$, that is, at
least half of all variable nodes are leaves.
(Here, a leaf is defined as a variable node that is distinct from the
root and has degree at most $1$.)
\end{lemma}

\begin{pf}
We proceed by induction on the maximum depth $t$ of the tree $G$ rooted
at $v$.
\begin{itemize}
\item\textit{Induction base}:
For a tree of depth 1, let $c = \deg(v) > 0$.
Since all check nodes have degree $3$ or more, $G$ has $\Nl\geq2c$
leaves and $|V| = \Nl+ 1$. Clearly, $\Nl\geq|V|/2$.
\item\textit{Inductive step}:
Consider $G$ having depth $t+1$ and perform $1$
round of synchronous peeling, resulting in $\peel(G) = G' =
(F',V',E')$. Let
$\Nl'$ be the number of leaves in $V'$. The inductive hypothesis implies
$|V'| \leq2 \Nl'$, since $G'$ is also a tree. Since, by construction, every
factor node has degree at least $3$ in $G$, every leaf in $G'$ must have
at least $2$ leaves in $G$ as descendants, that is, $2\Nl' \leq\Nl$, where
$\Nl$ is the number of leaves in $G$. Combining these two inequalities yields
\[
|V| = \bigl|V'\bigr| + \Nl\leq2 \Nl' + \Nl\leq2\Nl,
\]
as desired.\quad\qed
\end{itemize}
\noqed\end{pf}

\begin{pf*}{Proof of Lemma~\ref{lemma:peeling_fast_on_trees}}
By Lemma~\ref{lemma:tree_leaves_collapsed}, if $G$ is a tree, at least one-half
of all variable nodes are leaves at every stage of peeling. Thus, $G$ is
peelable and $\TC(G) \leq\lceil\log_2 |V|\rceil$. (After $\lceil
\log_2
|V|\rceil-1$ rounds of peeling, we have $2$ or less variable nodes
remaining, and hence no checks. At most one more round of peeling leads
to annihilation.)

Now suppose $G$ is unicyclic. Each factor in the cycle has degree at
least $3$, hence it has a neighbor outside the cycle and must
eventually get peeled. Breaking ties arbitrarily, let $a$ be the first
factor in the cycle to be peeled, and let $u \in\partial a$ be the
variable node that ``causes'' it to get peeled (clearly $u$ is not in the
cycle). Let $t_u \leq\TC(G)$ be the
peeling round in which $u$ and $a$ are peeled. Consider the subtree
$G_u = (F_u, V_u, E_u)$ rooted at $u$ defined as follows: $G_u$ is the
maximal connected subgraph of $G$ that includes $u$, but not $a$. Using
Lemma~\ref{lemma:tree_leaves_collapsed} on this subtree and reasoning
as above, we have $t_u \leq\lceil\log_2 |V_u|\rceil\leq\lceil
\log_2
|V|\rceil$.

As at least one factor node in the unicycle is peeled in round $t_u$,
we must
have that $J_{t_u}$ is a tree or forest, which by Lemma~\ref
{lemma:tree_leaves_collapsed} can be peeled in at most $\lceil
\log_2
|V|\rceil$ additional iterations, since the number of variable nodes
in the
$J_{t_u}$ is at most $|V|$. Thus, $\TC(G) \leq t_u + \lceil\log_2
|V|\rceil$. Combining these two inequalities yields
\[
\TC(G) \leq t_u + \bigl\lceil\log_2 |V|\bigr\rceil\leq2\bigl\lceil
\log_2 |V|\bigr\rceil.
\]
\upqed\end{pf*}

\begin{pf*}{Proof of Lemma~\ref{lemma:GW_bound}}
The lemma can be derived from known results (see, e.g., \cite
{Branching}), but
we find it easier to provide an independent proof.

We use a generating function approach to prove the bound
%
%
\begin{equation}
\prob \bigl[ Z_T > (\beta\theta)^T \bigr] \leq2 \exp
\bigl(-\const(\beta/2)^{T}\bigr).
\end{equation}
Equation \eqref{eq:GWBound} follows (eventually for a different
constant $C$) via union bound.

Define $f(s) \equiv\E[s^{Z_1}] = \sum_{j=0}^\infty s^{j}
b_j$. By assumption, it is clear that $f(s)$ is finite for $s \in(0,
1/(1-\delta))$. Define $f^{(t)}(s) \equiv\E[s^{Z_t}]$ for $t \geq1$
[so that $f(s) = f^{(1)}(s)$]. It is well known that
%
%
\begin{equation}
\label{eq:leveltau_generatingfn} f^{(t)}(s) = f\bigl(f^{(t-1)}(s)\bigr)
\end{equation}
for $\tau\geq2$. It follows that $f^{(t)}(s)$ is finite for $s \in(0,
1/(1-\delta))$, and all $\tau\geq2$.

By dominated convergence $f$ is differentiable at $0$ with $f'(0) =
\theta$. Hence, there exists $\ve_0>0$ such that, for all $\ve\in
[0,\ve_0]$
%
%
\begin{equation}
f(1+\ve) \leq1 + 2 \theta\ve.
\end{equation}
By applying the recursion (\ref{eq:leveltau_generatingfn}) and the
fact that $f$ is monotone increasing, we obtain, for all $\ve\in
[0,\ve_0]$
obtain
%
%
\begin{equation}
f^{(T)}(1+\ve) \leq1 + (2 \theta)^T\ve.
\end{equation}
In particular setting $\ve= \ve_0/(2 \theta)^T$, we get
$f^{(T)}(1+\ve) \le1+\ve_0\le2$.

Finally, by Markov inequality,
\begin{eqnarray*}
\prob \bigl\{Z_T \geq(\beta\theta)^T \bigr\} &\leq& (1+
\ve)^{-(\beta\theta)^T} f^{(T)}(1+\ve)
\\
&\le&2 \biggl(1-\frac{\ve}{2} \biggr)^{(\beta\theta)^T}\le2 e^{-(\beta
\theta)^T/2},
\end{eqnarray*}
which completes the proof.
\end{pf*}
u
\section{Proof of Technical Lemmas of Section \texorpdfstring{\lowercase{\protect\ref{secperrr}}}{6}}
\label{app:periphery}

\begin{pf*}{Proof of Lemma~\ref{lem:monotonocity_bba}}
We prove this lemma by induction. Let $\Bl^{(t)}$ and $\Bu^{(t)}$ be
the result
of $t$ steps of backbone augmentation on graphs $G_{\mathrm s}$ and $G$
with initial
graphs $\Bl^{(0)}$ and $\Bu^{(0)}$, respectively. By assumption $\Bl^{(0)}
\subseteq\Bu^{(0)}$. Now assume $\Bl^{(t)} \subseteq\Bu^{(t)}$. It
is enough to
show that if $a \in\Bl^{(t+1)} \setminus\Bl^{(t)}$ then $a \in\Bu
^{(t+1)}$.
Since $a \in\Bl^{(t+1)} \setminus\Bl^{(t)}$, we know that $a \in G$
and has at most
one neighbor outside of $\Bl^{(t)} $. By induction assumption $\Bl^{(t)}
\subseteq\Bu^{(t)}$ and, therefore, $a$ has at most one neighbor outside
$\Bl^{(t)}$. Hence, either $a \in\Bu^{(t)}$ or it is added to $\Bu
^{(\infty)}$ at
step $t+1$.
\end{pf*}

\begin{pf*}{Proof of Lemma~\ref{lem:bdd_branch_fact}}
Define $f(x) = 1-\exp\{-k\alpha x^{k-1}\}$.
It follows immediately from the definition of
$\alpha_{\mathrm{d}}(k)$, that for $\alpha> \ad(k)$, we have $Q >
0$ and
$f'(Q) \le1$.
Furthermore, a straightforward calculation yields
%
%
\begin{equation}
f'(Q) = k(k-1)\alpha Q^{k-2} \exp\bigl\{-k\alpha
Q^{k-1}\bigr\}.
\end{equation}
It is therefore sufficient to exclude the case $f'(Q)=1$.
Solving the equations $f(Q) = Q$ and $f'(Q)=1$, we get the following equation
for $Q$:
%
%
\begin{equation}
-(1-Q)\log(1-Q) = \frac{Q}{k-1}, %
\end{equation}
which has a unique solution $Q_*(k)$ due to the concavity of the
left-hand side. We can then solve for $\alpha$ yielding the unique value
$\alpha=\alpha_*(k)$ such that $f(Q)=Q$ and $f'(Q)=1$ admits a
solution. On the other hand, these two equations are satisfied at
$\ad(k)$ by a continuity argument. It follows that $\ad(k)
=\alpha_*(k)$, and hence $f'(Q)<1$ for all $\alpha>\ad(k)$.
\end{pf*}
\end{appendix}

\section*{Acknowledgements}
While this paper was being finished, we became aware that Dimitris Achlioptas
and Michael Molloy concurrently obtained related results on the same
problem. The two papers are independent. Further, they use different
techniques and
establish somewhat different results.

%

%
%






\printaddresses
\end{document}